\documentclass[journal,onecolumn,twoside]{IEEEtran}
\normalsize
\ifCLASSINFOpdf
\else
\fi
\usepackage{setspace}
\usepackage{color}
\usepackage{cite}
\usepackage{amsmath}
\usepackage{amsfonts}
\usepackage{amssymb}
\usepackage{amsthm}
\usepackage{tikz}
\usepackage{cases}
\usepackage{bm}
\usepackage{graphicx}
    \graphicspath{{../}}
    \DeclareGraphicsExtensions{.pdf}
\usepackage[caption=true,font=footnotesize]{subfig}
\usepackage{multirow}
\usepackage{makecell}
\usepackage{mathdots}
\usepackage{booktabs}
\usepackage{url}
\usepackage{bm}
\usepackage{xtab}
\usepackage{pifont}
\usepackage{array}
\usepackage{algorithm}
\usepackage{algorithmic}
\usepackage{enumerate}
\usetikzlibrary{arrows}
\usepackage{caption}
\usepackage{soul}
\usepackage{xcolor}

\allowdisplaybreaks[4]
\usepackage[colorlinks,
            linkcolor=blue,
            anchorcolor=blue,
            citecolor=blue]{hyperref}
\theoremstyle{plain}

\newtheorem{theorem}{Theorem}
\newtheorem{lemma}[theorem]{Lemma}

\newtheorem{corollary}[theorem]{Corollary}
\theoremstyle{definition}

\newtheorem{conjecture}{Conjecture}

\begin{document}

\title{Sequence Reconstruction Problem for Ternary Deletion Channels}

\author{Xiang Wang,~Han Li, and ~Fang-Wei Fu
        % <-this % stops a space
\thanks{X. Wang is with the School of Mathematics, Statistics and Mechanics, Beijing University of Technology, Beijing, 100124, China (e-mail: xwang@bjut.edu.cn). Han Li is with the Chern Institute of Mathematics and LPMC, Nankai University, Tianjin 300071, China (e-mail: hli@mail.nankai.edu.cn).  F.-W. Fu is with the Chern Institute of Mathematics and LPMC, Nankai University, Tianjin 300071, China (e-mail: fwfu@nankai.edu.cn).}}% <-this % stops a space
%\thanks{Manuscript received October 26, 2023; revised December 8, 2023.}}

%\markboth{Journal of \LaTeX\ Class Files,~Vol.~1, No.~2, December~2023}%

%\IEEEpubid{0000--0000~\copyright~2023 IEEE}

\maketitle

\begin{abstract}
The sequence reconstruction problem was proposed by Levenshtein in 2001. In this model, a sequence from a code is transmitted over several channels, and the decoder receives the distinct outputs from each channel. The main problem is to determine the minimum number of channels required to reconstruct the transmitted
sequence. In the combinatorial context, the sequence reconstruction problem is equivalent to finding the value of $N_q(n,d,t)$, defined as the size of the largest intersection of two metric balls of radius $t$, where the distance between their centers is at least $d$ and the sequences are $q$-ary sequences of the length $n$. Levenshtein first discussed this problem in the uncoded sequence setting and determined the value of $N_q(n,1,t)$ for any $n\geqslant t$. Moreover, Gabrys and Yaakobi studied this problem in the context of binary one-deletion-correcting codes and determined the value of $N_2(n,2,t)$ for $t\geqslant 2$. 

In this paper we study this problem for $3$-ary sequences of length $n$ over the deletion channel, where the transmitted sequence belongs to a one-deletion-correcting code and there are $t$ deletions in every channel. Specifically, we determine $N_3(n,2,t)$ for $t\geqslant 2$.
\end{abstract}


\section{Introduction}
\IEEEPARstart{L}{evenshtein} first proposed the sequence reconstruction problem in $2000$ \cite{L1}. In this model, a sequence $\mathbf{x}$ from some code $\mathcal{C}$ is transmitted through multiple noisy channels. A decoder receives the distinct outputs from each channel to recover $\mathbf{x}$. The sequence reconstruction problem was initially motivated by the fields of biology and chemistry, but has recently regained interest due to the emergence of certain new data storage media such as racetrack memories~\cite{Parkin,Chee} and DNA-based storage~\cite{Church,Yazdi,Lenz}.

Let $\mathbb{Z}_q^n$ denote the set of all $q$-ary sequences of length $n$ and let $\rho: \mathbb{Z}_q^n\times \mathbb{Z}_q^n\rightarrow \mathbb{N}$ be a metric on $\mathbb{Z}_q^n$. Let $\mathcal{C} \subseteq \mathbb{Z}_q^n$ be a code with minimum distance $d$ under the metric $\rho$. Suppose a sequence from $\mathcal{C}$ is transmitted, and each channel introduces at most $t$ errors. Levenshtein~\cite{L1} established  that the minimum number of transmission channels must be greater than the largest intersection of two balls of radius $t$ centered at distinct codewords with least distance $d$ apart:
\begin{equation}
N_q(n,d,t)=\max\limits_{\mathbf{x}_1,\mathbf{x}_2\in \mathbb{Z}_q^n,\rho(\mathbf{x}_1,\mathbf{x}_2)\geqslant d }\{|B_t(\mathbf{x}_1)\cap B_t(\mathbf{x}_2)|\},\label{eq1}
\end{equation}
where $B_t(\mathbf{x})=\{\mathbf{y}\in \mathbb{Z}_q^n|\rho(\mathbf{x,y})\leqslant t\}$ is the ball of radius $t$ centered at $\mathbf{x}$. We refer to the problem of finding $N_q(n,d,t)$ as the~\emph{sequence reconstruction problem}.

Levenshtein~\cite{L1} studied the sequence reconstruction problem, as outlined in Eq. (\ref{eq1}), over some channels, including those based on Hamming distance, Johnson graphs, and other distance metrics. Subsequently, the problem was examined in the context of permutations~\cite{K1,Yaakobi,Wang1,Wang2,Wang3}, general error graphs~\cite{L3,L4}, the Grassmann graph~\cite{Yaakobi}, and deletions or insertions~\cite{Sala1,Lan,Sun1,Zhang,Song}.

The sequence reconstruction problem over the deletion channel has attracted some interest.   Levenshtein~\cite{L1,L2} determined  $N_q(n,1,t)$ for  $\mathcal{C}=\mathbb{Z}_q^n$ ($q \geqslant 2$). Furthermore, Gabrys and Yaakobi~\cite{Gabrys} solved the problem for a binary one-deletion-correcting code, finding $N_2(n,2,t)$. Recently, Pham, Goyal, and Kiah~\cite{Pham},\cite{Pham2} provided an asymptotic solution for $N_2(n,d,t)$ with $d \geqslant 2$. That is, $N_2(n,d,t)=\frac{\binom{2d}{d}}{(t-d)!}n^{t-d}-O(n^{t-d-1})$ for $0\leq d\leq t<n$. For the case $d=t$, they proved that $N_2(n,d,t)=\binom{2t}{t}$. 
Furthermore, Zhang et al. \cite{Zhang} studied the sequence reconstruction problem for the binary $3$-deletion channel and characterized pairs of distinct binary sequences $(\mathbf{x},\mathbf{y})$ for which $|D_3(\mathbf{x})\cap D_3(\mathbf{y})|\in \{19,20\}$, given that the distance between $\mathbf{x}$ and $\mathbf{y}$ is at least $3$. Here, $D_3(\mathbf{x})$ and $D_3(\mathbf{y})$ denote the deletion balls of radius $3$ centered at $\mathbf{x}, \mathbf{y}\in \{0,1\}^{n}$, respectively.

In this paper, we study the sequence reconstruction problem for $3$-ary sequences under the deletion metric, assuming the Levenshtein distance between distinct codewords is at least $2$. We aim to determine the minimum number of channel outputs required to correct $t$ deletions in each channel. We define the \emph{deletion ball} of a $q$-ary sequence $\mathbf{x}$ with radius $t$ as the set of all $q$-ary sequences that can be derived from $\mathbf{x}$ by making $t$ deletions. Let $D_q(m,s)$ denote the maximum size of a deletion ball for $q$-ary sequences of length $m$ with $s$ deletions ($q\geqslant 2$). Our main result is showing that for $t\geqslant 2$, if $\frac{3t}{2}\geqslant n\geqslant t$, then $$N_3(n,2,t)=3^{n-t};$$ if $5\geqslant t\geqslant 2$ and $n\geqslant \max\{6,\lfloor \frac{3t}{2} \rfloor+1\}$, then
\begin{align*}
N_3(n,2,t)=M_1(n,t);
\end{align*}
if $t\geqslant 6$ and $3t-1\geqslant n\geqslant \lfloor \frac{3t}{2}\rfloor +1$, then
\begin{align*}
N_3(n,2,t)=\max\{M_0(n,t),M_1(n,t)\};
\end{align*}
if $t\geqslant 6$ and $n\geqslant 3t$, then
$$N_3(n,2,t)=M_1(n,t),$$
where 
\begin{align*}
M_0(n,t)\triangleq& D_3(n-4,t-2)+3D_3(n-5,t-2)+4D_3(n-5,t-3)+3D_3(n-6,t-3)+D_3(n-6,t-4)\\
&+2D_3(n-7,t-3)+2D_3(n-7,t-4)+D_3(n-8,t-4)+D_3(n-12,t-7)-D_3(n-10,t-5),
\end{align*}
and
$$ M_1(n,t)\triangleq D_3(n-4,t-2)+5D_3(n-5,t-2)+4D_3(n-5,t-3)+3D_3(n-6,t-3)+D_3(n-6,t-4)+D_3(n-8,t-5).$$
Specifically, when $t$ is small, for example, $t=2, 3, 4,$ or $5$, we have
\begin{align*}
N_3(n,2,t)=M_1(n,t)=
\begin{cases}
 6,&~\text{if~} t=2,\\
6n-22,&~\text{if~} t=3,\\
3n^2-25n+48,&~\text{if~} t=4,\\
n^3-14n^2+55n-30,&~\text{if~} t=5,
\end{cases}
\end{align*}
for $n\geqslant \max\{6,\lfloor \frac{3t}{2} \rfloor+1\}$. Furthermore, when $t=2$ or $3$, we have $N_3(4,2,2)=4$, $N_3(5,2,2)=6$, and $N_3(5,2,3)=8$; when $t=6$, we have $N_3(10,2,6)=M_0(10,6)=74$ since  $M_1(10,6)=73$, and $N_3(n,2,6)=M_1(n,6)=\frac{1}{12}(3n^4-62n^3+381n^2-202n-3204)$ for $n\geqslant 11$.

The structure of the remainder of the paper is as follows. In Section~\ref{sec2}, we introduce our notation and establish some foundational results. In Section~\ref{sec3}, we demonstrate that $M_0(n,t)$ and $M_1(n,t)$ are lower bounds; that is, $N_3(n,2,t)\geq \max\{M_0(n,t),M_1(n,t)\}$. In Section~\ref{sec4}, we show that $\max\{M_0(n,t),M_1(n,t)\}$ is also an upper bound, thereby proving the equality. Section~\ref{sec5} concludes the paper.

\section{Definitions and Preliminaries}
\label{sec2}

Let $\mathbb{Z}_q$ denote the set $\{0,1,...,q-1\}$ and $[n]$ denote the set $\{1,2,...,n-1,n\}$. Let $\mathbb{Z}_q^n$ be the set of all length-$n$ $q$-ary sequences. We use lower-case letters to denote scalars, bold lowercase letters to denote vectors or sequences. For $\mathbf{x}=(x_1,x_2,...,x_n)\in \mathbb{Z}_q^n$, let $\mathbf{x}_{[i,j]}$ denote the projection of $\mathbf{x}$ onto the index interval $[i,j]=\{i,i+1,\ldots,j\}$, i.e., $\mathbf{x}_{[i,j]}=(x_i,\ldots,x_j)$ where $i\leqslant j$ are integers. Let $|\mathbf{x}|$ be the length of $\mathbf{x}$. A \emph{run} of $\mathbf{x}$ is a maximal interval which consists of the same symbol.

For two $q$-ary sequences $\mathbf{u}=(u_1,...,u_m),\mathbf{v}=(v_1,...,v_n)$, we denote $\mathbf{u}\circ\mathbf{v}=(u_1,...,u_m,v_1,...,v_n)$ as the concatenation of $\mathbf{u}$ and $\mathbf{v}$. Let $\mathcal{C} \subseteq \mathbb{Z}_q^n$ be a set and $\mathbf{u}$ be a sequence of length at most $n$. We denote by $\mathcal{C}^{\mathbf{u}}$ the subset of $\mathcal{C}$ consisting of all sequences that begin with the sequence $\mathbf{u}$. Furthermore, we define $\mathcal{C}_{\mathbf{u_2}}^{\mathbf{u_1}}$ as the collection of sequences within $\mathcal{C}$ that begin with $\mathbf{u_1}$ and end with $\mathbf{u_2}$. For a sequence $\mathbf{u} = (u_1, \ldots, u_m)$ and a set $\mathcal{C} \subseteq \mathbb{Z}_q^n$, the set resulting from the concatenation of $\mathbf{u}$ with every sequence in $\mathcal{C}$, denoted $\mathbf{u} \circ \mathcal{C}$, is given by
\begin{equation}
\mathbf{u} \circ \mathcal{C} = \{\mathbf{u} \circ \mathbf{c} = (u_1, \ldots, u_m, c_1, \ldots, c_n) | \mathbf{c} = (c_1, \ldots, c_n) \in \mathcal{C}\}.\nonumber
\end{equation}

A sequence $\mathbf{y} \in \mathbb{Z}_q^{n-t}$ is a \emph{$t$-subsequence} of $\mathbf{x} \in \mathbb{Z}_q^n$ if $\mathbf{y}$ is obtained by deleting $t$ symbols from $\mathbf{x}$. The \emph{deletion ball} of radius $t$ centered at $\mathbf{x} \in \mathbb{Z}_q^n$ is defined as the collection of $t$-subsequences of $\mathbf{x}$, denoted by $D_t(\mathbf{x})$. Formally, for any $\mathbf{x}\in \mathbb{Z}_q^n$, we have
\begin{equation}
D_t(\mathbf{x}) = \{\mathbf{y}\in \mathbb{Z}_q^{n-t} | \mathbf{y} \text{ is a $t$-subsequence of } \mathbf{x}\},\nonumber
\end{equation}
with $D_t(\mathbf{x}) = \emptyset$ for $t < 0$ or $t > n$. For $\mathbf{x}\in \mathbb{Z}_q^{n+k}$ and $\mathbf{y}\in \mathbb{Z}_q^{n}$, define $D_q(\mathbf{x},\mathbf{y};t+k,t)=D_{t+k}(\mathbf{x})\cap D_t(\mathbf{y})$. The \emph{Levenshtein distance} between $\mathbf{x}$ and $\mathbf{y}$ is:
\begin{equation}
d_L(\mathbf{x},\mathbf{y})=\min\{s\geq 0| D_q(\mathbf{x},\mathbf{y};s+k,s) \neq \emptyset\}.\nonumber
\end{equation}
A code $\mathcal{C} \subseteq \mathbb{Z}_q^n$ has \emph{minimum Levenshtein distance} $t$ if $d_L(\mathbf{x},\mathbf{y}) \geq t$ for all distinct $\mathbf{x},\mathbf{y} \in \mathcal{C}$, and is a \emph{$(t-1)$-deletion-correcting code} when its minimum Levenshtein distance is at least $t$.

Define $N_q(n,d,t)$ as the maximum size of the intersection of two deletion balls of radius $t$ centered at sequences $\mathbf{x},\mathbf{y} \in \mathbb{Z}_q^n$ with $d_L(\mathbf{x},\mathbf{y}) \geq d$:
\begin{equation}
N_q(n,d,t)=\max
\{|D_q(\mathbf{x},\mathbf{y};t,t)|:\mathbf{x},\mathbf{y}\in \mathbb{Z}_q^n, d_L(\mathbf{x},\mathbf{y})\geqslant d\}.\nonumber
\end{equation}
Its generalization for different sequence lengths is:
\begin{equation}
N_q(n,t+k,t,d)=\max
\{|D_q(\mathbf{x},\mathbf{y};t+k,t)|:\mathbf{x} \in \mathbb{Z}_q^{n+k},\mathbf{y}\in \mathbb{Z}_q^n, d_L(\mathbf{x},\mathbf{y})\geqslant d\}.\nonumber
\end{equation}
Consider transmitting a sequence from a $(d-1)$-deletion-correcting code $\mathcal{C} \subseteq \mathbb{Z}_q^n$ over $N$ distinct channels, each introducing exactly $t$ deletions, with distinct outputs. Levenshtein~\cite{L1} proved that the minimum number of channels required for guaranteed reconstruction is $N_q(n,d,t)+1$.

Let $\sigma=(\sigma_1,...,\sigma_q)$ be an ordering of $\mathbb{Z}_q$ in some order, and define the periodic sequence $\mathbf{c}_q(n,\sigma)=(c_1,c_2,...,c_n)$ where $c_i=\sigma_i$ for $1\leqslant i \leqslant q$ and $c_i=c_{i-q}$ for $i>q$.  When $\sigma_i=i-1$ for all $i\in [q]$, denote $\mathbf{c}_q(n,I_q)$ as $\mathbf{a}_n$, where $I_q=(0,1,\ldots,q-1)$. Let $\mathbf{c}_q(n)=\{\mathbf{c}_q(n,\sigma)|\sigma=(\sigma_1,...,\sigma_q)~\text{is an ordering of all symbols in $\mathbb{Z}_q$}\}$. For integers $0 < t < n$, denote $D_q(n,t)$ by the maximum cardinality of a deletion ball of radius $t$ in $\mathbb{Z}_q^n$. That is,  $D_q(n,t) = \max\{|D_t(\mathbf{x})| : \mathbf{x} \in \mathbb{Z}_q^n\}$. It is known from \cite{Hirschberg} that
\begin{equation}
D_q(n,t)=|D_t(\mathbf{x})|=\sum\limits_{i=0}^{t}\binom{n-t}{i}D_{q-1}(t,t-i),\label{eq2}
\end{equation}
where $\mathbf{x}\in \mathbf{c}_q(n)$. Specifically, $D_2(n,t)=\sum\limits_{i=0}^{t}\binom{n-t}{i}$ and $D_3(n,t)=\sum\limits_{i=0}^{t}\binom{n-t}{i}\sum\limits_{j=0}^{t-i}\binom{i}{j}$. Moreover, the function $D_2(n,t)$ satisfies the recurrence relation
\begin{equation}
D_2(n,t)=D_2(n-1,t)+D_2(n-2,t-1),\nonumber
\end{equation}
and the function $D_3(n,t)$ satisfies the recurrence relation
\begin{equation}
D_3(n,t)=D_3(n-1,t)+D_3(n-2,t-1)+D_3(n-3,t-2),\label{eq3}
\end{equation}
where $n\geqslant t+1$. Moreover, if $n\geqslant t\geqslant \frac{2n}{3}$, then
\begin{equation}
D_3(n,t)=3^{n-t}.\label{eq4}
\end{equation}
For convenience, we define $D_3(n,t)=0$ in cases where $t>n$, $n<0$, or $t<0$, and define $D_3(n,n)=1$ in case where $n\geqslant 1$.

Some conclusions concerning $N_q(n,d,t)$ have been established, as outlined below. When $d=1$, Levenshtein \cite{L1} has provided the following result for $N_q(n,1,t)$.
\begin{theorem}[Levenshtein \cite{L1}]
For any $n\geqslant t+1$ and $q\geqslant 2$, 
\begin{equation}
N_q(n,1,t)=\sum\limits_{i=1}^{q-1}D_q(n-i-1,t-i)+D_q(n-2,t-1)=D_q(n,t)-D_q(n-1,t)+D_q(n-2,t-1).\nonumber
\end{equation}
In particular, if $q=3$, then
\begin{equation}
N_3(n,1,t)=\sum\limits_{i=1}^{2}D_3(n-i-1,t-i)+D_3(n-2,t-1)=2D_3(n-2,t-1)+D_3(n-3,t-2).\label{eq5}
\end{equation}
For $\mathbf{x} = (c_1,c_2,c_3,\ldots,c_n) \in \mathbf{c}_q(n)$ and $\mathbf{y} = (c_2,c_1,c_3,\ldots,c_n)$, $d_L(\mathbf{x},\mathbf{y}) = 1$ and $|D_t(\mathbf{x}) \cap D_t(\mathbf{y})| = N_q(n,1,t)$.
\label{thm1}
\end{theorem}

For the case where $d=2$, Gabrys and Yaakobi \cite{Gabrys} have determined $N_2(n,2,t)$.
\begin{theorem}[Gabrys and Yaakobi \cite{Gabrys}]
For $t \geqslant 2$ and $n \geqslant \max\{8,2t+1\}$,
\begin{align}
N_2(n,2,t)=2D_2(n-4,t-2)+2D_2(n-5,t-2)+2D_2(n-7,t-2)+D_2(n-6,t-3)+D_2(n-7,t-3).\nonumber
\end{align}
\label{thm2}
\end{theorem}

Pham, Goyal, and Kiah \cite{Pham} have found asymptotically exact estimations for $N_2(n,d,t)$ for $0 \leqslant d \leqslant t$.
\begin{theorem}[Pham, Goyal, and Kiah \cite{Pham}]
For $n>t\geqslant d \geqslant 0$,
\begin{align}
N_2(n,d,t)=\frac{\binom{2d}{d}}{(t-d)!}n^{t-d}-O(n^{t-d-1}).\nonumber
\end{align}
Furthermore, when $t=d\geqslant 1$ and $n \geqslant 4t-2$, we have
\begin{equation}
N_2(n,d,d)=\binom{2d}{d}.\nonumber
\end{equation}
\label{thm3}
\end{theorem}

Building on the aforementioned definitions and preliminaries, we aim to determine  $N_3(n,2,t)$ for $t\geqslant 2$ in the following theorem. This paper mainly discusses the combinatorial problem posed in Eq. (\ref{eq1}) for $d=2$ and $q=3$. Our primary contribution is demonstrating that $N_3(n,2,t)$ equals $M_1(n,t)$ if $5\geqslant t\geqslant 2$ and $n\geqslant\max\{6,\lfloor \frac{3t}{2} \rfloor+1\}$, or $t\geqslant 6$ and $n\geqslant 3t$, equals $\max\{M_0(n,t),M_1(n,t)\}$ if $t\geqslant 6$ and
$3t-1\geqslant n\geqslant \lfloor \frac{3t}{2} \rfloor+1$,
and $N_3(n,2,t)=3^{n-t}$ when $\frac{3t}{2} \geqslant n\geqslant t$ and $t\geqslant2$. It is easily verified that if $\frac{3t}{2}\geqslant n\geqslant t$ and $t\geqslant 2$, then $N_3(n,2,t)=3^{n-t}$. To see this, consider the sequences $\mathbf{x}=(0,1,2,0,1,2,\mathbf{a}_{n-6})$ and $\mathbf{y}=(1,0,2,1,0,2,\mathbf{a}_{n-6})$ for $n\geqslant 6$. The Levenshtein distance between $\mathbf{x}$ and 
$\mathbf{y}$ is $2$, as evidenced by the non-empty intersection of their deletion balls of radius $2$ over the first $6$ symbols:
$$D_2(0,1,2,0,1,2)\cap D_2(1,0,2,1,0,2)=\{(0,2,0,2),(0,2,1,2),(1,2,0,2),(1,2,1,2)\}.$$
Moreover, for $\frac{3t}{2} \geqslant n\geqslant t$, $D_t(\mathbf{x})=D_t(\mathbf{y})=\mathbb{Z}_3^{n-t}$ since both of these sequences are concatenations of at least $\lfloor n/3 \rfloor$ triples (e.g., $(0,1,2)$ or $(1,0,2)$), and $n-t\leqslant \lfloor n/3 \rfloor$. Furthermore, when $\max\{\frac{3t}{2},6\} \geqslant n\geqslant t$ and $t\geqslant2$, we can verify that $N_3(n,2,t)=3^{n-t}$ by using a computerized search. 

\begin{theorem}
For $t\geqslant 2$, if $\frac{3t}{2} \geqslant n\geqslant t$, then $$N_3(n,2,t)=3^{n-t};$$
if $5\geqslant t\geqslant 2$ and $ n\geqslant \max\{6,\lfloor \frac{3t}{2} \rfloor+1\}$, then
\begin{align*}
N_3(n,2,t)=M_1(n,t);
\end{align*}
if $t\geqslant 6$ and $3t-1\geqslant n\geqslant \lfloor \frac{3t}{2} \rfloor+1$, then
\begin{align*}
N_3(n,2,t)=\max\{M_0(n,t),M_1(n,t)\};
\end{align*}
if $t\geqslant 6$ and $n\geqslant 3t$, then
\begin{align*}
N_3(n,2,t)=M_1(n,t).
\end{align*}
\label{thm4}
\end{theorem}

\section{Lower Bounds on $N_3(n,2,t)$}
\label{sec3}
The primary objective of this section is to establish the value of $M_i(n,t)$ as a lower bound for $N_3(n,2,t)$ for $i\in\{0,1\}$, where $t \geq 2$. First, when $t\geqslant2$ and $9\geqslant n \geqslant \max\{6,\lfloor \frac{3t}{2}\rfloor+1\}$, we can verify that $N_3(n,2,t)=M_1(n,t)$ by using a computerized search. In the following section, restrict our attention to the case where $n \geqslant \max\{9,\lfloor \frac{3t}{2}\rfloor+1\}$. We demonstrate that for any  $n \geqslant \max\{9,\lfloor \frac{3t}{2}\rfloor+1\}$, the sequences $\mathbf{x}_0=(0,1,2,\mathbf{a}_{n-5},c,b)$ and $\mathbf{y}_0=(1,0,2,\mathbf{a}_{n-3})=(1,0,2,\mathbf{a}_{n-5},b,c)$ have a Levenshtein distance of two. Here, the last two elements of $\mathbf{a}_{n-3}$ is $b$ and $c$. Additionally, the sequences $\mathbf{x}=\mathbf{a}_n=(0,1,2,0,1,2,\mathbf{a}_{n-6})$ and $\mathbf{y}=(1,0,2,1,0,2,\mathbf{a}_{n-6})$ also have a Levenshtein distance of two. We will prove the following equalities: 
\begin{align}
|D_3(\mathbf{x}_0,\mathbf{y}_0;t,t)|=M_0(n,t),\nonumber\\
|D_3(\mathbf{x},\mathbf{y};t,t)|=M_1(n,t).\nonumber
\end{align}
To compute the values of $|D_3(\mathbf{x}_0,\mathbf{y}_0;t,t)|$ and $|D_3(\mathbf{x},\mathbf{y};t,t)|$, and to ensure that $d_L(\mathbf{x}_0,\mathbf{y}_0) \geqslant 2$ and $d_L(\mathbf{x},\mathbf{y}) \geqslant 2$, we will use the following lemmas. Since the cardinality of a set can be computed by splitting it into mutually disjoint sets according to the prefixes and suffixes of its sequences, Lemmas $1$ and $2$ in \cite{Gabrys} also apply to general $q$-ary sequences for $q\geqslant 2$.

\begin{lemma}
Let $n,m_1,m_2,q$ be positive integers such that $m_1+m_2\leq n$. For any $\mathbf{u}\in \mathbb{Z}_q^n$ and $q\geqslant2$,
\begin{equation}
|D_t(\mathbf{u})|=\sum\limits_{\mathbf{u}_1\in \mathbb{Z}_q^{m_1}}\sum\limits_{\mathbf{u}_2\in \mathbb{Z}_q^{m_2}}|D_t(\mathbf{u})_{\mathbf{u}_2}^{\mathbf{u}_1}|.\nonumber
\end{equation}
\label{lm1}
\end{lemma}

\begin{lemma}
Let $n,m_1,m_2,t,q$ be positive integers such that $m_1+m_2\leq n$, $\mathbf{u}=(u_1,...,u_n)\in \mathbb{Z}_q^n, \mathbf{u}_1\in \mathbb{Z}_q^{m_1}, \mathbf{u}_2\in \mathbb{Z}_q^{m_2}$, $q\geqslant 2$.
Assume that $k_1$ is the smallest integer such that $\mathbf{u}_1$ is a subsequence of $(u_1,...,u_{k_1})$ and $k_2$ is the largest integer where $\mathbf{u}_2$ is a subsequence of $(u_{k_2},...,u_n)$. If $k_1<k_2$, then
\begin{equation}
D_t(\mathbf{u})_{\mathbf{u}_2}^{\mathbf{u}_1}=\mathbf{u}_1\circ D_{t^{*}}(u_{k_1+1},...,u_{k_2-1})\circ \mathbf{u}_2,\nonumber
\end{equation}
where $t^*=t-(k_1-m_1)-(n-k_2+1-m_2)$. Consequently, $|D_t(\mathbf{u})_{\mathbf{u}_2}^{\mathbf{u}_1}|=|D_{t^{*}}(u_{k_1+1},...,u_{k_2-1})|$.
\label{lm2}
\end{lemma}

We first establish that $d_L(\mathbf{x}_0,\mathbf{y}_0)=2$ and $d_L(\mathbf{x},\mathbf{y})=2$.

\begin{lemma}
Let $n\geq 6$. Define sequences $\mathbf{x}_0=(0,1,2,\mathbf{a}_{n-5},c,b)$, $\mathbf{y}_0=(1,0,2,\mathbf{a}_{n-3})=(1,0,2,\mathbf{a}_{n-5},b,c)$ (where the last two elements of $\mathbf{a}_{n-3}$ is $b$ and $c$), $\mathbf{x}=(0,1,2,0,1,2,\mathbf{a}_{n-6})$, and  $\mathbf{y}=(1,0,2,1,0,2,\mathbf{a}_{n-6})$. Then
$d_L(\mathbf{x}_0,\mathbf{y}_0)=2$ and 
$d_L(\mathbf{x},\mathbf{y})=2$.
\label{lm3}
\end{lemma}
\begin{proof}
First, we prove $d_L(\mathbf{x},\mathbf{y})=2.$ Let $\mathbf{x}=(x_1,x_2,...,x_n)$ and $\mathbf{y}=(y_1,y_2,...,y_n)$. The subsequence $(0,2,0,2,\mathbf{a}_{n-6})$ is obtained by deleting the second and fourth elements from $\mathbf{x}$, or by deleting the first and third elements from $\mathbf{y}$. Thus, $d_L(\mathbf{x},\mathbf{y})\leqslant 2$.

Let $\mathcal{X}=D_1(\mathbf{x})\cap D_1(\mathbf{y})$. By Lemmas \ref{lm1} and \ref{lm2}, we have $\mathcal{X}=\bigcup\limits_{i=0}^{2}\mathcal{X}^{i}$, and can prove that  $\mathcal{X}=\emptyset$ because of $\mathcal{X}^{0}=\mathcal{X}^{1}=\mathcal{X}^{2}=\emptyset$. Thus, $d_L(\mathbf{x},\mathbf{y})=2$. Similarly, $d_L(\mathbf{x}_0,\mathbf{y}_0)=2.$

%To show $d_L(\mathbf{x},\mathbf{y})>1$, we prove $D_3(\mathbf{x},\mathbf{y};1,1) = \emptyset$. Let $\mathcal{X}=D_t(\mathbf{x})\cap D_t(\mathbf{y})$. By Lemmas \ref{lm1} and \ref{lm2}, we have $\mathcal{X}=\bigcup\limits_{i=0}^{2}\mathcal{X}^{i}$,  where
%$$\mathcal{X}^{0}=0\circ \big(D_1(1,2,0,1,2,\mathbf{a}_{n-6})\cap D_0(2,1,0,2,\mathbf{a}_{n-6})\big), $$
%$$\mathcal{X}^{1}=1\circ \big(D_0(2,0,1,2,\mathbf{a}_{n-6})\cap D_1(0,2,1,0,2,\mathbf{a}_{n-6})\big), $$
%$$\mathcal{X}^{2}=2\circ \big(D_{-1}(0,1,2,\mathbf{a}_{n-6})\cap D_{-1}(1,0,2,\mathbf{a}_{n-6})\big). $$
%Since $\mathcal{X}^{0}=\mathcal{X}^{1}=\mathcal{X}^{2}=\emptyset$, we have $\mathcal{X}=\emptyset$. Thus, $d_L(\mathbf{x},\mathbf{y})=2$. Similarly, $d_L(\mathbf{x}_0,\mathbf{y}_0)=2.$
\end{proof}

\begin{lemma}
For $n\geqslant 4$ and $n\geqslant t+1$,
\begin{equation}
|D_t(0,2,\mathbf{a}_{n-2})|=D_3(n-1,t)+D_3(n-4,t-3)+D_3(n-2,t-1).\nonumber
\end{equation}
\label{lm4}
\end{lemma}
\begin{proof}
Let $\mathbf{x}=(0,2,\mathbf{a}_{n-2})$. By lemmas \ref{lm1} and \ref{lm2},  $|D_t(\mathbf{x})|=|D_t(\mathbf{x})^{0}|+|D_t(\mathbf{x})^{1}|+|D_t(\mathbf{x})^{2}|$, where
$$|D_t(\mathbf{x})^{0}|=|0\circ D_{t}(2,\mathbf{a}_{n-2})|=D_3(n-1,t), $$
$$|D_t(\mathbf{x})^{1}|=|1\circ D_{t-3}(2,\mathbf{a}_{n-4})|=D_3(n-4,t-3),$$
$$|D_t(\mathbf{x})^{2}|=|2\circ D_{t-1}(\mathbf{a}_{n-2})|=D_3(n-2,t-1).$$
Thus, $|D_t(0,2,\mathbf{a}_{n-2})|=D_3(n-1,t)+D_3(n-2,t-1)+D_3(n-4,t-3)$ for $n\geqslant t+1$ and $n\geqslant 4$.
\end{proof}

\begin{lemma}
Let $n\geqslant 9$, $n\geqslant t+4$, and let $a, b, c$ be the last three elements of $\mathbf{a}_{n-6}$. Then
\begin{align*}
|D&_{t-2}(1,2,\mathbf{a}_{n-9},a,c,b)\cap D_{t}(2,0,1,2,\mathbf{a}_{n-9},a,b,c)|=\\
&=D_3(n-6,t-3)+2D_3(n-7,t-3)+D_3(n-7,t-4)+D_3(n-8,t-4)+D_3(n-12,t-7)-D_3(n-10,t-5).
\end{align*}
\label{lm5}
\end{lemma}
\begin{proof}
Let $\mathcal{X}=D_{t-2}(1,2,\mathbf{a}_{n-9},a,c,b)\cap D_{t}(2,0,1,2,\mathbf{a}_{n-9},a,b,c)$.  First, consider $n \geqslant 12$ and $n \geqslant t+5$. By Lemmas \ref{lm1} and \ref{lm2}, $|\mathcal{X}|=|\mathcal{X}^{0}|+|\mathcal{X}^{1}|+|\mathcal{X}^{2}|$,
where
\begin{align*}
|\mathcal{X}^0|=&|D_{t-4}(1,2,\mathbf{a}_{n-12},a,c,b)\cap D_{t-1}(1,2,\mathbf{a}_{n-12},a,b,c,a,b,c)|=|D_{t-4}(1,2,\mathbf{a}_{n-12},a,c,b)|,\\
|\mathcal{X}^1|=&|D_{t-2}(2,\mathbf{a}_{n-9},a,c,b)\cap D_{t-2}(2,\mathbf{a}_{n-9},a,b,c)|,\\
|\mathcal{X}^2|=&|D_{t-3}(\mathbf{a}_{n-9},a,c,b)\cap D_{t}(\mathbf{a}_{n-9},a,b,c,a,b,c)|=|D_{t-3}(\mathbf{a}_{n-9},a,c,b)|.
\end{align*}
Decomposing further: 
\begin{align*}
|\mathcal{X}^0|&=|\mathcal{X}_a^0|+|\mathcal{X}_b^0|+|\mathcal{X}_c^0|\\
&=|D_{t-6}(1,2,\mathbf{a}_{n-12})|+|D_{t-4}(1,2,\mathbf{a}_{n-12},a,c)|+|D_{t-5}(1,2,\mathbf{a}_{n-12},a)|\\
&\overset{(a)}{=}D_3(n-10,t-6)+D_3(n-9,t-4)+D_3(n-10,t-5)+D_3(n-12,t-7)+D_3(n-9,t-5)\\
&\overset{(b)}{=}D_3(n-7,t-4)+D_3(n-12,t-7)-D_3(n-11,t-6),
\end{align*}
where $(a)$ follows from Lemma $\ref{lm4}$, and $(b)$ follows from Eq. $(\ref{eq3})$ since $n\geqslant t+5$. Moreover,
\begin{align*}
|\mathcal{X}^1|=N_3(n-5,1,t-2)=2D_3(n-7,t-3)+D_3(n-8,t-4),    
\end{align*}
from Eq. $(\ref{eq5})$ under the condition $n\geqslant t+4$.
Similarly, using the decomposition method for $\mathcal{X}^0$, for $n\geqslant t+5$, we have
\begin{align*}
|\mathcal{X}^2|=D_3(n-6,t-3)+D_3(n-11,t-6)-D_3(n-10,t-5).
\end{align*}
%\begin{align*}
%|\mathcal{X}^2|&=|\mathcal{X}_a^2|+|\mathcal{X}_b^2|+|\mathcal{X}_c^2|\\
%&=|D_{t-5}(\mathbf{a}_{n-9})|+|D_{t-3}(\mathbf{a}_{n-9},a,c)|+|D_{t-4}(\mathbf{a}_{n-9},a)|\\
%&\overset{(a)}{=}D_3(n-9,t-5)+D_3(n-8,t-3)+D_3(n-9,t-4)+D_3(n-11,t-6)+D_3(n-8,t-4)\\
%&\overset{(b)}{=}D_3(n-6,t-3)+D_3(n-11,t-6)-D_3(n-10,t-5),
%\end{align*}
%where $(a)$ follows from Lemma $\ref{lm4}$, and $(b)$ follows from Eq. $(\ref{eq3})$ since $n\geqslant t+5$.

When $n\geqslant 12$ and $n=t+4$, then $|D_{t-4}(1,2,\mathbf{a}_{n-12},a,c)|=1=D_3(n-8,t-4)$ and $|D_{t-3}(\mathbf{a}_{n-9},a,c)|=1=D_3(n-7,t-3)$. When $n=9, 10,$ or $11$ and $n\geqslant t+4$, it can be verified that $|\mathcal{X}^0|=D_3(n-7,t-4)-D_3(n-11,t-6)$ and $|D_{t-3}(\mathbf{a}_{n-9},a,c,b)|=D_3(n-6,t-3)+D_3(n-11,t-6)-D_3(n-10,t-5)$. Hence, cases $n=t+4$ and $n\geqslant 9$ yield identical expressions after boundary verification.

Thus, when $n\geqslant 9$ and $n\geqslant t+4$, 
\begin{align*}
|\mathcal{X}|=D_3(n-6,t-3)+2D_3(n-7,t-3)+D_3(n-7,t-4)+D_3(n-8,t-4)+D_3(n-12,t-7)-D_3(n-10,t-5).
\end{align*}
\end{proof}

By Lemmas $\ref{lm1}$-$\ref{lm5}$, we have the following theorem which provides  $|D_t(\mathbf{x}_0)\cap D_t(\mathbf{y}_0)|$ and $|D_t(\mathbf{x})\cap D_t(\mathbf{y})|$ for $t\geq 2$ and $n\geqslant \max\{9,\lfloor\frac{3t}{2}\rfloor+1\}$.

\begin{theorem}
For $t\geqslant 2$ and $n\geqslant \max\{6,\lfloor\frac{3t}{2}\rfloor+1\}$,
\begin{equation}
N_3(n,2,t)\geqslant \max\{M_0(n,t),M_1(n,t)\}.\nonumber
\end{equation}
\label{thm5}
\end{theorem}
\begin{proof}
When $9\geqslant n\geqslant \max\{6,\lfloor\frac{3t}{2}\rfloor+1\}$, we can verify that $N_3(n,2,t)=M_1(n,t)$ by using a computerized search. If $n\geqslant \max\{9,\lfloor\frac{3t}{2}\rfloor+1\}$ and  $t\geqslant 2$, then we have  $n\geqslant t+4$.  Let $\mathbf{x}=(0,1,2,0,1,2,\mathbf{a}_{n-6})$, $\mathbf{y}=(1,0,2,1,0,2,\mathbf{a}_{n-6})$, and $\mathcal{X}=D_t(\mathbf{x})\cap D_t(\mathbf{y})$. By Lemmas $\ref{lm1}$ and $\ref{lm2}$, we have $|\mathcal{X}|=|\mathcal{X}^{0}|+|\mathcal{X}^{1}|+|\mathcal{X}^{2}|$,
where
\begin{align*}
|\mathcal{X}^0|=&|D_{t}(1,2,0,1,2,\mathbf{a}_{n-6})\cap D_{t-1}(2,1,0,2,\mathbf{a}_{n-6})|,\\
|\mathcal{X}^1|=&|D_{t-1}(2,0,1,2,\mathbf{a}_{n-6})\cap D_{t}(0,2,1,0,2,\mathbf{a}_{n-6})|,\\
|\mathcal{X}^2|=&|D_{t-2}(0,1,2,\mathbf{a}_{n-6})\cap D_{t-2}(1,0,2,\mathbf{a}_{n-6})|.
\end{align*}

Decomposing $\mathcal{X}^0$, we have  $|\mathcal{X}^0|=|\mathcal{X}^{(0,0)}|+|\mathcal{X}^{(0,1)}|+|\mathcal{X}^{(0,2)}|$ with
\begin{align*}
|\mathcal{X}^{(0,0)}|=&|D_{t-2}(1,2,\mathbf{a}_{n-6})\cap D_{t-3}(2,\mathbf{a}_{n-6})|=|D_{t-3}(2,\mathbf{a}_{n-6})|\overset{(a)}{=}D_3(n-5,t-3),\\
|\mathcal{X}^{(0,1)}|=&|D_{t}(2,0,1,2,\mathbf{a}_{n-6})\cap D_{t-2}(0,2,\mathbf{a}_{n-6})|=|D_{t-2}(0,2,\mathbf{a}_{n-6})|\overset{(b)}{=}
D_3(n-5,t-2)+D_3(n-6,t-3)+D_3(n-8,t-5),\\
|\mathcal{X}^{(0,2)}|=&|D_{t-1}(0,1,2,\mathbf{a}_{n-6})\cap D_{t-1}(1,0,2,\mathbf{a}_{n-6})|\overset{(c)}{=}N_3(n-3,1,t-1)=2D_3(n-5,t-2)+D_3(n-6,t-3),
\end{align*}
where $(a)$ follows from Eq. $(\ref{eq2})$, $(b)$ follows by Lemma $\ref{lm4}$, $(c)$ follows from Theorem $\ref{thm1}$ since $n\geqslant t+3$. Thus,
\begin{align}
|\mathcal{X}^0|=3D_3(n-5,t-2)+D_3(n-5,t-3)+2D_3(n-6,t-3)+D_3(n-8,t-5).\nonumber
\end{align}

Similarly, using the decomposition method for $\mathcal{X}^0$, it follows that 
%decomposing $\mathcal{X}^{1}$, we have $|\mathcal{X}^1|=|\mathcal{X}^{(1,0)}|+|\mathcal{X}^{(1,1)}|+|\mathcal{X}^{(1,2)}|$ with
%\begin{align*}
%|\mathcal{X}^{(1,0)}|=&|D_{t-2}(1,2,\mathbf{a}_{n-6})\cap D_{t}(2,1,0,2,\mathbf{a}_{n-6})|=|D_{t-2}(1,2,\mathbf{a}_{n-6})|\overset{(a)}{=}D_3(n-4,t-2),\\
%|\mathcal{X}^{(1,1)}|=&|D_{t-3}(2,\mathbf{a}_{n-6})\cap D_{t-2}(0,2,\mathbf{a}_{n-6})|=|D_{t-3}(2,\mathbf{a}_{n-6})|\overset{(b)}{=}
%D_3(n-5,t-3),\\
%|\mathcal{X}^{(1,2)}|=&|D_{t-1}(0,1,2,\mathbf{a}_{n-6})\cap D_{t-1}(1,0,2,\mathbf{a}_{n-6})|\overset{(c)}{=}N_3(n-3,1,t-1)=2D_3(n-5,t-2)+D_3(n-6,t-3),
%\end{align*}
%where $(a),(b)$ follow from Eq. $(\ref{eq2})$, $(c)$ follows from Theorem $\ref{thm1}$ since $n\geqslant t+3$. Hence,
\begin{align}
|\mathcal{X}^1|&=D_3(n-4,t-2)+2D_3(n-5,t-2)+D_3(n-5,t-3)+D_3(n-6,t-3),\nonumber
\end{align}
for $n\geqslant t+3$.

By Theorem $\ref{thm1}$, for $n\geqslant t+3$, we have
\begin{equation}
|\mathcal{X}^2|=|D_{t-2}(0,1,2,\mathbf{a}_{n-6})\cap D_{t-2}(1,0,2,\mathbf{a}_{n-6})|=N_3(n-3,1,t-2)=2D_3(n-5,t-3)+D_3(n-6,t-4).\nonumber
\end{equation}

Thus, for any $t\geqslant 2$, and $n\geqslant \max\{9,\lfloor\frac{3n}{2}\rfloor+1\}$, 
\begin{align}
|\mathcal{X}|&=|\mathcal{X}^{0}|+|\mathcal{X}^{1}|+|\mathcal{X}^{2}|\nonumber\\
&=D_3(n-4,t-2)+5D_3(n-5,t-2)+4D_3(n-5,t-3)+3D_3(n-6,t-3)+D_3(n-6,t-4)+D_3(n-8,t-5)\nonumber\\
&=M_1(n,t).\nonumber
\end{align}

Similarly, let $\mathbf{x}_0=(0,1,2,\mathbf{a}_{n-6},a,c,b)=(0,1,2,0,1,2,\mathbf{a}_{n-9},a,c,b)$,  $\mathbf{y}_0=(1,0,2,\mathbf{a}_{n-6},a,b,c)=(1,0,2,0,1,2,\mathbf{a}_{n-9},\\a,b,c)$, $\mathcal{Y}=D_t(\mathbf{x}_0)\cap D_t(\mathbf{y}_0)$, and let $a, b, c$ be the last three elements of $\mathbf{a}_{n-3}$ such that $\{a,b,c\}=\mathbb{Z}_3$. By Lemmas $\ref{lm1}$ and $\ref{lm2}$, we have $|\mathcal{Y}|=|\mathcal{Y}^{(0,0)}|+
|\mathcal{Y}^{(0,1)}|+|\mathcal{Y}^{(0,2)}|+|\mathcal{Y}^{(1,0)}|+
|\mathcal{Y}^{(1,1)}|+|\mathcal{Y}^{(1,2)}|+|\mathcal{Y}^{2}|$ with
\begin{align*}
|\mathcal{Y}^{(0,0)}|=&|D_{t-2}(1,2,\mathbf{a}_{n-9},a,c,b)\cap D_{t-2}(1,2,\mathbf{a}_{n-9},a,b,c)|=N_3(n-4,1,t-2)=2D_3(n-6,t-3)+D_3(n-7,t-4),\\
|\mathcal{Y}^{(0,1)}|=&|D_{t}(2,\mathbf{a}_{n-6},a,c,b)\cap D_{t-3}(2,\mathbf{a}_{n-9},a,b,c)|=|D_{t-3}(2,\mathbf{a}_{n-6})|=D_3(n-5,t-3),\\
|\mathcal{Y}^{(0,2)}|=&|D_{t-1}(\mathbf{a}_{n-6},a,c,b)\cap D_{t-1}(\mathbf{a}_{n-6},a,b,c)|=N_3(n-3,1,t-1)=2D_3(n-5,t-2)+D_3(n-6,t-3),\\
|\mathcal{Y}^{(1,0)}|\overset{(a)}{=}&|D_{t-2}(1,2,\mathbf{a}_{n-9},a,c,b)\cap D_{t}(2,0,1,2,\mathbf{a}_{n-9},a,b,c)|=D_3(n-6,t-3)+2D_3(n-7,t-3)+D_3(n-7,t-4)\\
&+D_3(n-8,t-4)+D_3(n-12,t-7)-D_3(n-10,t-5),\\
|\mathcal{Y}^{(1,1)}|=&|D_{t-3}(2,\mathbf{a}_{n-9},a,c,b)\cap D_{t-3}(2,\mathbf{a}_{n-9},a,b,c)|=N_3(n-5,1,t-3)=2D_3(n-7,t-4)+D_3(n-8,t-5),\\
|\mathcal{Y}^{(1,2)}|=&|D_{t-1}(\mathbf{a}_{n-6},a,b,c)\cap D_{t-1}(\mathbf{a}_{n-6},a,c,b)|=N_3(n-3,1,t-1)=2D_3(n-5,t-2)+D_3(n-6,t-3),\\
|\mathcal{Y}^2|=&|D_{t-2}(\mathbf{a}_{n-6},a,b,c)\cap D_{t-2}(\mathbf{a}_{n-6},a,c,b)|=N_3(n-3,1,t-2)=2D_3(n-5,t-3)+D_3(n-6,t-4),
\end{align*}
where $(a)$ follows from Lemma $\ref{lm5}$ under the condition $n\geqslant \max\{9,t+4\}$.

Thus, for any $t\geqslant 2$, and $n\geqslant \max\{9,\lfloor\frac{3n}{2}\rfloor+1\}$, 
\begin{align}
|\mathcal{Y}|&=|\mathcal{Y}^{(0,0)}|+|\mathcal{Y}^{(0,1)}|+|\mathcal{Y}^{(0,2)}|+|\mathcal{Y}^{(1,0)}|+
|\mathcal{Y}^{(1,1)}|+|\mathcal{Y}^{(1,2)}|+|\mathcal{Y}^{2}|\nonumber\\
&=D_3(n-4,t-2)+3D_3(n-5,t-2)+4D_3(n-5,t-3)+3D_3(n-6,t-3)+D_3(n-6,t-4)+\nonumber\\
&~~~+2D_3(n-7,t-3)+2D_3(n-7,t-4)+D_3(n-8,t-4)+D_3(n-12,t-7)-D_3(n-10,t-5)\nonumber\\
&=M_0(n,t).\nonumber
\end{align}
So, the theorem follows.
\end{proof}

\section{The upper bound}
\label{sec4}
In this section, we prove that the lower bound for $N_3(n,2,t)$ is also an upper bound for $n \geqslant \max\{9,\lfloor\frac{3t}{2}\rfloor+1\}$ and $t\geqslant2$ since $N_3(n,2,t)=M_1(n,t)$ for  $9\geqslant n\geqslant \max\{6,\lfloor\frac{3t}{2}\rfloor+1\}$ and $t\geqslant 2$. We analyze different cases for $\mathbf{x},\mathbf{y}\in \mathbb{Z}_3^n$ and establish the result for all scenarios. To prove the upper bound, we require the following lemmas and identities for $D_3(n,t)$:
\begin{equation}
D_3(n,t)\leqslant D_3(n+1,t+1),\label{eq6}
\end{equation}
and
\begin{equation}
D_3(n,t)\leqslant D_3(n+1,t).\label{eq7}
\end{equation}

The following corollary is a consequence of  Eqs. $(\ref{eq3})$, $(\ref{eq6})$, and $(\ref{eq7})$.
\begin{corollary}
For $n\geqslant t+1$ and $i\in \{0,1\}$, 
\begin{align}
M_i(n,t)&=M_i(n-1,t)+M_i(n-2,t-1)+M_i(n-3,t-2),\label{eq8}\\
D_3(n,t)&\leqslant 3D_3(n-1,t).\label{eq9}
\end{align}
\label{cor1}
\end{corollary}

Let $\mathbf{x}=(x_1,..,x_{n+k})\in \mathbb{Z}_3^{n+k}$, $\mathbf{y}=(y_1,...,y_n)\in \mathbb{Z}_3^n$, and let $\mathcal{X}=D_{k+t}(\mathbf{x})\cap D_t(\mathbf{y})$ ($n \geqslant t \geqslant 1$, $k \geqslant 0$), Lemmas \ref{lm1} and \ref{lm2} imply that for $a \in \mathbb{Z}_3$,
$$|\mathcal{X}^a|=|a\circ \big(D_{t+k-\ell}(x_{2+\ell},...,x_{n+k})\cap D_{t-\ell^*}(y_{2+\ell^*},...,y_n)\big)|=|D_{t+k-\ell}(x_{2+\ell},...,x_{n+k})\cap D_{t-\ell^*}(y_{2+\ell^*},...,y_n)|,$$
where the subscripts $\ell,\ell^*$ are some non-negative integers such that $x_{1+\ell}$ and $y_{1+\ell^*}$ are the first occurrence of the symbol $a$ in $\mathbf{x}$ and $\mathbf{y}$, respectively. Unless otherwise stated, $\ell$ and $\ell^*$ are defined this way during the decomposition of $\mathcal{X}$, with $\ell, \ell^* \geqslant 0$ for simplicity. First, we give some results of $N_3(n,t+1,t,1)$ for some $n$ and $t$ in Lemmas $\ref{lm6}$ and $\ref{lm7}$ which will be used later. The proofs of these two lemmas
is given in Appendix $\ref{APP-A}$.

\begin{lemma}
\label{lm6}
For $n\geqslant 4$, we have 
$$N_3(n,2,1,1)\leqslant 3.$$ 
\end{lemma}

For convenience, define $N_3(n,t)=D_3(n-2,t-1)+2D_3(n-3,t-1)+D_3(n-3,t-2)+D_3(n-4,t-2)$ for $\frac{2n}{3}>t\geqslant 1$.

\begin{lemma}
\label{lm7}
For $n\geqslant \max\{4,t\}$ and $t\geqslant 1$:
\begin{itemize}
\item If $2n/3>t$, then $N_3(n,t+1,t,1) = N_3(n,t)$.
\item If $2n/3 \leqslant t$, then $N_3(n,t+1,t,1)=3^{n-t}$.
\end{itemize}
\end{lemma}

By Lemma $\ref{lm7}$, we have the following result.
\begin{corollary}
For $n\geqslant \lfloor \frac{3t}{2} \rfloor +1$ and $t\geqslant 1$, $N_3(n-2,t,t-1,1)=N_3(n-2,t-1)$.
\label{cor2}
\end{corollary}
\begin{proof}
When $2(n-2)>3(t-1)$, the result holds by Lemma \ref{lm7}. If $2(n-2) \leqslant 3(t-1)$, then $n \geqslant \lfloor 3t/2 \rfloor + 1$ implies $2n \geqslant 3t+1$. Thus, $2n=3t+1$ (i.e., $t=2k+1$, $n=3k+2$, $k \geqslant 0$). When $k=0$ then $N_3(n-2,t,t-1,1)=N_3(n-2,t-1)=1$. When $k\geqslant 1$ then $3(t-1)=2(n-2)$. By Lemma $\ref{lm7}$, it follows that $N_3(n-2,t,t-1,1)=3^{n-t-1}=3^k=N_3(n-2,t-1)$. 
%Moreover, by Eq. $(\ref{eq4})$, $N_3(n-2,t-1)=N_3(3k,2k)=D_3(3k-2,2k-1)+2D_3(3k-3,2k-1)+D_3(3k-3,2k-2)+D_3(3k-4,2k-2)=3^{k-1}+2\cdot 3^{k-2}+3^{k-1}+3^{k-2}=3^k=N_3(n-2,t,t-1,1)$.
So, the lemma follows.
\end{proof}

Let $\mathbf{x}=(x_1,x_2,...,x_n), \mathbf{y}=(y_1,y_2,...,y_n)\in \mathbb{Z}_3^n$ with $d_L(\mathbf{x},\mathbf{y})\geqslant 2$. For $c\in \mathbb{Z}_3$, we define $c+\mathbf{x}=(c+x_1,c+x_2,...,c+x_n)$, where $c+x_i$ represents the addition operation performed on $c$ and $x_i$ in the ring of integers modulo $3$ (i.e., $c+x_i \pmod{3}$) for any $i\in [n]$. To prove that $|D_3(\mathbf{x},\mathbf{y};t,t)|\leqslant \max\{M_0(n,t),M_1(n,t)\}$ for $x_1\neq y_1$, we need some properties of $|D_3(\cdot)|$ in Lemmas $\ref{lm8}$-$\ref{lm11}$. In order to establish this result, we first state the following three useful lemmas. The proof of Lemma $\ref{lm9}$ appears in Appendix $\ref{APP-B}$. Appendix $\ref{APP-C}$ contains the proofs of Lemmas $\ref{lm10}$ and $\ref{lm11}$.

\begin{lemma}
\label{lm8}
For $\mathbf{x}\in \mathbb{Z}_3^{n}$ and $c\in \{1,2\}$,
\begin{equation}
|D_t(\mathbf{x})|=|D_t(c+\mathbf{x})|.\nonumber
\end{equation}
\end{lemma}
\begin{proof}
Without loss of generality, we only consider $c=1$. We transform the vector $\mathbf{x}$ into the vector $1+\mathbf{x}$ by applying the function $f$ to each component, where $f(0)=1,f(1)=2,f(2)=0$. If $\mathbf{y}\in D_t(\mathbf{x})$, then $1+\mathbf{y}\in D_t(1+\mathbf{x})$. Thus, $|D_t(\mathbf{x})|\leqslant |D_t(1+\mathbf{x})|$. Similarly, $|D_t(1+\mathbf{x})|\leqslant |D_t(\mathbf{x})|$. So, the lemma follows.
\end{proof}

\begin{lemma}
\label{lm9}
For $\mathbf{x}\in \mathbb{Z}_3^{m}$, $\mathbf{y}\in \mathbb{Z}_3^n$, and $a\in \mathbb{Z}_3$, there exists some element $c\in \{1,2\}$ such that
\begin{equation}
|D_t(\mathbf{x},a,a,\mathbf{y})|\leqslant |D_t(\mathbf{x},a,c+a,c+\mathbf{y})|. \nonumber
\end{equation}
\end{lemma}

\begin{lemma}
\label{lm10}
Let $\mathbf{w} = (\mathbf{x}, a, b, a, \mathbf{y}) \in \mathbb{Z}_3^n$ with $a \neq b$, $n\geqslant 4$, $(\mathbf{x}, a, b) \in \mathbf{c}_3(s+2)$, and $\{a,b,c\} = \mathbb{Z}_3$. Define $\mathbf{x}'=(x'_1,...,x'_s)$ by interchanging $a$ and $c$ in $\mathbf{x}$.  That is, 
\begin{equation*}
x_i'=
\begin{cases}
c & \text{if } x_i = a, \\
a & \text{if } x_i = c, \\
b & \text{if } x_i = b,
\end{cases}
\end{equation*}
for all $i\in [s]$. Then
\begin{equation}
|D_t(\mathbf{w})|\leqslant |D_t(\mathbf{x}',c,b,a,\mathbf{y})|,\nonumber
\end{equation}
where $(\mathbf{x}',c,b,a)\in \mathbf{c}_3(s+3)$.
\end{lemma}

\begin{lemma}
\label{lm11}
For $n\geqslant t+2$ and $\mathbf{x}\notin \mathbf{c}_3(n)$,
\begin{equation}
|D_t(\mathbf{x})|\leqslant D_3(n-2,t)+D_3(n-2,t-1)+D_3(n-3,t-1)+D_3(n-3,t-2)+D_3(n-5,t-3).\nonumber
\end{equation}
\end{lemma}

Furthermore,  in order to establish the result that $|D_3(\mathbf{x},\mathbf{y};t,t)|\leqslant \max\{M_0(n,t),M_1(n,t)\}$ for $x_1\neq y_1$, the result holds for some $\mathbf{x}$ and $\mathbf{y}$ in Lemmas $\ref{lm12}$ and $\ref{lm13}$. Appendices $\ref{APP-D}$ and $\ref{APP-E}$ contain the proofs of Lemmas $\ref{lm12}$ and $\ref{lm13}$, respectively.

\begin{lemma}
\label{lm12}
For $n\geqslant \max\{9,\lfloor \frac{3t}{2}\rfloor +1\}$, $t\geqslant 2$, $\mathbb{Z}_3=\{a,b,c\}$, and sequences $\mathbf{x}=(b,a,c,a,b,x_6,...,x_n),\mathbf{y}=(a,b,c,a,b,y_6,...,\\y_n)\in \mathbb{Z}_3^n$ with $d_L(\mathbf{x},\mathbf{y})\geqslant 2$,
\begin{equation}
|D_t(\mathbf{x})\cap D_t(\mathbf{y})|\leqslant M_0(n,t).\nonumber
\end{equation}   
\end{lemma}

\begin{lemma}
\label{lm13}
For $n\geqslant \max\{9,\lfloor \frac{3t}{2}\rfloor +1\}$, $t\geqslant 2$,  $\mathbb{Z}_3=\{a,b,c\}$, and sequences $\mathbf{x}=(b,a,c,b,a,x_6,...,x_n),\mathbf{y}=(a,b,c,a,b,y_6,...,\\y_n)\in \mathbb{Z}_3^n$,
\begin{equation}
|D_t(\mathbf{x})\cap D_t(\mathbf{y})|\leqslant M_1(n,t).\nonumber
\end{equation}
\end{lemma}

By Lemmas $\ref{lm11}$-$\ref{lm13}$, we begin by restricting our attention to sequences $\mathbf{x}=(x_1,...,x_n), \mathbf{y}=(y_1,...,y_n)$ that differ on the first bit. That is, $x_1\neq y_1$. In Lemma $\ref{lm14}$, we will prove that $|D_3(\mathbf{x},\mathbf{y};t,t)|\leqslant \max\{M_0(n,t),M_1(n,t)\}$ for $n\geqslant 9$ and $\frac{2n}{3}>t\geqslant 2$.

\begin{lemma}
\label{lm14}
For $\mathbf{x}=(x_1,x_2,...,x_n), \mathbf{y}=(y_1,y_2,...,y_n)\in \mathbb{Z}_3^n$ with $d_L(\mathbf{x},\mathbf{y})\geqslant 2$ and $x_1\neq y_1$, $$|D_3(\mathbf{x},\mathbf{y};t,t)|\leqslant \max\{M_0(n,t),M_1(n,t)\}$$ 
holds when $n\geqslant \max\{9,\lfloor \frac{3t}{2}\rfloor +1\}$ and $t\geqslant 2$.
\end{lemma}
\begin{proof}
Since $n\geqslant \max\{9,\lfloor \frac{3t}{2}\rfloor +1\}$, we have $n\geqslant t+4$. Let $\mathcal{X}=D_t(\mathbf{x})\cap D_t(\mathbf{y})$. Without loss of generality, let $x_1=1$ and $y_1=0$. By Lemmas $\ref{lm1}$ and $\ref{lm2}$, $|\mathcal{X}|=|\mathcal{X}^0|+|\mathcal{X}^1|+|\mathcal{X}^2|$ with
\begin{align}
\mathcal{X}^0&=0\circ \big(D_{t-1-\ell_0}(\mathbf{x}_{[3+\ell_0,n]})\cap D_t(\mathbf{y}_{[2,n]})\big),\nonumber\\
\mathcal{X}^1&=1\circ \big(D_{t}(\mathbf{x}_{[2,n]})\cap D_{t-1-\ell_1^*}(\mathbf{y}_{[3+\ell_1^*,n]})\big),\nonumber\\
\mathcal{X}^2&=2\circ \big(D_{t-1-\ell_2}(\mathbf{x}_{[3+\ell_2,n]})\cap D_{t-1-\ell_2^*}(\mathbf{y}_{[3+\ell_2^*,n]})\big),\nonumber
\end{align}
where $\ell_0,\ell_1^*,\ell_2,\ell_2^*$ are non-negative integers such that $x_{2+\ell_0}=0$, $x_{2+\ell_2}=2$, $y_{2+\ell_1^{*}}=1$, and $y_{2+\ell_2^{*}}=2$. Moreover, if $\ell_0=0$, then $x_2=0$; if $\ell_1^*=0$, then $y_2=1$; if $\ell_2=0$ or $\ell_2^*=0$, then $x_2=2$ or $y_2=2$, respectively. Specially, based on the values of $\ell_0$ and $\ell_1^*$, we discuss the value of $|\mathcal{X}|$ in the following four cases: \textbf{Case A:} $\ell_0=\ell_1^*=0$; \textbf{Case B:} $\ell_0,\ell_1^*\geqslant 1$; \textbf{Case C:} $\ell_0=0, \ell_1^*\geqslant 1$; \textbf{Case D:} $\ell_0\geqslant 1, \ell_1^*=0$.

\textbf{Case A:} Given that $\ell_0=\ell_1^*=0$, it follows that $x_2=0$, $y_2=1$, and $\ell_2,\ell_2^*\geqslant 1$. In \textbf{Case A}, based on the values of $\ell_2$ and $\ell_2^*$, we consider the following two subcases: \textbf{Case A1:} $\ell_2\geqslant 2$ or $\ell_2^*\geqslant 2$; \textbf{Case A2:} $\ell_2=\ell_2^*=1$.

\textbf{Case A1:} Given that $\ell_2\geqslant 2$ or $\ell_2^*\geqslant 2$, we have 
\begin{align}
|\mathcal{X}^2|\leqslant \min\{|D_{t-1-\ell_2}(\mathbf{x}_{[3+\ell_2,n]})|,|D_{t-1-\ell_2^*}(\mathbf{y}_{[3+\ell_2^*,n]})|\}\leqslant D_3(n-4,t-3),\nonumber
\end{align}
from Eq. $(\ref{eq3})$ under the condition $n\geqslant t+2$. Since $d_L(\mathbf{x,y})\geqslant 2$, $(x_1,x_2)=(1,0)$, and $(y_1,y_2)=(0,1)$, we have $\mathbf{x}_{[3,n]}\notin D_1(\mathbf{y}_{[2,n]})$ and $\mathbf{y}_{[3,n]}\notin D_1(\mathbf{x}_{[2,n]})$. By Lemma $\ref{lm7}$, it follows that
\begin{align}
|\mathcal{X}^0|=|D_{t-1}(\mathbf{x}_{[3,n]})\cap D_{t}(\mathbf{y}_{[2,n]})|&\leqslant N_3(n-2,t,t-1,1)\nonumber\\
&=D_3(n-4,t-2)+2D_3(n-5,t-2)+D_3(n-5,t-3)+D_3(n-6,t-3),\nonumber\\
|\mathcal{X}^1|=|D_{t}(\mathbf{x}_{[2,n]})\cap D_{t-1}(\mathbf{y}_{[3,n]})|&\leqslant N_3(n-2,t,t-1,1)\nonumber\\
&=D_3(n-4,t-2)+2D_3(n-5,t-2)+D_3(n-5,t-3)+D_3(n-6,t-3).\nonumber
\end{align}
Thus, 
\begin{align*}
|\mathcal{X}|&=|\mathcal{X}^{0}|+|\mathcal{X}^{1}|+|\mathcal{X}^2|\\
&\leqslant 2D_3(n-4,t-2)+D_3(n-4,t-3)+4D_3(n-5,t-2)+2D_3(n-5,t-3)+2D_3(n-6,t-3)\\
&\overset{(a)}{=}M_1(n,t)-\big(D_3(n-5,t-3)+D_3(n-8,t-5)-D_3(n-7,t-4)-D_3(n-7,t-5)\big)\\
&\overset{(b)}{=} M_1(n,t)-\big(D_3(n-6,t-3)+D_3(n-7,t-4)+D_3(n-8,t-5)+D_3(n-8,t-5)-D_3(n-7,t-4)\\
&~~~~-D_3(n-8,t-5)-D_3(n-9,t-6)-D_3(n-10,t-7)\big)\\
&=M_1(n,t)-\big(D_3(n-6,t-3)-D_3(n-9,t-6)\big)+\big(D_3(n-8,t-5)-D_3(n-10,t-7)\big)\\
&\overset{(c)}{\leqslant} M_1(n,t),
\end{align*}
where $(a), (b)$ follow from Eq. $(\ref{eq3})$ since $n\geqslant t+3$, $(c)$ follows from Eq. $(\ref{eq6})$.

\textbf{Case A2:} Given that $\ell_2=\ell_2^*=1$, it follows that $x_3=y_3=2$. Thus, $\mathbf{x}_{[4,n]}\neq \mathbf{y}_{[4,n]}$ because $d_L(\mathbf{x},\mathbf{y})\geqslant 2$. Furthermore, for $n\geqslant t+3$,
\begin{align}
|\mathcal{X}^2|=|D_{t-2}(\mathbf{x}_{[4,n]})\cap D_{t-2}(\mathbf{y}_{[4,n]})|\leqslant N_3(n-3,1,t-2)=2D_3(n-5,t-3)+D_3(n-6,t-4).\label{eq10} 
\end{align}
We further decompose $\mathcal{X}^0$ and $\mathcal{X}^1$ with $(x_1,x_2,x_3)=(1,0,2)$ and $(y_1,y_2,y_3)=(0,1,2)$ as follows:
$|\mathcal{X}^0|=|\mathcal{X}^{(0,0)}|+|\mathcal{X}^{(0,1)}|+|\mathcal{X}^{(0,2)}|$ and $|\mathcal{X}^1|=|\mathcal{X}^{(1,0)}|+|\mathcal{X}^{(1,1)}|+|\mathcal{X}^{(1,2)}|$ with
\begin{align}
|\mathcal{X}^{(0,0)}|&=|D_{t-2-\ell_{(0,0)}}(x_{5+\ell_{(0,0)}},...,x_n)\cap D_{t-2-\ell_{(0,0)}^*}(y_{5+\ell_{(0,0)}^*},...,y_n)|,\nonumber\\
|\mathcal{X}^{(0,1)}|&=|D_{t-2-\ell_{(0,1)}}(x_{5+\ell_{(0,1)}},...,x_n)\cap D_{t}(2,y_{4},...,y_n)|,\nonumber\\
|\mathcal{X}^{(0,2)}|&=|D_{t-1}(x_{4},...,x_n)\cap D_{t-1}(y_{4},...,y_n)|,\nonumber\\
|\mathcal{X}^{(1,0)}|&=|D_{t}(2,x_{4},...,x_n)\cap D_{t-2-\ell_{(1,0)}^*}(y_{5+\ell_{(1,0)}^*},...,y_n)|,\nonumber\\
|\mathcal{X}^{(1,1)}|&=|D_{t-2-\ell_{(1,1)}}(x_{5+\ell_{(1,1)}},...,x_n)\cap D_{t-2-\ell_{(1,1)}^*}(y_{5+\ell_{(1,1)}^*},...,y_n)|,\nonumber\\
|\mathcal{X}^{(1,2)}|&=|D_{t-1}(x_{4},...,x_n)\cap D_{t-1}(y_{4},...,y_n)|,\label{eq11}
\end{align}
where 
$\ell_{(0,0)},\ell_{(0,0)}^*,\ell_{(0,1)},\ell_{(1,0)}^*,\ell_{(1,1)}^*,\ell_{(1,1)}\geqslant 0$ and $\mathbf{x}_{[4,n]}\neq \mathbf{y}_{[4,n]}$. Since $\mathbf{x}_{[4,n]}\neq \mathbf{y}_{[4,n]}$ and $n\geqslant t+3$, we have
\begin{align}
|\mathcal{X}^{(0,2)}|\leqslant N_3(n-3,1,t-1)=2D_3(n-5,t-2)+D_3(n-6,t-3),\label{eq12}\\
|\mathcal{X}^{(1,2)}|\leqslant N_3(n-3,1,t-1)=2D_3(n-5,t-2)+D_3(n-6,t-3).\label{eq13}
\end{align}
Since $\ell_{(0,0)}\neq \ell_{(0,1)}$ and $\ell_{(0,0)}, \ell_{(0,1)}\geqslant 0$, by Eq. $(\ref{eq11})$ we have
\begin{equation}
|\mathcal{X}^{(0,0)}|+|\mathcal{X}^{(0,1)}|\leqslant |D_{t-2-\ell_{(0,0)}}(x_{5+\ell_{(0,0)}},...,x_n)|+|D_{t-2-\ell_{(0,1)}}(x_{5+\ell_{(0,1)}},...,x_n)|\leqslant D_3(n-5,t-3)+D_3(n-4,t-2).\label{eq14}
\end{equation}
Since $\ell_{(1,0)}^*\neq \ell_{(1,1)}^*$ and $\ell_{(1,0)}^*,\ell_{(1,1)}^*\geqslant 0$, by Eq. $(\ref{eq11})$  we have
\begin{equation}
|\mathcal{X}^{(1,0)}|+|\mathcal{X}^{(1,1)}|\leqslant |D_{t-2-\ell_{(1,0)}^*}(y_{5+\ell_{(1,0)}^*},...,y_n)|+|D_{t-2-\ell_{(1,1)}^*}(y_{5+\ell_{(1,1)}^*},...,y_n)|\leqslant D_3(n-5,t-3)+D_3(n-4,t-2).\label{eq15}
\end{equation}
In \textbf{Case A2}, based on the values of $\ell_{(0,0)},\ell_{(0,1)},\ell_{(1,0)}^*,$ and $\ell_{(1,1)}^*$, we consider the value of $|\mathcal{X}|$ by the following three cases: \textbf{Case A2-1:} $\ell_{(0,0)},\ell_{(0,1)}\geqslant 1$ or $\ell_{(1,0)}^*,\ell_{(1,1)}^*\geqslant 1$; \textbf{Case A2-2:} $\{\ell_{(0,0)},\ell_{(0,1)}\}=\{0,a\}$ or $\{\ell_{(1,0)}^*,\ell_{(1,1)}^*\}=\{0,a\}$ with $a\geqslant 2$; \textbf{Case A2-3:} $\ell_{(0,0)}=\ell_{(1,0)}^*=0,\ell_{(0,1)}=\ell_{(1,1)}^*= 1$; \textbf{Case A2-4:} $\ell_{(0,0)}=\ell_{(1,0)}^*=1,\ell_{(0,1)}=\ell_{(1,1)}^*=0$; \textbf{Case A2-5:} $\ell_{(0,0)}=0,\ell_{(0,1)}=1$ and $\ell_{(1,1)}^*=0,\ell_{(1,0)}^*=1$; \textbf{Case A2-6:} $\ell_{(0,0)}=1,\ell_{(0,1)}=0$ and $\ell_{(1,0)}^*=0,\ell_{(1,1)}^*=1$.

\textbf{Case A2-1:} Given that $\ell_{(0,0)},\ell_{(0,1)}\geqslant 1$ or $\ell_{(1,0)}^*,\ell_{(1,1)}^*\geqslant 1$, without loss of generality, we consider the case where $\ell_{(0,0)},\ell_{(0,1)}\geqslant 1$. Since $\ell_{(0,0)},\ell_{(0,1)}\geqslant 1$ and $\ell_{(0,0)}\neq \ell_{(0,1)}$, by Eq. $(\ref{eq11})$ we have
\begin{equation}
|\mathcal{X}^{(0,0)}|+|\mathcal{X}^{(0,1)}|\leqslant |D_{t-2-\ell_{(0,0)}}(x_{5+\ell_{(0,0)}},...,x_n)|+|D_{t-2-\ell_{(0,1)}}(x_{5+\ell_{(0,1)}},...,x_n)|\leqslant D_3(n-5,t-3)+D_3(n-6,t-4).\label{eq16}
\end{equation}
Thus, by Eqs. $(\ref{eq10}),(\ref{eq12}),(\ref{eq13}),(\ref{eq15}),(\ref{eq16})$, we have
\begin{align*}
|\mathcal{X}|&=\big(|\mathcal{X}^{(0,0)}|+|\mathcal{X}^{(0,1)}|\big)+|\mathcal{X}^{(0,2)}|+\big(|\mathcal{X}^{(1,0)}|+|\mathcal{X}^{(1,1)}|\big)+|\mathcal{X}^{(1,2)}|+|\mathcal{X}^2|\\
&\leqslant D_3(n-5,t-3)+D_3(n-6,t-4)+2D_3(n-5,t-2)+D_3(n-6,t-3)+D_3(n-5,t-3)+D_3(n-4,t-2)\\
&~~~~+2D_3(n-5,t-2)+D_3(n-6,t-3)+2D_3(n-5,t-3)+D_3(n-6,t-4)\\
&\overset{(a)}{=}M_1(n,t)+\big(D_3(n-7,t-4)-D_3(n-5,t-2)\big)+(D_3(n-9,t-6)-D_3(n-6,t-3))\overset{(b)}{\leqslant} M_1(n,t),
\end{align*}
where $(a)$ follows from Eq. $(\ref{eq3})$ under the condition $n\geqslant t+3$, and $(b)$ follows from Eq. $(\ref{eq6})$.

\textbf{Case A2-2:} Given that $\{\ell_{(0,0)},\ell_{(0,1)}\}=\{0,a\}$ or $\{\ell_{(1,0)}^*,\ell_{(1,1)}^*\}=\{0,a\}$ with $a\geqslant 2$, without loss of generality, we consider the case where $\{\ell_{(0,0)},\ell_{(0,1)}\}=\{0,a\}$. By Eq. $(\ref{eq11})$ we have
\begin{equation}
|\mathcal{X}^{(0,0)}|+|\mathcal{X}^{(0,1)}|\leqslant |D_{t-2-\ell_{(0,0)}}(x_{5+\ell_{(0,0)}},...,x_n)|+|D_{t-2-\ell_{(0,1)}}(x_{5+\ell_{(0,1)}},...,x_n)|\leqslant D_3(n-4,t-2)+D_3(n-6,t-4).\nonumber
\end{equation}
Thus, combining with Eqs. $(\ref{eq10}),(\ref{eq12}),(\ref{eq13}),(\ref{eq15})$, we have
\begin{align*}
|\mathcal{X}|&=\big(|\mathcal{X}^{(0,0)}|+|\mathcal{X}^{(0,1)}|\big)+|\mathcal{X}^{(0,2)}|+\big(|\mathcal{X}^{(1,0)}|+|\mathcal{X}^{(1,1)}|\big)+|\mathcal{X}^{(1,2)}|+|\mathcal{X}^2|\\
&\leqslant D_3(n-4,t-2)+D_3(n-6,t-4)+2D_3(n-5,t-2)+D_3(n-6,t-3)+D_3(n-5,t-3)+D_3(n-4,t-2)\\
&~~~~+2D_3(n-5,t-2)+D_3(n-6,t-3)+2D_3(n-5,t-3)+D_3(n-6,t-4)\\
&\overset{(a)}{=}M_1(n,t)+\big(D_3(n-7,t-4)-D_3(n-6,t-3)\big)+\big(D_3(n-9,t-6)-D_3(n-8,t-5)\big)\overset{(b)}{\leqslant }M_1(n,t),
\end{align*}
where $(a)$ follows from Eq. $(\ref{eq3})$ under the condition $n\geqslant t+3$, and $(b)$ follows from Eq. $(\ref{eq6})$.

\textbf{Case A2-3:} Since $\ell_{(0,0)}=\ell_{(1,0)}^*=0,\ell_{(0,1)}=\ell_{(1,1)}^*= 1$, it follows that $\mathbf{x}_{[1,5]}=(1,0,2,0,1),\mathbf{y}_{[1,5]}=(0,1,2,0,1)$; and $\mathbf{x}_{[6,n]}\neq \mathbf{y}_{[6,n]}$ because  $d_L(\mathbf{x,y})\geqslant 2$. By setting $a=0,b=1,c=2$ in Lemma $\ref{lm12}$, we have $|\mathcal{X}|\leqslant M_0(n,t)$.

\textbf{Case A2-4:} Given that $\ell_{(0,0)}=\ell_{(1,0)}^*=1,\ell_{(0,1)}=\ell_{(1,1)}^*=0$, we have $\mathbf{x}_{[1,5]}=(1,0,2,1,0),\mathbf{y}_{[1,5]}=(0,1,2,1,0)$; and $\mathbf{x}_{[6,n]}\neq \mathbf{y}_{[6,n]}$ because  $d_L(\mathbf{x,y})\geqslant 2$. By swapping $0$ and $1$, we can obtain the same result of \textbf{Case A2-3}. Thus, we have $|\mathcal{X}|\leqslant M_0(n,t).$

\textbf{Case A2-5:} Since $\ell_{(0,0)}=0,\ell_{(0,1)}=1$ and $\ell_{(1,1)}^*=0,\ell_{(1,0)}^*=1$, we have $\mathbf{x}_{[1,5]}=(1,0,2,0,1)$ and $\mathbf{y}_{[1,5]}=(0,1,2,1,0)$. By Eq. $(\ref{eq11})$, we have
\begin{align}
|\mathcal{X}^{(0,0)}|&\leqslant |D_{t-3}(y_6,...,y_n)|\leqslant D_3(n-5,t-3),~
|\mathcal{X}^{(0,1)}|\leqslant |D_{t-3}(x_6,...,x_n)|\leqslant D_3(n-5,t-3),\nonumber\\
|\mathcal{X}^{(1,0)}|&\leqslant |D_{t-3}(y_6,...,y_n)|\leqslant D_3(n-5,t-3),~
|\mathcal{X}^{(1,1)}|\leqslant |D_{t-3}(x_6,...,x_n)|\leqslant D_3(n-5,t-3).\label{eq17}
\end{align}
Thus, by Eqs. $(\ref{eq10}),(\ref{eq12}),(\ref{eq13}),(\ref{eq17})$ we have
\begin{align*}
|\mathcal{X}|&=|\mathcal{X}^{(0,0)}|+|\mathcal{X}^{(0,1)}|+|\mathcal{X}^{(0,2)}|+|\mathcal{X}^{(1,0)}|+|\mathcal{X}^{(1,1)}|+|\mathcal{X}^{(1,2)}|+|\mathcal{X}^2|\\
&\leqslant D_3(n-5,t-3)+D_3(n-5,t-3)+2D_3(n-5,t-2)+D_3(n-6,t-3)+D_3(n-5,t-3)+D_3(n-5,t-3)\\
&~~~~+2D_3(n-5,t-2)+D_3(n-6,t-3)+2D_3(n-5,t-3)+D_3(n-6,t-4)\\
&\overset{(a)}{=} M_1(n,t)+D_3(n-7,t-4)+D_3(n-8,t-5)-2D_3(n-5,t-2)\overset{(b)}{\leqslant} M_1(n,t),
\end{align*}
where $(a)$ follows from Eq. $(\ref{eq3})$, and $(b)$ follows from Eq. $(\ref{eq6})$.

\textbf{Case A2-6:} Given that $\ell_{(0,0)}=1,\ell_{(0,1)}=0$ and $\ell_{(1,0)}^*=0,\ell_{(1,1)}^*=1$, we have $\mathbf{x}_{[1,5]}=(1,0,2,1,0),\mathbf{y}_{[1,5]}=(0,1,2,0,1)$. 
By setting $a=0,b=1,c=2$ in Lemma $\ref{lm13}$, we have $|\mathcal{X}|\leqslant M_1(n,t)$.

\textbf{Case B:} Since $\ell_0,\ell_1^*\geqslant 1$, it follows that
\begin{align}
|\mathcal{X}^0|\leqslant |D_{t-1-\ell_0}(\mathbf{x}_{[3+\ell_0,n]})|\leqslant D_3(n-3,t-2),\nonumber\\
|\mathcal{X}^1|\leqslant |D_{t-1-\ell_1^*}(\mathbf{y}_{[3+\ell_1^*,n]})|\leqslant D_3(n-3,t-2)\nonumber.
\end{align}
If $\ell_2=\ell_2^*=0$, then $\mathbf{x}_{[3,n]}\neq \mathbf{y}_{[3,n]}$ because $d_L(\mathbf{x,y})\geqslant 2$. Thus,
\begin{align}
|\mathcal{X}^2|=|D_{t-1}(\mathbf{x}_{[3,n]})\cap D_{t-1}(\mathbf{y}_{[3,n]})|\leqslant N_3(n-2,1,t-1)=2D_3(n-4,t-2)+D_3(n-5,t-3).\nonumber
\end{align}
If $\ell_2>0$ or $\ell_2^*>0$, then
\begin{align}
|\mathcal{X}^2|\leqslant\max\{|D_{t-1-\ell_2}(\mathbf{x}_{[3+\ell_2,n]})|,|D_{t-1-\ell_2^*}(\mathbf{y}_{[3+\ell_2^*,n]})|&\leqslant D_3(n-3,t-2)\leqslant 2D_3(n-4,t-2)+D_3(n-5,t-3).\nonumber
\end{align}
Thus, given that $\ell_0,\ell_1^*\geqslant 1$, it follows that
\begin{align}
|\mathcal{X}|&=|\mathcal{X}^0|+|\mathcal{X}^1|+|\mathcal{X}^2|\nonumber\\
&\leqslant 2D_3(n-3,t-2)+2D_3(n-4,t-2)+D_3(n-5,t-3)\nonumber\\
&\overset{(a)}{=}M_1(n,t)+D_3(n-9,t-6)+3D_3(n-7,t-4)-2D_3(n-5,t-2)-D_3(n-6,t-3)-D_3(n-8,t-5),\nonumber\\
&\overset{(b)}{\leqslant}M_1(n,t),\nonumber
\end{align}
where $(a)$ follows from Eq. $(\ref{eq3})$, and $(b)$ follows from Eq. $(\ref{eq6})$.

\textbf{Case C:} Given that $\ell_0=0,\ell_1^*\geqslant 1$, it follows that $x_2=0$, $y_2\neq 1$, and $\ell_2\geqslant 1$. Since $d_L(\mathbf{x,y})\geqslant 2$, we have $\mathbf{x}_{[3,n]}\notin D_1(\mathbf{y}_{[2,n]})$. By Lemma $\ref{lm7}$, it follows that
\begin{align}
|\mathcal{X}^0|=|D_{t-1}(\mathbf{x}_{[3,n]})\cap D_{t}(\mathbf{y}_{[2,n]})|&\leqslant N_3(n-2,t,t-1,1)\nonumber\\
&\overset{(a)}{=}D_3(n-4,t-2)+2D_3(n-5,t-2)+D_3(n-5,t-3)+D_3(n-6,t-3),\nonumber
\end{align}
where $(a)$ follows from Corollary $\ref{cor2}$. Since $\ell_1^*\geqslant 1$, we have
\begin{align}
|\mathcal{X}^1|\leqslant |D_{t-1-\ell_1^*}(\mathbf{y}_{[3+\ell_1^*,n]})|\leqslant D_3(n-3,t-2)\nonumber.
\end{align}
If $\ell_2=1$ and $\ell_2^*=0$, then $x_3=y_2=2$ and $\mathbf{x}_{[4,n]}\notin D_1(\mathbf{y}_{[3,n]})$ because $d_L(\mathbf{x,y})\geqslant 2$. Since $2n>3t$, we have $2(n-3)>3(t-2)$ which satisfies the conditions of $n,t$ in Lemma $\ref{lm7}$. By Lemma $\ref{lm7}$ we have
\begin{align}
|\mathcal{X}^2|=|D_{t-2}(\mathbf{x}_{[4,n]})\cap D_{t-1}(\mathbf{y}_{[3,n]})|&\leqslant N_3(n-3,t-1,t-2,1)\nonumber\\
&=D_3(n-5,t-3)+2D_3(n-6,t-3)+D_3(n-6,t-4)+D_3(n-7,t-4).\nonumber
\end{align}
If $\ell_2\geqslant 2$ or $\ell_2^*\geqslant 1$, then we only consider $\ell_2\geqslant 2, \ell_2^*\geqslant 2$, or $\ell_2=\ell_2^*=1$ since $\ell_2\geqslant 1$. When $\ell_2\geqslant 2~\text{or}~\ell_2^*\geqslant 2$,
we have
\begin{align}
|\mathcal{X}^2|&=|D_{t-1-\ell_2}(\mathbf{x}_{[3+\ell_2,n]})\cap D_{t-1-\ell_2^*}(\mathbf{y}_{[3+\ell_2^*,n]})|\leqslant D_3(n-4,t-3)\nonumber\\
&\overset{(a)}{\leqslant} D_3(n-5,t-3)+2D_3(n-6,t-3)+D_3(n-6,t-4)+D_3(n-7,t-4),\nonumber
\end{align}
where $(a)$ follows from Eqs. $(\ref{eq3})$ and $(\ref{eq6})$ under the condition $n\geqslant t+3$.
When $\ell_2=\ell_2^*=1$, then $x_3=y_3=2$ and  $\mathbf{x}_{[4,n]}\neq \mathbf{y}_{[4,n]}$ because $d_L(\mathbf{x,y})\geqslant 2$. By Theorem $\ref{thm1}$
we have
\begin{align}
|\mathcal{X}^2|&=|D_{t-2}(\mathbf{x}_{[4,n]})\cap D_{t-2}(\mathbf{y}_{[4,n]})|\leqslant N_3(n-3,1,t-2)=2D_3(n-5,t-3)+D_3(n-6,t-4)\nonumber\\
&\overset{(a)}{\leqslant} D_3(n-5,t-3)+2D_3(n-6,t-3)+D_3(n-6,t-4)+D_3(n-7,t-4),\nonumber
\end{align}
where $(a)$ follows from Eqs. $(\ref{eq3})$ and $(\ref{eq6})$ since $n\geqslant t+3$.

Given that $\ell_0=0,\ell_1^*\geqslant 1$, by computing the value of $|\mathcal{X}^i|$ for $i\in\{0,1,2\}$, we have
\begin{align}
|\mathcal{X}|&=|\mathcal{X}^0|+|\mathcal{X}^1|+|\mathcal{X}^2|\nonumber\\
&\leqslant D_3(n-4,t-2)+2D_3(n-5,t-2)+D_3(n-5,t-3)+D_3(n-6,t-3)+D_3(n-3,t-2)+D_3(n-5,t-3)\nonumber\\
&~~~~+2D_3(n-6,t-3)+D_3(n-6,t-4)+D_3(n-7,t-4)\nonumber\\
&\overset{(a)}{=}M_1(n,t)+D_3(n-6,t-3)+D_3(n-6,t-4)+2D_3(n-7,t-4)-2D_3(n-5,t-2)\nonumber\\
&~~~~-D_3(n-5,t-3)-D_3(n-8,t-5)\nonumber\\
&\overset{(b)}{\leqslant}M_1(n,t)+D_3(n-6,t-3)+D_3(n-6,t-4)-D_3(n-5,t-3)-D_3(n-8,t-5)\nonumber\\
&\overset{(c)}{=}M_1(n,t)+D_3(n-9,t-6)-D_3(n-8,t-5)\overset{(d)}{\leqslant} M_1(n,t),\nonumber
\end{align}
where $(a),(c)$ follow from Eq. $(\ref{eq3})$, and $(b),(d)$ follow from Eq. $(\ref{eq6})$ under the condition $n\geqslant t+3$.

\textbf{Case D:} Since $\ell_1^*=0,\ell_0\geqslant 1$, it follows that $y_2=1$, $x_2\neq 0$, and $\ell_2^*\geqslant 1$. Similarly, by using the method of proving \textbf{Case 3}, we also prove that
$$|\mathcal{X}|\leqslant M_1(n,t),$$
for $\ell_1^*=0,\ell_0\geqslant 1$.

By the above discussion, we have $|D_3(\mathbf{x},\mathbf{y};t,t)|\leqslant M_1(n,t)$ for $n\geqslant 9$, $\frac{2n}{3}>t\geqslant 2$, $d_L(\mathbf{x},\mathbf{y})\geqslant 2,$ and $x_1\neq y_1$.
\end{proof}

To compare the values of $M_0(n,t)$ and $M_1(n,t)$,  define $f(n,t)=M_1(n,t)-M_0(n,t)$ for  $n\geqslant \max\{9,\lfloor\frac{3t}{2}\rfloor+1\}$ and $t\geqslant 2$. By the definitions of $M_0(n,t)$ and $M_1(n,t)$, it is easily verified that $$f(n,t)=D_3(n-5,t-2)+D_3(n-6,t-2)+D_3(n-8,t-5)+D_3(n-10,t-5)-D_3(n-7,t-3)-2D_3(n-7,t-4)-D_3(n-12,t-7)$$ for $n\geqslant \max\{9,\lfloor\frac{3t}{2}\rfloor+1\}$ and $t\geqslant 2$.
When $t=2,3,4,$ or $5$, we have $f(n,t)\geqslant 0$ for $n\geqslant \max\{9,\lfloor\frac{3t}{2}\rfloor+1\}$. Next, we consider only the case where $t\geqslant 6$ and $n\geqslant \lfloor \frac{3t}{2}\rfloor +1$. In Lemma $\ref{lm15}$, we give some results of $f(n,t)$. The proof of Lemma $\ref{lm15}$ is given in Appendix $\ref{APP-F}$.

\begin{lemma}
\label{lm15}
Let $t\geqslant 6$ and $n\geqslant \max\{9,\lfloor \frac{3t}{2}\rfloor +1\}$. Then we have $f(3k+1,2k)=-1$ and $f(3k+2,2k+1)=0$ for any $k\geqslant 3$. Furthermore, if $n\geqslant 3t$, then $f(n,t)>0$.
\end{lemma}

In order to obtain the result of $|D_t(\mathbf{x})\cap D_t(\mathbf{y})|\leqslant \max\{M_0(n,t),M_1(n,t)\}$ for any $\mathbf{x},\mathbf{y}\in \mathbb{Z}_3^n$ with $d_L(\mathbf{x},\mathbf{y})\geqslant 2$, we require some properties of $M_i(n,t)$ for $i\in\{0,1\}$ in Lemmas $\ref{lm16}$-$\ref{lm18}$. Appendix $\ref{APP-G}$ contains the proofs of Lemmas $\ref{lm16}$-$\ref{lm18}$.

\begin{lemma}
\label{lm16}
Let $n\geqslant \max\{9,\lfloor \frac{3t}{2} \rfloor +1\}$ and $t \geqslant 2$. If $2(n-i) \leqslant 3(t-i+1)$, then $M_1(n-i,t-i+1)=3^{n-t-1}$ for $i\in\{1,2\}$. 
%If $2(n-i) \leqslant 3(t-i+1)$, then $M_0(n-i,t-i+1)=3^{n-t-1}$ except for the cases where $M_0(10,7)=3^3-1$, where $i\in\{1,2\}$.
\end{lemma}

\begin{lemma}
\label{lm17}
For $n\geqslant t+4$ and $t\geqslant 2$, we have
$$D_3(n-2,t-2)\leqslant \min\{M_0(n,t),M_1(n,t)\}.$$
Moreover, for $n\geqslant \max\{9,\lfloor \frac{3t}{2} \rfloor +1\}$ and $t\geqslant 2$, we have 
$$M_i(n-1,t)+2N_3(n-3,t-1,t-2,1)\leqslant M_1(n,t)$$ for $i\in \{0,1\}$.
\end{lemma}

\begin{lemma}
\label{lm18}
For $n\geqslant \max\{9,\lfloor \frac{3t}{2} \rfloor +1\}$ and $t\geqslant 2$, we have
\begin{align*}
D_3(n-4,t-3)&\leqslant N_3(n-3,t-1,t-2,1),\\
N_3(n-4,t-2,t-3,1)&\leqslant M_1(n-3,t-2),\\
M_0(n-2,t-1)+N_3(n-4,t-2,t-3,1)&\leqslant M_1(n-3,t-2)+M_1(n-2,t-1),\\
M_i(n-3,t-1)+2N_3(n-5,t-2,t-3,1)+M_0(n-3,t-2)&\leqslant M_1(n-3,t-2)+M_1(n-2,t-1)~\text{for}~i\in\{0,1\}.
\end{align*}
\end{lemma}

The main result follows:

\begin{theorem}
For $n\geqslant \max\{6,\lfloor \frac{3t}{2} \rfloor +1\}$ and $t\geqslant 2$,  $$N_3(n,2,t)\leqslant \max\{M_0(n,t),M_1(n,t)\}.$$
\label{thm5}
\end{theorem}
\begin{proof}
Since $n\geqslant \max\{9,\lfloor \frac{3t}{2} \rfloor +1\}$, we have $n\geqslant t+4$. When $9\geqslant n\geqslant \max\{6,\lfloor \frac{3t}{2}\rfloor+1\}$ and $t\geqslant 2$, we can verify that $N_3(n,2,t)=M_1(n,t)$ by using a computerized search. Next, we consider the case where $n\geqslant 10$. Let $\mathbf{x}=(\mathbf{u},\mathbf{v},\mathbf{w}), \mathbf{y}=(\mathbf{u},\mathbf{v}',\mathbf{w})\in \mathbb{Z}_3^n$, and denote $\mathcal{X}=D_t(\mathbf{x})\cap D_t(\mathbf{y})$, where $|\mathbf{u}|=p,|\mathbf{w}|=q,|\mathbf{x}|=|\mathbf{y}|=n$, and $d_L(\mathbf{x,y})\geqslant 2$. If $p=0$ or $q=0$, then by Lemma $\ref{lm14}$ we have $|\mathcal{X}|\leqslant \max\{M_0(n,t),M_1(n,t)\}$. Next, we consider only the case where $p,q\geqslant 1$. For convenience, let $\mathbf{x}=(x_1,...,x_n), \mathbf{y}=(y_1,...,y_n)$. For some fixed $q\geqslant 1$, we will prove that $|\mathcal{X}|\leqslant M_1(n,t)$ by induction on the length of $\mathbf{u}$. Based on the value of $p$, for $2n\geqslant 3t+1$, we now prove that $|\mathcal{X}|\leqslant M_1(n,t)$ in the following three cases: \textbf{Case A:} $p=1$; \textbf{Case B:} $p=2$; \textbf{Case C:} $p=3$.

\textbf{Case A:} Given that $p=1$, it follows that $x_1=y_1$ and $x_2\neq y_2$. Let $\mathbb{Z}_3=\{x_1,a,b\}$. By Lemmas $\ref{lm1}$ and $\ref{lm2}$, we have
$$|\mathcal{X}|=|\mathcal{X}^{x_1}|+|\mathcal{X}^{a}|+|\mathcal{X}^{b}|$$
with $\mathcal{X}^{x_1}=x_1\circ \big(D_t(\mathbf{x}_{[2,n]})\cap D_t(\mathbf{y}_{[2,n]})\big)$, $\mathcal{X}^{a}=a\circ\big(D_{t-1-\ell_0}(\mathbf{x}_{[3+\ell_0,n]})\cap D_{t-1-\ell_0^*}(\mathbf{y}_{[3+\ell_0^*,n]})\big)$, $\mathcal{X}^{b}=b\circ\big(D_{t-1-\ell_1}(\mathbf{x}_{[3+\ell_1,n]})\cap D_{t-1-\ell_1^*}(\mathbf{y}_{[3+\ell_1^*,n]})\big)$, where $\ell_0,\ell_0^*,\ell_1,\ell_1^*\geqslant 0$. Moreover, since $d_L(\mathbf{x,y})\geqslant 2$ and $x_1=y_1$, we have  $d_L(\mathbf{x}_{[2,n]},\mathbf{y}_{[2,n]})\geqslant 2$. Thus, for $n\geqslant 10$, if $2(n-1)\geqslant 3t+1$, then by Lemma $\ref{lm14}$
\begin{align}
|\mathcal{X}^{x_1}|\leqslant \max\{M_0(n-1,t),M_1(n-1,t)\}.\label{eq18}
\end{align}
Under the condition of $2n\geqslant 3t+1$, if $2(n-1)\leqslant 3t$, then $|\mathcal{X}^{x_1}|\leqslant 3^{n-1-t}$. By Lemma $\ref{lm16}$, we also have Eq. $(\ref{eq18})$.

When $t=2$, we have that  $|\mathcal{X}^{a}|=|\mathcal{X}^{b}|=0$. For $t\geqslant 3$, to compute $|\mathcal{X}^{a}|$ and $|\mathcal{X}^{b}|$, based on the values of $\ell_0, \ell_0^*,\ell_1, \ell_1^*$, we analyze the following two cases: \textbf{Case A1:} $\ell_0\geqslant 2, \ell_0^*\geqslant 2,\ell_1\geqslant 2,$ or $\ell_1^*\geqslant 2$; \textbf{Case A2:} $\ell_0,\ell_0^*,\ell_1,\ell_1^*\leqslant 1$.

\textbf{Case A1:} Given that $\ell_0\geqslant 2, \ell_0^*\geqslant 2,\ell_1\geqslant 2,$ or $\ell_1^*\geqslant 2$, without loss of generality, we let $\ell_0\geqslant 2$. Then we have
\begin{align}
|\mathcal{X}^{a}|\leqslant |D_{t-1-\ell_0}(\mathbf{x}_{[3+\ell_0,n]})|\leqslant |D_{t-3}(\mathbf{x}_{[5,n]})|\leqslant D_3(n-4,t-3)
\overset{(a)}{\leqslant} N_3(n-3,t-1,t-2,1)\label{eq19},
\end{align}
where $(a)$ follows from Lemma $\ref{lm18}$. We discuss the value of $|\mathcal{X}^{b}|$. Since $\ell_0\geqslant 2$, we have $x_i\neq a$ for $i\in \{1,2,3\}$.
Moreover, since $p=1$, we have $x_2\neq y_2$. According to the values of $x_2$ and $y_2$, we consider the following four cases: \textbf{Case A1-1:} $x_2=b,y_2=a$; \textbf{Case A1-2:} $x_2=x_1,y_2=b$; \textbf{Case A1-3:} $x_2=b,y_2=x_1$; \textbf{Case A1-4:} $x_2=x_1,y_2=a$.

\textbf{Case A1-1:} We further discuss the value of $|\mathcal{X}^{b}|$ based on the value of $y_3$. %For $y_3=b$, by using Lemma $\ref{lm7}$, we have 
%\begin{align}
%|\mathcal{X}^{b}|=|D_{t-1}(\mathbf{x}_{[3,n]})\cap D_{t-2}(\mathbf{y}_{[4,n]})|\leqslant N_3(n-3,t-1,t-2,1).\nonumber
%\end{align}
Given $x_2=b$ and $y_2=a$, for $y_3=b$, we have $x_1=y_1$ and $x_2=y_3$. Thus, for $y_3=b$,  since $d_L(\mathbf{x,y})\geqslant 2$, it follows that $\mathbf{y}_{[4,n]}\notin D_1(\mathbf{x}_{[3,n]})$. Moreover, since $2n>3t$, we have $2(n-3)>3(t-2)$. Thus, for $y_3=b$, by Lemma $\ref{lm7}$ we have 
\begin{align}
|\mathcal{X}^{b}|=|D_{t-1}(\mathbf{x}_{[3,n]})\cap D_{t-2}(\mathbf{y}_{[4,n]})|\leqslant N_3(n-3,t-1,t-2,1).\nonumber
\end{align}
If $y_3\neq b$, then $\ell_1^*\geqslant 2$. Thus, for $y_3\neq b$,
\begin{align}
|\mathcal{X}^{b}|&=|D_{t-1}(\mathbf{x}_{[3,n]})\cap D_{t-1-\ell_1^*}(\mathbf{y}_{[3+\ell_1^*,n]})|\leqslant |D_{t-1-\ell_1^*}(\mathbf{y}_{[3+\ell_1^*,n]})|\leqslant |D_{t-3}(\mathbf{y}_{[5,n]})|\leqslant D_3(n-4,t-3)\nonumber\\
&\overset{(a)}{\leqslant} N_3(n-3,t-1,t-2,1),\nonumber
\end{align}
where $(a)$ follows from Lemma $\ref{lm18}$. 

Hence, given that $x_2=b$ and $y_2=a$, it follows that
\begin{equation}
|\mathcal{X}^{b}|\leqslant N_3(n-3,t-1,t-2,1)\label{eq20}.
\end{equation}
By Eqs. $(\ref{eq18})$-$(\ref{eq20})$ and Lemma $\ref{lm17}$, we have
\begin{align*}
|\mathcal{X}|=|\mathcal{X}^{x_1}|+|\mathcal{X}^a|+|\mathcal{X}^b|\leqslant \max\{M_0(n-1,t),M_1(n-1,t)\}+2N_3(n-3,t-1,t-2,1)\leqslant M_1(n,t).
\end{align*}

\textbf{Cases A1-2} and \textbf{A1-3}: When $x_2=x_1,y_2=b$ or $x_2=b,y_2=x_1$, we consider the value of $|\mathcal{X}^{b}|$ based on the value of $x_3$ or $y_3$. Similarly, by the method of discussing \textbf{Case A1-1}, we have
$$|\mathcal{X}|\leqslant \max\{M_0(n-1,t),M_1(n-1,t)\}+2N_3(n-3,t-1,t-2,1)\leqslant M_1(n,t).$$

\textbf{Case A1-4:} When $x_2=x_1$ and $y_2=a$, we will compute the value of $|\mathcal{X}^{b}|$ based on the values of $x_3$ and $y_3$. For $x_3=y_3=b$, it follows that 
\begin{align}
|\mathcal{X}^{b}|&=|D_{t-2}(\mathbf{x}_{[4,n]})\cap D_{t-2}(\mathbf{y}_{[4,n]})|\leqslant N_3(n-3,1,t-2)\leqslant N_3(n-3,t-1,t-2,1).\nonumber
\end{align}
since $\mathbf{x}_{[4,n]}\neq \mathbf{y}_{[4,n]}$.

For $y_3\neq b$ or $x_3\neq b$, we have  $\ell_1^*\geqslant 2$ or $\ell_1\geqslant2$ such that
\begin{align}
|\mathcal{X}^{b}|\leqslant \min\{|D_{t-1-\ell_1}(\mathbf{x}_{[3+\ell_1,n]})|,|D_{t-1-\ell_1^*}(\mathbf{y}_{[3+\ell_1^*,n]})|\}
\leqslant D_3(n-4,t-3)\overset{(a)}{\leqslant} N_3(n-3,t-1,t-2,1),\nonumber
\end{align}
where $(a)$ follows from Lemma $\ref{lm18}$.

%then $\ell_1^*\geqslant 2$ or $\ell_1\geqslant 2$. Thus, we have
%\begin{align}
%|\mathcal{X}^{b}|&=|D_{t-1-\ell_1}(\mathbf{x}_{[3+\ell_1,n]})\cap D_{t-1-\ell_1^*}(\mathbf{y}_{[3+\ell_1^*,n]})|\leqslant \min\{|D_{t-1-\ell_1}(\mathbf{x}_{[3+\ell_1,n]})|,|D_{t-1-\ell_1^*}(\mathbf{y}_{[3+\ell_1^*,n]})|\}\nonumber\\
%&\leqslant D_3(n-4,t-3)\overset{(a)}{\leqslant} N_3(n-3,t-1,t-2,1),\nonumber
%\end{align}
%where $(a)$ follows from Lemma $\ref{lm18}$. 

%since $d_L(\mathbf{x,y})\geqslant 2$ and $x_1=y_1$, we have $\mathbf{x}_{[4,n]}\neq \mathbf{y}_{[4,n]}$. Thus, since $n\geqslant t+4$, we have 
%\begin{align}
%|\mathcal{X}^{b}|&=|D_{t-2}(\mathbf{x}_{[4,n]})\cap D_{t-2}(\mathbf{y}_{[4,n]})|\leqslant N_3(n-3,1,t-2)=2D_3(n-5,t-3)+D_3(n-6,t-4)\nonumber\\
%&\leqslant D_3(n-5,t-3)+2D_3(n-6,t-3)+D_3(n-6,t-4)+D_3(n-7,t-4)= N_3(n-3,t-1,t-2,1).\nonumber
%\end{align}
%If $y_3\neq b$ or $x_3\neq b$, then $\ell_1^*\geqslant 2$ or $\ell_1\geqslant 2$. Thus, we have
%\begin{align}
%|\mathcal{X}^{b}|&=|D_{t-1-\ell_1}(\mathbf{x}_{[3+\ell_1,n]})\cap D_{t-1-\ell_1^*}(\mathbf{y}_{[3+\ell_1^*,n]})|\leqslant \min\{|D_{t-1-\ell_1}(\mathbf{x}_{[3+\ell_1,n]})|,|D_{t-1-\ell_1^*}(\mathbf{y}_{[3+\ell_1^*,n]})|\}\nonumber\\
%&\leqslant D_3(n-4,t-3)\overset{(a)}{\leqslant} N_3(n-3,t-1,t-2,1),\nonumber
%\end{align}
%where $(a)$ follows from Lemma $\ref{lm18}$.  

Hence, if $x_2=x_1$ and $y_2=a$, then we have
\begin{equation}
|\mathcal{X}^{b}|\leqslant N_3(n-3,t-1,t-2,1)\label{eq21}.
\end{equation}
By Eqs. $(\ref{eq18})$, $(\ref{eq19})$, and $(\ref{eq21})$, we have
\begin{align*}
|\mathcal{X}|=|\mathcal{X}^{x_1}|+|\mathcal{X}^a|+|\mathcal{X}^b|\leqslant \max\{M_0(n-1,t),M_1(n-1,t)\}+2N_3(n-3,t-1,t-2,1)\leqslant M_1(n,t).
\end{align*}
By the above discussion, if $\ell_0\geqslant 2$, then we have $|\mathcal{X}|\leqslant M_1(n,t).$

\textbf{Cases A2:} Given that $\ell_0,\ell_0^*,\ell_1,\ell_1^*\leqslant 1$, $\ell_0\neq \ell_1$, $\ell_0\neq \ell_0^*$, $\ell_0^*\neq \ell_1^*$, without loss of generality, let $\ell_0=0, \ell_1=1$, $\ell_0^*=1$, and $\ell_1^*=0$. Then $x_1=y_1$, $x_2=y_3=a$, and $x_3=y_2=b$. Thus, we have $\mathbf{x}_{[4,n]}\notin D_1(\mathbf{y}_{[3,n]})$ and $\mathbf{y}_{[4,n]}\notin D_1(\mathbf{x}_{[3,n]})$ because $d_L(\mathbf{x,y})\geqslant 2$. Moreover, since $2n>3t$, we have $2(n-3)>3(t-2)$. Thus, by Lemma $\ref{lm7}$ we obtain
\begin{align}
|\mathcal{X}^{a}|&=|D_{t-1}(\mathbf{x}_{[3,n]})\cap D_{t-2}(\mathbf{y}_{[4,n]})|\leqslant N_3(n-3,t-1,t-2,1),\nonumber\\
|\mathcal{X}^{b}|&=|D_{t-2}(\mathbf{x}_{[4,n]})\cap D_{t-1}(\mathbf{y}_{[3,n]})|\leqslant  N_3(n-3,t-1,t-2,1).\nonumber
\end{align}
Thus,
\begin{align*}
|\mathcal{X}|=|\mathcal{X}^{x_1}|+|\mathcal{X}^a|+|\mathcal{X}^b|\leqslant \max\{M_0(n-1,t),M_1(n-1,t)\}+2N_3(n-3,t-1,t-2,1)\leqslant M_1(n,t).
\end{align*}

Therefore, when $p=1$, for $n\geqslant10$ and $t\geqslant 2$, we have
\begin{align}
|\mathcal{X}|\leqslant \max\{M_0(n-1,t),M_1(n-1,t)\}+2N_3(n-3,t-1,t-2,1)\leqslant M_1(n,t).\nonumber
\end{align}

\textbf{Case B:} Given that $p=2$, it follows that $x_1=y_1$, $x_2=y_2$, and $x_3\neq y_3$. For convenience, let $\mathbb{Z}_3=\{x_1,a,b\}$. By Lemmas $\ref{lm1}$ and $\ref{lm2}$, we have
$$|\mathcal{X}|=|\mathcal{X}^{x_1}|+|\mathcal{X}^{a}|+|\mathcal{X}^{b}|,$$
such that $\mathcal{X}^{x_1}=x_1\circ \big(D_t(\mathbf{x}_{[2,n]})\cap D_t(\mathbf{y}_{[2,n]})\big)$, $\mathcal{X}^{a}=a\circ\big(D_{t-1-\ell_0}(\mathbf{x}_{[3+\ell_0,n]})\cap D_{t-1-\ell_0^*}(\mathbf{y}_{[3+\ell_0^*,n]})\big)$, $\mathcal{X}^{b}=b\circ\big(D_{t-1-\ell_1}(\mathbf{x}_{[3+\ell_1,n]})\cap D_{t-1-\ell_1^*}(\mathbf{y}_{[3+\ell_1^*,n]})\big)$, where $\ell_0,\ell_0^*,\ell_1,\ell_1^*\geqslant 0$. Moreover, since $d_L(\mathbf{x,y})\geqslant 2$ and $x_1=y_1$, we have $d_L(\mathbf{x}_{[2,n]},\mathbf{y}_{[2,n]})\geqslant 2$.

%Consider the case where $n\geqslant 11$. 
For $n\geqslant 11$, under the condition of $2(n-1)\geqslant 3t+1$, by using the above result with $p=1$,
\begin{align}
|\mathcal{X}^{x_1}|\leqslant  M_1(n-1,t),\label{eq22}
\end{align} since $x_2=y_2$ and $x_3\neq y_3$. Under the condition of $2n\geqslant 3t+1$, if $2(n-1)\leqslant 3t$, then $|\mathcal{X}^{x_1}|\leqslant 3^{n-1-t}$. By Lemma $\ref{lm16}$ we also have Eq. $(\ref{eq22})$. When $n=10$, we also have Eq. $(\ref{eq22})$ since $|\mathcal{X}^{x_1}|\leqslant N_3(9,2,t)=M_1(9,t)$.

Based on the values of $x_2$ and $y_2$, we discuss the following two subcases: \textbf{Case B1:} $x_2\neq x_1$; \textbf{Case B2:} $x_2=x_1$.

%Under the condition of $2n\geqslant 3t+1$, if $2(n-1)\leqslant 3t$, then $|\mathcal{X}^{x_1}|\leqslant 3^{n-1-t}$. By Lemma $\ref{lm16}$ we also have Eq. $(\ref{eq22})$. 

\textbf{Case B1:} Since $x_2\neq x_1$, for convenience, let $x_2=y_2=a$. Then $d_L(\mathbf{x}_{[3,n]},\mathbf{y}_{[3,n]})\geqslant 2$ and $\ell_0=\ell_0^*=0$ since $d_L(\mathbf{x,y})\geqslant 2$, $x_1=y_1$, and $x_2=y_2$.  If $2(n-2)\leqslant 3(t-1)$, then $M_1(n-2,t-1)=3^{n-t-1}\geqslant |D_{t-1}(\mathbf{x}_{[3,n]})|\geqslant |\mathcal{X}^{a}|$. For $n\geqslant 11$, if $2(n-2)\geqslant 3(t-1)+1$, then by Lemma $\ref{lm14}$ we have
\begin{align}
|\mathcal{X}^{a}|=|D_{t-1}(\mathbf{x}_{[3,n]})\cap D_{t-1}(\mathbf{y}_{[3,n]})|\leqslant \max\{M_0(n-2,t-1),M_1(n-2,t-1)\}.\label{eq23}
\end{align}
When $n=10$, we also have Eq. $(\ref{eq23})$ since  $|\mathcal{X}^{a}|\leqslant N_3(8,2,t-1)=M_1(8,t-1)$.

Based on the values of $x_3$ and $y_3$, we discuss $|\mathcal{X}^{b}|$ in the following two cases: \textbf{Case B1-1:} $x_3\neq b$ and $y_3\neq b$; \textbf{Case B1-2} $x_3=b$ or $y_3=b$. 

\textbf{Case B1-1:} Given that $x_3\neq b$ and $y_3\neq b$, for $x_4=y_4=b$, by the method of discussing \textbf{Case A1-4}, we have 
\begin{align}
|\mathcal{X}^{b}|&=|D_{t-3}(\mathbf{x}_{[5,n]})\cap D_{t-3}(\mathbf{y}_{[5,n]})|\leqslant N_3(n-4,1,t-3)
\leqslant N_3(n-4,t-2,t-3,1).\nonumber
\end{align}

%it follows that $x_1=y_1,x_2=y_2,x_4=y_4$. Combining with the condition $d_L(\mathbf{x,y})\geqslant 2$, we have $\mathbf{x}_{[5,n]}\neq \mathbf{y}_{[5,n]}$. Thus, for $x_4=y_4=b$,
%we have $\ell_1=\ell_1^*=2$, and
%\begin{align}
%|\mathcal{X}^{b}|&=|D_{t-3}(\mathbf{x}_{[5,n]})\cap D_{t-3}(\mathbf{y}_{[5,n]})|\leqslant N_3(n-4,1,t-3)
%=2D_3(n-6,t-4)+D_3(n-7,t-5)\nonumber\\
%&=D_3(n-6,t-4)+D_3(n-7,t-4)+D_3(n-8,t-5)+D_3(n-9,t-6)+D_3(n-7,t-5)\nonumber\\
%&\overset{(a)}{\leqslant} D_3(n-6,t-4)+D_3(n-7,t-4)+D_3(n-8,t-5)+D_3(n-7,t-4)+D_3(n-7,t-5)\nonumber\\
%&=N_3(n-4,t-2,t-3,1),\nonumber
%\end{align}
%where $(a)$ follows from Eq. $(\ref{eq6})$ under the condition $n\geqslant t+4$. If $x_4\neq b$ or $y_4\neq b$, then $\ell_1\geqslant 3$ or $\ell_1^*\geqslant 3$. 
For $x_4\neq b$ or $y_4\neq b$, it follows that $\ell_1\geqslant 3$ or $\ell_1^*\geqslant 3$ such that
\begin{align}
|\mathcal{X}^{b}|\leqslant \min\{|D_{t-1-\ell_1}(\mathbf{x}_{[3+\ell_1,n]})|, |D_{t-1-\ell_1^*}(\mathbf{y}_{[3+\ell_1^*,n]})|\}\leqslant D_3(n-5,t-4)\overset{(a)}{\leqslant} N_3(n-4,t-2,t-3,1),\nonumber
\end{align}
where $(a)$ follows from Lemma $\ref{lm18}$.

\textbf{Case B1-2:} Given that $x_3=b$ or $y_3=b$, without loss of generality, let $x_3=b$ and $y_3\neq b$ under the condition $p=2$. 

If $y_4\neq b$, then $\ell_1^*\geqslant 3$. Thus, for $y_4\neq b$, we have 
\begin{align}
|\mathcal{X}^{b}|&\leqslant |D_{t-1-\ell_1^*}(\mathbf{y}_{[3+\ell_1^*,n]})|\leqslant D_3(n-5,t-4)\overset{(a)}{\leqslant} N_3(n-4,t-2,t-3,1),\nonumber
\end{align}
where  $(a)$ follows from Lemma $\ref{lm18}$. 

If $y_4=b$, then $x_3=y_4$. For $x_1=x_2$, $x_2=y_2$, $x_3=y_4$, and $d_L(\mathbf{x,y})\geqslant 2$, it follows that 
$\mathbf{y}_{[5,n]}\notin D_1(\mathbf{x}_{[4,n]})$.
Moreover, for $2n\geqslant 3t+1$, we have $2(n-4)\geqslant 3(t-3)+1$. For $y_4=b$, by Lemma $\ref{lm7}$ we have
\begin{align}
|\mathcal{X}^{b}|=|D_{t-2}(\mathbf{x}_{[4,n]})\cap D_{t-3}(\mathbf{y}_{[5,n]})|\leqslant N_3(n-4,t-2,t-3,1).\nonumber
\end{align}
Therefore, in \textbf{Case B1} where $x_2\neq x_1$, we obtain
\begin{align}
|\mathcal{X}^{b}|&\leqslant N_3(n-4,t-2,t-3,1).\label{eq24}
\end{align}

\textbf{Case B2:} Given that $x_2=x_1$, for  $x_3\neq y_3$, without loss of generality let $x_3=b\neq x_1$. Then $\ell_0\geqslant 2$. Thus, we have 
\begin{align}
|\mathcal{X}^{a}|\leqslant |D_{t-1-\ell_0}(\mathbf{x}_{[3+\ell_0,n]})|\leqslant D_3(n-4,t-3)\overset{(a)}{\leqslant} M_1(n-2,t-1),\label{eq25}
\end{align}
where $(a)$ follows from Lemma $\ref{lm18}$. Based on the value of $y_4$, we discuss $|\mathcal{X}^{b}|$ in the following two cases: \textbf{Case B2-1:} $y_4\neq b$; \textbf{Case B2-2:} $y_4=b$.

\textbf{Case B2-1:} Given that $y_4\neq b$, it follows that $\ell_1^*\geqslant 3$. Thus,
\begin{align}
|\mathcal{X}^{b}|\leqslant |D_{t-1-\ell_1^*}(\mathbf{y}_{[3+\ell_1^*,n]})|\leqslant D_3(n-5,t-4)
\overset{(a)}{\leqslant} N_3(n-4,t-2,t-3,1),\nonumber
\end{align}
where $(a)$ follows from Lemma $\ref{lm18}$. 

\textbf{Case B2-2:} Given that $y_4=b$, it follows that $x_3=y_4$. Combining with the conditions $x_1=x_2$, $x_2=y_2$, and  $d_L(\mathbf{x,y})\geqslant 2$, we have $\mathbf{y}_{[5,n]}\notin D_1(\mathbf{x}_{[4,n]})$. Moreover, for $2n\geqslant 3t+1$, it follows that $2(n-4)\geqslant 3(t-3)+1$.  By Lemma $\ref{lm7}$, we have
\begin{align}
|\mathcal{X}^{b}|=|D_{t-2}(\mathbf{x}_{[4,n]})\cap D_{t-3}(\mathbf{y}_{[5,n]})|\leqslant N_3(n-4,t-2,t-3,1).\nonumber
\end{align}

Thus, in \textbf{Case B2} where $x_2=x_1$, we have
\begin{align}
|\mathcal{X}^{b}|&\leqslant N_3(n-4,t-2,t-3,1).\label{eq26}
\end{align}

Therefore, when $p=2$, for $n\geqslant 10$, by Eqs. $(\ref{eq22})$-$(\ref{eq26})$ we have
\begin{align*}
|\mathcal{X}|&=|\mathcal{X}^{x_1}|+|\mathcal{X}^a|+|\mathcal{X}^b|\\
&\leqslant M_1(n-1,t)+\max\{M_0(n-2,t-1),M_1(n-2,t-1)\}+N_3(n-4,t-2,t-3,1)\\
&\leqslant M_1(n-1,t)+\max\{M_1(n-2,t-1)+N_3(n-4,t-2,t-3,1),M_0(n-2,t-1)+N_3(n-4,t-2,t-3,1)\}\\
&\overset{(a)}{\leqslant} M_1(n-1,t)+M_1(n-2,t-1)+M_1(n-3,t-2)\overset{(b)}{=}M_1(n,t),
\end{align*}
where $(a)$ follows from Lemma $\ref{lm18}$, and $(b)$ follows from Eq. $(\ref{eq8})$.

\textbf{Case C:} For $p=3$, it follows that $x_1=y_1$, $x_2=y_2$, $x_3=y_3$ and $x_4\neq y_4$. For convenience, let $\mathbb{Z}_3=\{x_1,a,b\}$. By Lemmas $\ref{lm1}$ and $\ref{lm2}$, we have
$$|\mathcal{X}|=|\mathcal{X}^{x_1}|+|\mathcal{X}^{a}|+|\mathcal{X}^{b}|,$$ where
$\mathcal{X}^{x_1}=x_1\circ \big(D_t(\mathbf{x}_{[2,n]})\cap D_t(\mathbf{y}_{[2,n]})\big)$, $\mathcal{X}^{a}=a\circ\big(D_{t-1-\ell_0}(\mathbf{x}_{[3+\ell_0,n]})\cap D_{t-1-\ell_0^*}(\mathbf{y}_{[3+\ell_0^*,n]})\big)$, $\mathcal{X}^{b}=b\circ\big(D_{t-1-\ell_1}(\mathbf{x}_{[3+\ell_1,n]})\cap D_{t-1-\ell_1^*}(\mathbf{y}_{[3+\ell_1^*,n]})\big)$ with $\ell_0,\ell_0^*,\ell_1,\ell_1^*\geqslant 0$. Moreover, since $d_L(\mathbf{x,y})\geqslant 2$ and $x_1=y_1$, we have $d_L(\mathbf{x}_{[2,n]},\mathbf{y}_{[2,n]})\geqslant 2$.
Thus, for $n\geqslant 11$, under the condition of $2(n-1)\geqslant 3t+1$,  by using the above result with $p=2$
\begin{align}
|\mathcal{X}^{x_1}|\leqslant  M_1(n-1,t),\label{eq27}
\end{align}
for $x_2=y_2,x_3=y_3,$ and $x_4\neq y_4$. Under the condition of $2n\geqslant 3t+1$, if $2(n-1)\leqslant 3t$, then $|\mathcal{X}^{x_1}|\leqslant 3^{n-1-t}$. By Lemma $\ref{lm16}$ we also have Eq. $(\ref{eq27})$. When $n=10$, we also have Eq. $(\ref{eq27})$ since $|\mathcal{X}^{x_1}|\leqslant N_3(9,2,t)=M_1(9,t)$.

%Under the condition of $2n\geqslant 3t+1$, if $2(n-1)\leqslant 3t$, then $|\mathcal{X}^{x_1}|\leqslant 3^{n-1-t}$. By Lemma $\ref{lm16}$ we also have Eq. $(\ref{eq27})$. 
Based on the values of $x_2$ and $x_3$, we discuss the following three cases: \textbf{Case C1:} $\{x_1,x_2,x_3\}=\mathbb{Z}_3=\{x_1,a,b\}$; \textbf{Case C2:} $\{x_i|i\in\{1,2,3\}\}=\{x_1,a\}$ or $\{x_1,b\}$; \textbf{Case C3:} $x_1=x_2=x_3$.

\textbf{Case C1:} Given that $\{x_1,x_2,x_3\}=\mathbb{Z}_3$, without loss of generality, let $x_2=a,x_3=b$. Since $d_L(\mathbf{x,y})\geqslant 2$, $x_1=y_1$, $x_2=y_2=a$, $x_3=y_3=b$, we have $d_L(\mathbf{x}_{[3,n]},\mathbf{y}_{[3,n]})\geqslant 2$, $d_L(\mathbf{x}_{[4,n]},\mathbf{y}_{[4,n]})\geqslant 2$, $\ell_0=\ell_0^*=0$, and $\ell_1=\ell_1^*=1$. If $2(n-2)\leqslant 3(t-1)$, then $M_1(n-2,t-1)=3^{n-t-1}\geqslant |D_{t-1}(\mathbf{x}_{[3,n]})|\geqslant |\mathcal{X}^{a}|$ and $M_i(n-3,t-1)+2N_3(n-5,t-2,t-3,1)=3^{n-t-2}+2\cdot 3^{n-t-2}=3^{n-t-1}$. For $n\geqslant12$, if $2(n-2)\leqslant 3(t-1)+1$, by using the result of $p=1$, we have
\begin{align}
|\mathcal{X}^{a}|\leqslant \max\{M_0(n-3,t-1),M_1(n-3,t-1)\}+2N_3(n-5,t-2,t-3,1) \leqslant  M_1(n-2,t-1).\label{eq28}
\end{align}
For $n\geqslant 12$ and $t\geqslant4$,  since $2(n-3)\geqslant 3(t-2)+1$, by Lemma $\ref{lm14}$ we have
\begin{align}
|\mathcal{X}^{b}|\leqslant \max\{M_0(n-3,t-2),M_1(n-3,t-2)\}.\label{eq29}
\end{align}
For $n=11$ or $10$, we also have Eqs. (\ref{eq28}) and (\ref{eq29}).

\textbf{Case C2:} Given that $\{x_i|i\in\{1,2,3\}\}=\{x_1,a\}$ or $\{x_1,b\}$, without loss of generality, let $\{x_i|i\in\{1,2,3\}\}=\{x_1,a\}$. If $x_2=a$, $x_i\neq b$, and $y_i\neq b$ for $i\in \{1,2,3\}$, then we have $\ell_0=\ell_0^*=0$ and $\ell_1,\ell_1^*\geqslant 2$. Similarly, we also have Eq. $(\ref{eq28})$. Moreover, for  $n\geqslant 12$ and $t\geqslant4$,  we have
\begin{align}
|\mathcal{X}^{b}|&=|D_{t-1-\ell_1}(\mathbf{x}_{[3+\ell_1,n]})\cap D_{t-1-\ell_1^*}(\mathbf{y}_{[3+\ell_1^*,n]})|\leqslant |D_{t-1-1}(\mathbf{x}_{[4,n]})\cap D_{t-1-1}(\mathbf{y}_{[4,n]})|\nonumber\\
&\overset{(a)}{\leqslant} \max\{M_0(n-3,t-2),M_1(n-3,t-2)\},\label{eq30}
\end{align}
where $(a)$ follows from Lemma $\ref{lm14}$ under the conditions $x_4\neq y_4$, $d_L(\mathbf{x}_{[4,n]},\mathbf{y}_{[4,n]})\geqslant 2$, and $2(n-3)\geqslant 3(t-2)+1$.
For $n=11$ or $10$, we also have Eq. (\ref{eq30}).

\textbf{Case C3:} Given $x_1=x_2=x_3$, we have $\ell_0,\ell_0^*,\ell_1,\ell_1^*\geqslant 2$. Thus, 
\begin{align*}
|\mathcal{X}^{a}|&=|D_{t-1-\ell_0}(\mathbf{x}_{[3+\ell_0,n]})\cap D_{t-1-\ell_0^*}(\mathbf{y}_{[3+\ell_0^*,n]})|\leqslant |D_{t-1}(\mathbf{x}_{[3,n]})\cap D_{t-1}(\mathbf{y}_{[3,n]})|,\\
|\mathcal{X}^{b}|&=|D_{t-1-\ell_1}(\mathbf{x}_{[3+\ell_1,n]})\cap D_{t-1-\ell_1^*}(\mathbf{y}_{[3+\ell_1^*,n]})|\leqslant |D_{t-2}(\mathbf{x}_{[4,n]})\cap D_{t-2}(\mathbf{y}_{[4,n]})|.
\end{align*}
Since $d_L(\mathbf{x,y})\geqslant 2$ and $x_1=y_1,x_2=y_2,x_3=y_3$, we have $d_L(\mathbf{x}_{[3,n]},\mathbf{y}_{[3,n]})\geqslant 2$ and $d_L(\mathbf{x}_{[4,n]},\mathbf{y}_{[4,n]})\geqslant 2$. Similarly, we also obtain Eqs. $(\ref{eq28})$ and $(\ref{eq29})$. 

Therefore, when $p=3$, for $n\geqslant 10$, by Eqs. $(\ref{eq27})$-$(\ref{eq30})$ we have
\begin{align*}
|&\mathcal{X}|=|\mathcal{X}^{x_1}|+|\mathcal{X}^a|+|\mathcal{X}^b|\\
&\leqslant \max\{M_1(n-1,t)+M_1(n-2,t-1)+M_1(n-3,t-2),M_1(n-1,t)+\max\{M_0(n-3,t-1),M_1(n-3,t-1)\}\\
&~~~~+2N_3(n-5,t-2,t-3,1)+M_0(n-3,t-2)\}\\
&\overset{(a)}{\leqslant} M_1(n-1,t)+M_1(n-2,t-1)+M_1(n-3,t-2)\overset{(b)}{=}M_1(n,t),
\end{align*}
where $(a)$ follows from Lemma $\ref{lm18}$, and $(b)$ follows from Eq. $(\ref{eq8})$.

For $p\in \{1,2,3\}$ and $n\geqslant10$, we have proved that $|\mathcal{X}|\leqslant M_1(n,t)$ with $2n\geqslant 3t+1$ which indicates the base of induction holds. Now for the induction step, assume the claim is true for cases of the structure of $p<s$ with $s\geqslant 4$ and $n\geqslant10$, and we look at $p=s$. For convenience, let $\mathbb{Z}_3=\{x_1,a,b\}$. By Lemmas $\ref{lm1}$ and $\ref{lm2}$, we have
$$|\mathcal{X}|=|\mathcal{X}^{x_1}|+|\mathcal{X}^{a}|+|\mathcal{X}^{b}|,$$
such that $\mathcal{X}^{x_1}=x_1\circ \big(D_t(\mathbf{x}_{[2,n]})\cap D_t(\mathbf{y}_{[2,n]})\big)$, $\mathcal{X}^{a}=a\circ\big(D_{t-1-\ell_0}(\mathbf{x}_{[3+\ell_0,n]})\cap D_{t-1-\ell_0^*}(\mathbf{y}_{[3+\ell_0^*,n]})\big)$, $\mathcal{X}^{b}=b\circ\big(D_{t-1-\ell_1}(\mathbf{x}_{[3+\ell_1,n]})\cap D_{t-1-\ell_1^*}(\mathbf{y}_{[3+\ell_1^*,n]})\big)$, where $\ell_0,\ell_0^*,\ell_1,\ell_1^*\geqslant 0$. Moreover, since $d_L(\mathbf{x,y})\geqslant 2$ and $x_1=y_1$, we have  $d_L(\mathbf{x}_{[2,n]},\mathbf{y}_{[2,n]})\geqslant 2$.
Thus, for $n\geqslant 11$, if $2(n-1)\geqslant 3t+1$, then by the assumption of $p=s-1$ 
\begin{align}
|\mathcal{X}^{x_1}|\leqslant  M_1(n-1,t),\label{eq31}
\end{align}
under the condition $\mathbf{x}_{[2,s]}=\mathbf{y}_{[2,s]}$ and $x_{s+1}\neq y_{s+1}$. 
If $2(n-1)\leqslant 3t$, then we have Eq. $(\ref{eq31})$. Furthermore, when $n=10$, we also have Eq. $(\ref{eq31})$ since $|\mathcal{X}^{x_1}|\leqslant N_3(9,2,t)=M_1(9,t)$.
Based on the value of $x_2$, we discuss the following two cases: \textbf{Case D1:} $x_2\neq x_1$; \textbf{Case D2:} $x_1=x_2$.

\textbf{Case D1:} When $x_2\neq x_1$, without loss of generality let $x_2=a$. Then $\ell_0=\ell_0^*=0$ and $\ell_1,\ell_1^*\geqslant 1$. Thus, for $n\geqslant12$, by using the method of proving Eq. $(\ref{eq28})$ and the assumption of $p=s-2$, we have
\begin{align}
|\mathcal{X}^{a}|&=D_{t-1}(\mathbf{x}_{[3,n]})\cap D_{t-1}(\mathbf{y}_{[3,n]})|\leqslant M_1(n-2,t-1)\label{eq32}.
\end{align}
For $n=11$ or $10$, we also have Eq. (\ref{eq32}).
Furthermore, for $n\geqslant13$,
\begin{align}
|\mathcal{X}^{b}|&=|D_{t-1-\ell_1}(\mathbf{x}_{[3+\ell_1,n]})\cap D_{t-1-\ell_1^*}(\mathbf{y}_{[3+\ell_1^*,n]})|\leqslant |D_{t-1-1}(\mathbf{x}_{[4,n]})\cap D_{t-1-1}(\mathbf{y}_{[4,n]})|\overset{(a)}{\leqslant} M_1(n-3,t-2),\label{eq33}
\end{align}
where $(a)$ follows from the assumption of $p=s-3$ under the conditions $d_L(\mathbf{x}_{[4,n]},\mathbf{y}_{[4,n]})\geqslant 2$. For $n=12, 11$, or $10$, we also have Eq. (\ref{eq33}).

\textbf{Case D2:} When $x_1=x_2$, then $\ell_0,\ell_0^*,\ell_1,\ell_1^*\geqslant 1$. Thus,
\begin{align*}
|\mathcal{X}^{a}|&=|D_{t-1-\ell_0}(\mathbf{x}_{[3+\ell_0,n]})\cap D_{t-1-\ell_0^*}(\mathbf{y}_{[3+\ell_0^*,n]})|\leqslant |D_{t-1}(\mathbf{x}_{[3,n]})\cap D_{t-1}(\mathbf{y}_{[3,n]})|,\\
|\mathcal{X}^{b}|&=|D_{t-1-\ell_1}(\mathbf{x}_{[3+\ell_1,n]})\cap D_{t-1-\ell_1^*}(\mathbf{y}_{[3+\ell_1^*,n]})|\leqslant |D_{t-2}(\mathbf{x}_{[4,n]})\cap D_{t-2}(\mathbf{y}_{[4,n]})|.
\end{align*}
For $2n\geqslant 3t+1$, if $2(n-2)\leqslant 3(t-1)$, by Lemma $\ref{lm16}$ we have $|\mathcal{X}^a|\leqslant 3^{n-t-1}=M_1(n-2,t-1)$. For $n\geqslant12$, if $2(n-2)\geqslant 3(t-1)+1$, then by the assumption of $p=s-2$ 
\begin{align}
|\mathcal{X}^{a}|\leqslant |D_{t-1}(\mathbf{x}_{[3,n]})\cap D_{t-1}(\mathbf{y}_{[3,n]})|\leqslant M_1(n-2,t-1).\label{eq34}
\end{align}
For $n=11$ or $10$, we also have Eq. (\ref{eq34}).
For $n\geqslant13$, since $2(n-3)\geqslant 3(t-2)+1$, by the assumption of $p=s-3$ 
\begin{align}                                    |\mathcal{X}^{b}|\leqslant |D_{t-2}(\mathbf{x}_{[4,n]})\cap D_{t-2}(\mathbf{y}_{[4,n]})|\leqslant M_1(n-3,t-2).\label{eq35}
\end{align}
For $n=12, 11$, or $10$, we also have Eq. (\ref{eq35}).

Therefore, for $p=s$, by Eqs. $(\ref{eq31})$-$(\ref{eq35})$ we have
\begin{align*}
|\mathcal{X}|=|\mathcal{X}^{x_1}|+|\mathcal{X}^a|+|\mathcal{X}^b|\leqslant M_1(n-1,t)+M_1(n-2,t-1)+M_1(n-3,t-2)\overset{(a)}{=}M_1(n,t),
\end{align*}
where $(a)$ follows from Eq. $(\ref{eq8})$. By the induction hypothesis, when $p\geqslant 1$ and $q\geqslant 1$, we have $|\mathcal{X}|\leqslant M_1(n,t)$. Combining Lemma $\ref{lm14}$, it follows that $N_3(n,2,t)=\max\{M_0(n,t),M_1(n,t)\}$ for $n\geqslant \max\{9,\lfloor \frac{3t}{2}\rfloor+1\}$ and $t\geqslant 2$.
\end{proof}

\section{Conclusion}
\label{sec5}

In this paper, we studied the sequence reconstruction problem over $3$-ary deletion channels. We determined the values of $N_3(n,2,t)$ for $t\geqslant 2$ and  $n\geqslant \max\{6,\lfloor \frac{3t}{2}\rfloor+1\}$. That is, $N_3(n,2,t)=\max\{M_0(n,t),M_1(n,t)\}$ with $M_0(n,t)=|D_t(\mathbf{x}_0)\cap D_t(\mathbf{y}_0)|$ and $M_1(n,t)=|D_t(\mathbf{x})\cap D_t(\mathbf{y})|$, where $\mathbf{x}_0=(0,1,2,\mathbf{a}_{n-5},\mathbf{a}_{n-3}(n-3),\mathbf{a}_{n-3}(n-4))$, $\mathbf{y}_0=(1,0,2,\mathbf{a}_{n-3})=(1,0,2,\mathbf{a}_{n-5},\mathbf{a}_{n-3}(n-4),\mathbf{a}_{n-3}(n-3))$, $\mathbf{x}=(0,1,2,0,1,2,\mathbf{a}_{n-6})$, and $\mathbf{y}=(1,0,2,1,0,2,\mathbf{a}_{n-6})$. Moreover, if $t=2,3,4,5$ and $n\geqslant \max\{6,\lfloor \frac{3t}{2}\rfloor+1\}$, or $n\geqslant 3t$, then $N_3(n,2,t)=M_1(n,t)$. When $t=2$ and $q\geqslant 4$, the values of $N_q(n,2,t)$ are not determined. For $q\geqslant 4$ and sufficiently large $n$, there is the following conjecture.
\begin{conjecture}
For $q\geqslant 4$ and sufficiently large $n$, we have
$$N_q(n,2,t)=|D_t(\mathbf{x})\cap D_t(\mathbf{y})|,$$
where $\mathbf{x}=(0,1,2,0,1,\mathbf{w})$, $\mathbf{y}=(1,0,2,1,0,\mathbf{w})$, and $\mathbf{w}=(w_1,...,w_{n-5})$ with $w_i\equiv i+1\mod q$.
\end{conjecture}

\appendices
\section{Proofs of Lemmas $\ref{lm6}$ and $\ref{lm7}$}\label{APP-A}
The purpose of this appendix is to give the proofs of Lemmas $\ref{lm6}$ and $\ref{lm7}$.

\begin{proof}{ \rm (The proof of Lemma $\ref{lm6}$)}
Now, by induction on $n$, we prove that $N_3(n,2,1,1)\leqslant 3$ for any $n\geqslant 4$. When $n=4$, by using a computerized search, we find that $N_3(4,2,1,1)=3$, which indicates that the base case of the induction holds. 

Assume that for some fixed $m\geqslant 4$, the proposition
\begin{equation}
N_3(m,2,1,1)\leqslant 3\nonumber
\end{equation}
holds.
When $n=m+1$, let $\mathbf{x}=(x_1,...,x_{m+2})$, $\mathbf{y}=(y_1,...,y_{m+1})$, and $\mathcal{X}=D_{2}(\mathbf{x})\cap D_{1}(\mathbf{y})$. 

If $x_1\neq y_1$, then $|\mathcal{X}^{x_1}|=|D_{2}(\mathbf{x}_{[2,m+2]})\cap D_{0-\ell_1^*}(\mathbf{y}_{[3+\ell_1^*,m+1]})|$, $ |\mathcal{X}^{y_1}|=|D_{1-\ell_2}(\mathbf{x}_{[3+\ell_2,m+2]})\cap D_{1}(\mathbf{y}_{[2,m+1]})|$, and $ |\mathcal{X}^{a}|=|D_{1-\ell_3}(\mathbf{x}_{[3+\ell_3,m+2]})\cap D_{0-\ell_3^*}(\mathbf{y}_{[3+\ell_3^*,m+1]})|$, where $a\in \mathbb{Z}_3\backslash \{x_1,y_1\}$,  $\ell_1^*,\ell_2,\ell_3,\ell_3^*\geqslant 0$. Moreover, $\ell_1^*=0$  implies $y_2=x_1$;  $\ell_2=0$ implies $x_2=y_1$; $\ell_3=0$ and $\ell_3^*=0$ imply $x_2=a$ and $y_2=a$, respectively. Thus,
\begin{align*}
|\mathcal{X}^{x_1}|&\leqslant |D_{0-\ell_1^*}(\mathbf{y}_{[3+\ell_1^*,m+1]})|\leqslant |D_{0}(\mathbf{y}_{[3,m+1]})|=1,\\
|\mathcal{X}^{y_1}|&\leqslant |D_{1-\ell_2}(\mathbf{x}_{[3+\ell_2,m+2]})\cap D_{1}(\mathbf{y}_{[2,m+1]})|\leqslant \max\{|D_{1}(\mathbf{x}_{[3,m+2]})\cap D_{1}(\mathbf{y}_{[2,m+1]})|,|D_{0}(\mathbf{x}_{[4,m+2]})|\}\leqslant 2,\\
|\mathcal{X}^{a}|&\leqslant |D_{0-\ell_3^*}(\mathbf{y}_{[3+\ell_3^*,m+1]})|\leqslant |D_{0}(\mathbf{y}_{[3,m+1]})|=1.
\end{align*}
Here, $\ell_2=0$ implies $x_2=y_1$ and $\mathbf{x}_{[3,m+2]}\neq \mathbf{y}_{[2,m+1]}$. This follows from  $\mathbf{y}\notin D_1(\mathbf{x})$, which leads to the bound $|D_{1}(\mathbf{x}_{[3,m+2]})\cap D_{1}(\mathbf{y}_{[2,m+1]})|\leqslant N_3(m,1,1)=2$. By Lemmas $\ref{lm1}$ and $\ref{lm2}$, it follows that $|\mathcal{X}|=|\mathcal{X}^{x_1}|+|\mathcal{X}^{y_1}|+|\mathcal{X}^a|\leqslant 4$, where $|\mathcal{X}|=4$ implies $\ell_1^*=0,\ell_2=0,\ell_3^*=0$. If $\ell_1^*=\ell_3^*=0$, then $y_2=x_1$ and $y_2=a$ which results in a contradiction since $x_1\neq a$. So, if $x_1\neq y_1$, then we have
\begin{equation}
|\mathcal{X}|=|\mathcal{X}^{x_1}|+|\mathcal{X}^{y_1}|+|\mathcal{X}^a|\leqslant 3.\nonumber
\end{equation}

If $x_1=y_1$, then $|\mathcal{X}^{x_1}|=|D_{2}(\mathbf{x}_{[2,m+2]})\cap D_{1}(\mathbf{y}_{[2,m+1]})|$, $ |\mathcal{X}^{a}|=|D_{1-\ell_2}(\mathbf{x}_{[3+\ell_2,m+2]})\cap D_{0-\ell_2^*}(\mathbf{y}_{[3+\ell_2^*,m+1]})|$, and $ |\mathcal{X}^{b}|=|D_{1-\ell_3}(\mathbf{x}_{[3+\ell_3,m+2]})\cap D_{0-\ell_3^*}(\mathbf{y}_{[3+\ell_3^*,m+1]})|$, where $\{a,b\}=\mathbb{Z}_3\backslash \{x_1\}$,  $\ell_2,\ell_2^*,\ell_3,\ell_3^*\geqslant 0$. Moreover,  $\ell_2=0$ implies $x_2=a$; $\ell_2^*=0$ implies $y_2=a$;  $\ell_3=0$ and $\ell_3^*=0$ imply $x_2=b$ and $y_2=b$, respectively. Consider the value of $|\mathcal{X}^{a}|$. Here, we have
\begin{align*}
|\mathcal{X}^{a}|&\leqslant |D_{0-\ell_2^*}(\mathbf{y}_{[3+\ell_2^*,m+1]})|\leqslant |D_{0}(\mathbf{y}_{[3,m+1]})|=1.
\end{align*}
If $\ell_2\geqslant 2$ or $\ell_2^*\geqslant 1$, then $|\mathcal{X}^a|=0$. Moreover, if $\ell_2=0~\text{or}~1$,  and $\ell_2^*=0$, we can verify that $|\mathcal{X}^{a}|=0$. Thus, 
%Thus, if $|\mathcal{X}^{a}|=1$, then we must have $\ell_2=0~\text{or}~1$,  and $\ell_2^*=0$. 
%When $\ell_2=0$ and $\ell_2^*=0$, we have $x_2=y_2=a$. Combining with $x_1=y_1$ and $d_L(\mathbf{x},\mathbf{y})\geqslant 1$, we have $\mathbf{y}_{[3,m+1]}\notin D_1(\mathbf{x}_{[3,m+2]})$, which makes $|\mathcal{X}^{a}|=0$. When $\ell_2=1$ and $\ell_2^*=0$, we have $x_3=a$ and $y_2=a$. If $|\mathcal{X}^{a}|=1$, then $\mathbf{x}_{[4,m+2]}=\mathbf{y}_{[3,m+1]}$, which makes $\mathbf{y}\in D_1(\mathbf{x})$ (i.e., $d_L(\mathbf{x,y})=0$) since $x_1=y_1$ and $x_3=y_2$. Thus, when  $\ell_2=1$ and $\ell_2^*=0$, we also have $|\mathcal{X}^{a}|=0$. So,
if $x_1=y_1$, then we have $|\mathcal{X}^{a}|=0$. Similarly, $|\mathcal{X}^{b}|=0$. Since $d_L(\mathbf{x,y})\geqslant 1$ and $x_1=y_1$, we have $d_L(\mathbf{x}_{[2,m+2]},\mathbf{y}_{[2,m+1]})\geqslant 1$. By the assumption of the length of sequences, we have $|\mathcal{X}^{x_1}|\leqslant 3$. If $x_1=y_1$, then we have
\begin{equation}
|\mathcal{X}|=|\mathcal{X}^{x_1}|+|\mathcal{X}^{a}|+|\mathcal{X}^b|\leqslant 3.\nonumber
\end{equation}

When $n=m+1$, the proposition
$$N_3(m+1,2,1,1)\leqslant 3$$
holds. By the inductive hypothesis, the proposition
$$N_3(n,2,1,1)\leqslant 3$$
holds for $n\geqslant 4$. 
\end{proof}

\begin{proof}{ \rm (The proof of Lemma $\ref{lm7}$)}
Let $\mathbf{x}\in \mathbb{Z}_3^{n+1}$ and $\mathbf{y}\in \mathbb{Z}_3^{n}$. Consider $t\geqslant \frac{2n}{3}$. The length of any subsequence of $\mathbf{y}$ obtained by deleting $t$ symbols is $n-t$, and $|D_t(\mathbf{y})|\leqslant 3^{n-t}$. Thus, $N_3(n,t+1,t,1)\leqslant 3^{n-t}$ for $t\geqslant \frac{2n}{3}$. Now we set $\mathbf{x}=(1,0,2,1,\mathbf{a}_{n-3})$ and $\mathbf{y}=(0,1,2,\mathbf{a}_{n-3})$.
Since the Levenshtein distance between  $\mathbf{x}_{[1,4]}$ and $\mathbf{y}_{[1,3]}$ is $1$, we have $d_L(\mathbf{x},\mathbf{y})\geqslant 1$.
For convenience, let $\mathcal{X}=D_{t+1}(\mathbf{x})\cap D_t(\mathbf{y})$. We verify that $|\mathcal{X}|=3^{n-t}$ for  $t\geqslant \frac{2n}{3}$. This shows that the bound is tight, so $N(n,t+1,t,1)= 3^{n-t}$ for $t\geqslant \frac{2n}{3}$.

Given $t<\frac{2n}{3}$, Lemmas $\ref{lm1}$ and $\ref{lm2}$, we have
%imply that $|\mathcal{X}|=|\mathcal{X}^0|+|\mathcal{X}^1|+|\mathcal{X}^2|$, where $$|\mathcal{X}^0|=|D_{t}(2,1,\mathbf{a}_{n-3})\cap D_{t}(1,2,\mathbf{a}_{n-3})|=N_3(n-1,1,t)=2D_3(n-3,t-1)+D_3(n-4,t-2),$$ $$|\mathcal{X}^1|=|D_{t+1}(0,2,1,\mathbf{a}_{n-3})\cap D_{t-1}(2,\mathbf{a}_{n-3})|=D_3(n-2,t-1),~\text{and}$$
%$$|\mathcal{X}^2|=|D_{t-1}(1,\mathbf{a}_{n-3})\cap D_{t-2}(\mathbf{a}_{n-3})|=D_3(n-3,t-2).$$ Thus, 
$|\mathcal{X}|=D_3(n-2,t-1)+2D_3(n-3,t-1)+D_3(n-3,t-2)+D_3(n-4,t-2)$. For $\frac{2n}{3}> t\geqslant 1$ and $n\geqslant 4$, we have $d_L(\mathbf{x},\mathbf{y})\geqslant 1$,  and $|D_3(\mathbf{x},\mathbf{y};t+1,t)|=D_3(n-2,t-1)+2D_3(n-3,t-1)+D_3(n-3,t-2)+D_3(n-4,t-2)$. Therefore, for $\frac{2n}{3}> t\geqslant 1$ and $n\geqslant 4$,
\begin{equation}
N_3(n,t+1,t,1)\geqslant |D_3(\mathbf{x},\mathbf{y};t+1,t)|=D_3(n-2,t-1)+2D_3(n-3,t-1)+D_3(n-3,t-2)+D_3(n-4,t-2).\label{eq36}
\end{equation}

Next, we prove that $N_3(\lfloor\frac{3t}{2}\rfloor+1,t+1,t,1)\leqslant N_3(\lfloor\frac{3t}{2}\rfloor+1,t)$ for any $t\geq 2$.  A direct computation shows that $N_3(4,3,2,1)\leqslant |D_3(4,2)|=6=N_3(4,2)$. When $t\geqslant 3$, let $m=\lfloor\frac{3t}{2}\rfloor+1$. Note that $m\geqslant t+2$ holds for $t\geqslant 3$. Moreover,
\begin{align*}
N_3(\lfloor\frac{3t}{2}\rfloor+1,t+1,t,1)&\leqslant D_3(m,t)
\overset{(a)}{=}D_3(m-1,t)+D_3(m-2,t-1)+D_3(m-3,t-2)\\
&\overset{(b)}{=}D_3(m-1,t)+D_3(m-3,t-1)+D_3(m-4,t-2)+D_3(m-5,t-3)+D_3(m-3,t-2),
\end{align*}
where $(a)$ and $(b)$ follow from Eq. $(\ref{eq3})$ for $m\geqslant t+2$. 

If $t=2k+1$ and $k\geqslant 1$, then $m=3k+2$. By Eq. $(\ref{eq4})$ and the conditions $2(3k+1)\leqslant 3(2k+1)$ and $2(3k)=3(2k)$, it follows that
 $D_3(m-1,t)=D_3(3k+1,2k+1)=3^{k}=D_3(3k,2k)=D_3(m-2,t-1)$.   Similarly, under $2(3k-1)\leqslant 3(2k)$ and $2(3k-3)=3(2k-2)$, we have $D_3(m-3,t-1)=D_3(3k-1,2k)=3^{k-1}=D_3(3k-3,2k-2)=D_3(m-5,t-3)$. Therefore, for $t=2k+1$, 
\begin{align*}
N_3(\lfloor\frac{3t}{2}\rfloor+1,t+1,t,1)&\leqslant D_3(m,t)=D_3(m-2,t-1)+2D_3(m-3,t-1)+D_3(m-4,t-2)+D_3(m-3,t-2)\\
&=N_3(m,t).
\end{align*}

If $t=2k$ and $k\geqslant 2$, then $m=3k+1$. By Eq. $(\ref{eq4})$, $D_3(m-1,t)=D_3(3k,2k)=3^{k}$ and $D_3(m-3,t-1)=D_3(3k-2,2k-1)=3^{k-1}$. The values $D_3(m-2,t-1)=D_3(3k-1,2k-1)=3^{k}-1$ and $D_3(m-5,t-3)=D_3(3k-4,2k-3)=3^{k-1}-1$  follow from the explicit formula  $D_3(m,t)=\sum\limits_{i=0}^{t}\binom{m-t}{i}\sum\limits_{j=0}^{t-i}\binom{i}{j}$. Therefore, 
for $t=2k$, 
\begin{align*}
N_3(\lfloor\frac{3t}{2}\rfloor+1,t+1,t,1)&\leqslant D_3(m,t)=D_3(m-2,t-1)+2D_3(m-3,t-1)+D_3(m-4,t-2)+D_3(m-3,t-2)\\
&=N_3(m,t).
\end{align*}
Thus, for any $t\geqslant 2$,
$$N_3(\lfloor\frac{3t}{2}\rfloor+1,t+1,t,1)\leqslant N_3(\lfloor\frac{3t}{2}\rfloor+1,t).$$

By induction on $n$ and $t$, we prove that $N_3(n,t+1,t,1)\leqslant N_3(n,t)$ for $n\geqslant \max\{\lfloor\frac{3t}{2}\rfloor+1,4\}$ and $t\geq 1$. When $t=1$ or $n=\lfloor\frac{3t}{2}\rfloor+1$, by Lemma $\ref{lm6}$ and the above discussion, we have $N_3(n,t+1,t,1)\leqslant N_3(n,t)$.
When $(n,t)=(6,3)$ or $(5,2)$, it is easily verified that $N_3(6,4,3,1)=17$ and $N_3(5,3,2,1)=9$, which satisfies the inequality $N_3(n,t+1,t,1)\leqslant N_3(n,t)$. These results indicate that the base case of the induction holds. 

Assume that for some fixed $t=k+1$ and $m\geqslant \lfloor \frac{3t}{2} \rfloor +1$ with $k\geqslant 1$, the following inequalities hold:
\begin{align}
N_3(m,k+2,k+1,1)&\leqslant D_3(m-2,k)+2D_3(m-3,k)+D_3(m-3,k-1)+D_3(m-4,k-1),\nonumber\\
N_3(m-1,k+1,k,1)&\leqslant D_3(m-3,k-1)+2D_3(m-4,k-1)+D_3(m-4,k-2)+D_3(m-5,k-2),\nonumber\\
N_3(m-2,k,k-1,1)&\leqslant D_3(m-4,k-2)+2D_3(m-5,k-2)+D_3(m-5,k-3)+D_3(m-6,k-3),\nonumber
\end{align} 
where $2(m-1)>3k$, $2(m-2)>3(k-1)$, and the last inequality also holds for $k=1$ since both sides equal zero in this case. Consider $n=m+1$, we will prove that $N_3(m+1,k+2,k+1,1)\leqslant D_3(m-1,k)+2D_3(m-2,k)+D_3(m-2,k-1)+D_3(m-3,k-1)$. For convenience, let $\mathbf{x}=(x_1,...,x_{m+2})$, $\mathbf{y}=(y_1,...,y_{m+1})$, and define $\mathcal{X}=D_{k+2}(\mathbf{x})\cap D_{k+1}(\mathbf{y})$. Since $t\geqslant 2$, it follows that $m\geqslant t+2=k+3$. Based on the relationship between $x_1$ and $y_1$, we analyze the following two cases: \textbf{Case A}: $x_1=y_1$; \textbf{Case B}: $x_1\neq y_1$. 

\textbf{Case A}: If $x_1=y_1\in \mathbb{Z}_3$, then $|\mathcal{X}^{x_1}|=|D_{k+2}(\mathbf{x}_{[2,m+2]})\cap D_{k+1}(\mathbf{y}_{[2,m+1]})|$, $ |\mathcal{X}^{a}|=|D_{k+1-\ell_1}(\mathbf{x}_{[3+\ell_1,m+2]})\cap D_{k-\ell_1^*}(\mathbf{y}_{[3+\ell_1^*,m+1]})|$,  and $|\mathcal{X}^{b}|=|D_{k+1-\ell_2}(\mathbf{x}_{[3+\ell_2,m+2]})\cap D_{k-\ell_2^*}(\mathbf{y}_{[3+\ell_2^*,m+1]})|$,
where $\ell_1,\ell_1^*,\ell_2,\ell_2^*\geqslant 0$, and $\{a,b\}=\mathbb{Z}_3\backslash \{x_1\}$. Moreover, $\ell_1~ \text{or}~\ell_1^*=0$ implies $x_2=a$ or $y_2=a$, respectively; similarly, $\ell_2~ \text{or}~\ell_2^*=0$ implies $x_2=b$ or $y_2=b$, respectively.
By the induction hypothesis, we have
\begin{equation}
|\mathcal{X}^{x_1}|\leqslant N_3(m,k+2,k+1,1)\leqslant D_3(m-2,k)+2D_3(m-3,k)+D_3(m-3,k-1)+D_3(m-4,k-1)\nonumber.
\end{equation}
It is easily verified that $\ell_1\neq \ell_2$ and $\ell_1^*\neq \ell_2^*$. Without loss of generality, assume $\ell_1^*<\ell_2^*$. Based on the values of $\ell_1^*$ and $\ell_2^*$,  we consider the following three cases: \textbf{Case A1:} $\ell_1^*\geqslant 1$; \textbf{Case A2:} $\ell_1^*=0, \ell_2^*\geqslant 2$; \textbf{Case A3:} $\ell_1^*=0, \ell_2^*=1$.   

\textbf{Case A1:} Given that $\ell_1^*\geqslant 1$, it follows that
\begin{align}
|\mathcal{X}^{a}|&\leqslant |D_{k-\ell_1^*}(\mathbf{y}_{[3+\ell_1^*,m+1]})|\leqslant |D_3(m-2,k-1)|\nonumber\\
&\overset{(a)}{=}D_3(m-3,k-1)+D_3(m-4,k-2)+D_3(m-5,k-3)\nonumber\\
&\overset{(b)}{\leqslant}D_3(m-3,k-1)+D_3(m-4,k-2)+D_3(m-4,k-2)\nonumber\\
&\overset{(c)}{\leqslant}D_3(m-3,k-1)+D_3(m-4,k-2)+3D_3(m-5,k-2)\nonumber\\
&\overset{(d)}{\leqslant}D_3(m-3,k-1)+2D_3(m-4,k-1)+D_3(m-4,k-2)+D_3(m-5,k-2),\nonumber
\end{align}
where $(a)$ follows from Eq. $(\ref{eq3})$ under the condition $m\geqslant k+2$, $(b),(d)$ follow from Eq. $(\ref{eq6})$, $(c)$ follows from Eq. $(\ref{eq9})$ under the condition $m\geqslant k+3$.
Similarly, we obtain
\begin{align}
|\mathcal{X}^{b}|&\leqslant |D_{k-\ell_2^*}(\mathbf{y}_{[3+\ell_2^*,m+1]})|\leqslant |D_3(m-3,k-2)|\nonumber\\
&\leqslant D_3(m-4,k-2)+2D_3(m-5,k-2)+D_3(m-5,k-3)+D_3(m-6,k-3).\nonumber
\end{align}
Thus, we have
\begin{align}
|\mathcal{X}|&=|\mathcal{X}^{x_1}|+|\mathcal{X}^{a}|+|\mathcal{X}^{b}|\nonumber\\
&\leqslant \big(D_3(m-2,k)+D_3(m-3,k-1)+D_3(m-4,k-2)\big)+2\big(D_3(m-3,k)+D_3(m-4,k-1)+D_3(m-5,k-2)\big)\nonumber\\
& +\big(D_3(m-3,k-1)+D_3(m-4,k-2)+D_3(m-5,k-3)\big)+\big(D_3(m-4,k-1)+D_3(m-5,k-2)+D_3(m-6,k-3)\big)\nonumber\\
&=D_3(m-1,k)+2D_3(m-2,k)+D_3(m-2,k-1)+D_3(m-3,k-1)=N_3(m+1,k+1).\label{eq37}
\end{align}

\textbf{Case A2:} For $\ell_1^*=0$ and $\ell_2^*\geqslant 2$, 
 we analyze this case in three subcases: \textbf{Case A2-1}: $\ell_1=0$; \textbf{Case A2-2}: $\ell_1=1$; \textbf{Case A2-3}: $\ell_1\geqslant 2$. 

First, since $\ell_2^*\geqslant 2$, we have
\begin{align}
|\mathcal{X}^{b}|&\leqslant |D_{k-\ell_2^*}(\mathbf{y}_{[3+\ell_2^*,m+1]})|\leqslant |D_3(m-3,k-2)|\nonumber\\
&\leqslant D_3(m-4,k-2)+2D_3(m-5,k-2)+D_3(m-5,k-3)+D_3(m-6,k-3).\nonumber
\end{align}

\textbf{Case A2-1}: When $\ell_1^*=0$, $\ell_2^*\geqslant 2$, $\ell_1=0$, we have $\mathbf{y}_{[3,m+1]}\notin D_1(\mathbf{x}_{[3,m+2]})$ under the conditions $d_L(\mathbf{x,y})\geqslant 1$, $x_1=y_1$, and $x_2=y_2=a$. Thus, $|\mathbf{y}_{[3,m+1]}|=m-1$ and $2(m-1)\geqslant (3t-1)=3k+2$. By the induction hypothesis,
\begin{align}
|\mathcal{X}^a|=|D_{k+1}(\mathbf{x}_{[3,m+2]})\cap D_{k}(\mathbf{y}_{[3,m+1]})|&\leqslant N_3(m-1,k+1,k,1)\nonumber\\
&\leqslant D_3(m-3,k-1)+2D_3(m-4,k-1)+D_3(m-4,k-2)+D_3(m-5,k-2).\nonumber
\end{align}
Combining these results, we obtain $$|\mathcal{X}|\leqslant D_3(m-1,k)+2D_3(m-2,k)+D_3(m-2,k-1)+D_3(m-3,k-1).$$

\textbf{Case A2-2}: When $\ell_1^*=0$, $\ell_2^*\geqslant 2$, $\ell_1=1$, we have $\mathbf{y}_{[3,m]}\neq\mathbf{x}_{[4,m+1]}$ under the conditions $d_L(\mathbf{x,y})\geqslant 1$, $x_1=y_1$, and $x_3=y_2=a$. Thus, for $m\geqslant k+3$, by Eqs. $(\ref{eq3})$, $(\ref{eq5})$, and $(\ref{eq6})$ we have
\begin{align}
|\mathcal{X}^a|&=|D_{k+1}(\mathbf{x}_{[4,m+2]})\cap D_{k}(\mathbf{y}_{[3,m+1]})|\leqslant N_3(m-1,1,k)=2D_3(m-3,k-1)+D_3(m-4,k-2)\nonumber\\
&\leqslant D_3(m-3,k-1)+2D_3(m-4,k-1)+D_3(m-4,k-2)+D_3(m-5,k-2).\nonumber
\end{align}
Consequently, we also obtain Eq. $(\ref{eq37})$.

\textbf{Case A2-3}: Given that $\ell_1^*=0$, $\ell_2^*\geqslant 2$, $\ell_1\geqslant 2$, it follows that  
\begin{align*}
|\mathcal{X}^a|&\leqslant |D_{k+1-\ell_1}(\mathbf{x}_{[3+\ell_1,m+2]})|\leqslant D_3(m-2,k-1)\\
&\leqslant D_3(m-3,k-1)+2D_3(m-4,k-1)+D_3(m-4,k-2)+D_3(m-5,k-2),
\end{align*}
from Eqs. $(\ref{eq3})$ and $(\ref{eq6})$ when $m\geqslant k+3$. Consequently, we also obtain Eq. $(\ref{eq37})$.

Therefore, when $\ell_1^*=0$ and $\ell_2^*\geqslant 2$, we have
$$|\mathcal{X}|\leqslant D_3(m-1,k)+2D_3(m-2,k)+D_3(m-2,k-1)+D_3(m-3,k-1)=N_3(m+1,k+1).$$

\textbf{Case A3:} For $\ell_1^*=0$ and $\ell_2^*=1$, we analyze this scenario in four distinct cases: \textbf{Case A3-1}: $\ell_1=0$ and $\ell_2=1$; \textbf{Case A3-2}: $\ell_1=0$ and $\ell_2\geqslant 2$; \textbf{Case A3-3}: $\ell_1=1$; \textbf{Case A3-4}: $\ell_1\geqslant 2$.

\textbf{Case A3-1}: When $\ell_1^*=\ell_1=0$ and $\ell_2^*=\ell_2=1$, we have $d_L(\mathbf{x}_{[3,m+2]},\mathbf{y}_{[3,m+1]})\geqslant 1$ and $d_L(\mathbf{x}_{[4,m+2]},\mathbf{y}_{[4,m+1]})\geqslant 1$. Given that $2(m-1)\geqslant 3k+1 $, our assumption implies
\begin{align}
|\mathcal{X}^a|&=|D_{k+1}(\mathbf{x}_{[3,m+2]})\cap D_{k}(\mathbf{y}_{[3,m+1]})|\leqslant N_3(m-1,k+1,k,1)\nonumber\\
&\leqslant D_3(m-3,k-1)+2D_3(m-4,k-1)+D_3(m-4,k-2)+D_3(m-5,k-2).\nonumber
\end{align}
Similarly, since $2(m-2)\geqslant 3(k-1)+1 $, our assumption yields
\begin{align*}
|\mathcal{X}^b|&=|D_{k}(\mathbf{x}_{[4,m+2]})\cap D_{k-1}(\mathbf{y}_{[4,m+1]})|\leqslant N_3(m-2,k,k-1,1)\\
&\leqslant D_3(m-4,k-2)+2D_3(m-5,k-2)+D_3(m-5,k-3)+D_3(m-6,k-3).
\end{align*}
For the special case $k=1$, we have $|\mathcal{X}^b|=0$, which is consistent with the bound on $|\mathcal{X}^b|$. Consequently, we also obtain Eq. $(\ref{eq37})$.

\textbf{Case A3-2}: Consider the case where $\ell_1^*=\ell_1=0$, $\ell_2^*=1$, and $\ell_2\geqslant 2$. Under these conditions, we have $d_L(\mathbf{x}_{[3,m+2]},\mathbf{y}_{[3,m+1]})\geqslant 1$. By the assumption, it follows that 
\begin{align}
|\mathcal{X}^a|\leqslant D_3(m-3,k-1)+2D_3(m-4,k-1)+D_3(m-4,k-2)+D_3(m-5,k-2).\nonumber
\end{align}
If $\ell_2=2$, then $\mathbf{x}_{[5,m+2]}\neq \mathbf{y}_{[4,m+1]}$ under the conditions $x_1=y_1$, $x_2=y_2$, $x_4=y_3$, and $d_L(\mathbf{x,y})\geqslant 1$. For $m\geqslant k+3$, by Eqs. $(\ref{eq3})$ and $(\ref{eq5})$ we have
\begin{align}
|\mathcal{X}^b|&=|D_{k-1}(\mathbf{x}_{[5,m+2]})\cap D_{k-1}(\mathbf{y}_{[4,m+1]})|\leqslant N_3(m-2,1,k-1)\nonumber\\
&\leqslant D_3(m-4,k-2)+2D_3(m-5,k-2)+D_3(m-5,k-3)+D_3(m-6,k-3).\nonumber
\end{align}
If $\ell_2\geqslant 3$, then
\begin{align}
|\mathcal{X}^b|&\leqslant |D_{k+1-\ell_2}(\mathbf{x}_{[3+\ell_2,m+2]})|\leqslant |D_{k-2}(\mathbf{x}_{[6,m+2]})|\leqslant D_3(m-3,k-2)\nonumber\\
&\leqslant D_3(m-4,k-2)+2D_3(m-5,k-2)+D_3(m-5,k-3)+D_3(m-6,k-3),\nonumber
\end{align}
from Eq. $(\ref{eq3})$ under the condition $m\geqslant k+3$.
Therefore, we obtain Eq. $(\ref{eq37})$.

\textbf{Case A3-3}: When $\ell_1^*=0$, $\ell_1=1$, then $d_L(\mathbf{x}_{[4,m+2]},\mathbf{y}_{[3,m+1]})\geqslant 1$. Thus, 
\begin{equation}
|\mathcal{X}^a|=|D_{k}(\mathbf{x}_{[4,m+2]})\cap D_k(\mathbf{y}_{[3,m+1]})|\leqslant N_3(m-1,1,k)=2D_3(m-3,k-1)+D_3(m-4,k-2).\nonumber
\end{equation}
For $\ell_2^*=1$, we have
\begin{align}
|\mathcal{X}^b|=|D_{k+1-\ell_2}(\mathbf{x}_{[3+\ell_2,m+2]})\cap D_{k-1}(\mathbf{y}_{[4,m+1]})|\leqslant |D_{k-1}(\mathbf{y}_{[4,m+1]})|\leqslant D_3(m-2,k-1).\nonumber
\end{align}
For $m\geqslant k+3$, by Eq. $(\ref{eq3})$, we have
\begin{align}
N_3(m-1,k)+N_3(m-2,k-1)-\big(2D_3(m-3,k-1)+D_3(m-4,k-2)+D_3(m-2,k-1)\big)\nonumber\\
=D_3(m-5,k-2)-D_3(m-6,k-3)\geqslant 0.\nonumber
\end{align}
Thus, it follows that
\begin{equation}
|\mathcal{X}^a|+|\mathcal{X}^b|\leqslant N_3(m-1,k)+N_3(m-2,k-1).\nonumber
\end{equation}
Combining with the value of $|\mathcal{X}^{x_1}|$, we obtain Eq. $(\ref{eq37})$.

\textbf{Case A3-4}: Consider the case where $\ell_1^*=0$, $\ell_2^*=1, \ell_1\geqslant 2$. In this scenario, we have
\begin{equation}
|\mathcal{X}^a|\leqslant|D_{k+1-\ell_1}(\mathbf{x}_{[3+\ell_1,m+2]})|\leqslant |D_{k-1}(\mathbf{x}_{[5,m+2]})|\leqslant D_3(m-2,k-1)\leqslant 2D_3(m-3,k-1)+D_3(m-4,k-2).\nonumber
\end{equation}
For the case where $\ell_2^*=1$, we have
\begin{align}
|\mathcal{X}^b|=|D_{k+1-\ell_2}(\mathbf{x}_{[3+\ell_2,m+2]})\cap D_{k-1}(\mathbf{y}_{[4,m+1]})|\leqslant |D_{k-1}(\mathbf{y}_{[4,m+1]})|\leqslant D_3(m-2,k-1).\nonumber
\end{align}
Applying a similar method as in \textbf{Case A3-3},  we also obtain Eq. $(\ref{eq37})$.

Based on the preceding analysis involving $\ell_1,\ell_2,\ell_1^*,\ell_2^*$, for $x_1=y_1$, it follows that
\begin{align}
|\mathcal{X}|\leqslant D_3(m-1,k)+2D_3(m-2,k)+D_3(m-2,k-1)+D_3(m-3,k-1)=N_3(m+1,k+1).\nonumber
\end{align}

\textbf{Case B:} Given that $x_1\neq y_1$, it follows that $|\mathcal{X}^{x_1}|=|D_{k+2}(\mathbf{x}_{[2,m+2]})\cap D_{k-\ell_0^*}(\mathbf{y}_{[3+\ell_0^*,m+1]})|$, $ |\mathcal{X}^{y_1}|=|D_{k+1-\ell_1}(\mathbf{x}_{[3+\ell_1,m+2]})\cap D_{k+1}(\mathbf{y}_{[2,m+1]})|$,  and $|\mathcal{X}^{c}|=|D_{k+1-\ell_2}(\mathbf{x}_{[3+\ell_2,m+2]})\cap D_{k-\ell_2^*}(\mathbf{y}_{[3+\ell_2^*,m+1]})|$,
where $\ell_1,\ell_0^*,\ell_2,\ell_2^*\geqslant 0$, and $\{c\}=\mathbb{Z}_3\backslash \{x_1,y_1\}$. Furthermore,
it can be verified that $\ell_1\neq \ell_2$ and $\ell_0^*\neq \ell_2^*$.  Based on the value of $\ell_1$, we analyze in the following two cases: \textbf{Case B1:} $\ell_1=0$; \textbf{Case B2:} $\ell_1\geqslant 1$. 

 First, it follows that
\begin{align}
|\mathcal{X}^{x_1}|&\leqslant |D_{k-\ell_0^*}(\mathbf{y}_{[3+\ell_0^*,m+1]})|\leqslant D_3(m-1-\ell_0^*,k-\ell_0^*)\nonumber,\\
|\mathcal{X}^{c}|&\leqslant |D_{k-\ell_2^*}(\mathbf{y}_{[3+\ell_2^*,m+1]})|\leqslant D_3(m-1-\ell_2^*,k-\ell_2^*)\nonumber,
\end{align}
where $\ell_0^*,\ell_2^*\geqslant 0$ and $\ell_0^*\neq \ell_2^*$. By Eq. $(\ref{eq6})$, we have
\begin{equation}
|\mathcal{X}^{x_1}|+|\mathcal{X}^{c}|\leqslant D_3(m-1,k)+D_3(m-2,k-1).\nonumber
\end{equation}

\textbf{Case B1:} Given that $\ell_1=0$, it follows that $\mathbf{x}_{[3,m+2]}\neq \mathbf{y}_{[2,m+1]}$ since $d_L(\mathbf{x,y})\geqslant 1, x_2=y_1,$ and $x_1\neq y_1$. Thus, we obtain
\begin{align}
|\mathcal{X}^{y_1}|&\leqslant |D_{k+1}(\mathbf{x}_{[3,m+2]})\cap D_{k+1}(\mathbf{y}_{[2,m+1]})|\leqslant N_3(m,1,k+1)\overset{(a)}{=}2D_3(m-2,k)+D_3(m-3,k-1),\nonumber
\end{align}
where $(a)$ follows from Eq. $(\ref{eq5})$ under the condition $m\geqslant k+3$.

Thus, for $\ell_1=0$,
\begin{align}
|\mathcal{X}|&=|\mathcal{X}^{x_1}|+|\mathcal{X}^{y_1}|+|\mathcal{X}^{c}|\nonumber\\
&\leqslant D_3(m-1,k)+2D_3(m-2,k)+D_3(m-2,k-1)+D_3(m-3,k-1)=N_3(m+1,k+1).\nonumber
\end{align}
\textbf{Case B2:} Given that $\ell_1\geqslant 1$, we obtain
\begin{align}
|\mathcal{X}^{y_1}|&\leqslant |D_{k+1-\ell_1}(\mathbf{x}_{[3+\ell_1,m+2]})|\leqslant D_3(m-1,k)\overset{(a)}{\leqslant}2D_3(m-2,k)+D_3(m-3,k-1),\nonumber
\end{align}
where $(a)$ follows from Eqs. $(\ref{eq3})$ and $(\ref{eq6})$.
Therefore, for $\ell_1\geqslant 1$, we also have $$|\mathcal{X}|\leqslant D_3(m-1,k)+2D_3(m-2,k)+D_3(m-2,k-1)+D_3(m-3,k-1).$$ 

From the above analysis, we conclude that for $x_1\neq y_1$, 
\begin{align}
|\mathcal{X}|\leqslant D_3(m-1,k)+2D_3(m-2,k)+D_3(m-2,k-1)+D_3(m-3,k-1)=N_3(m+1,k+1).\nonumber
\end{align}
Thus, we have established the proposition
$$N_3(m+1,k+2,k+1,1)\leqslant D_3(m-1,k)+2D_3(m-2,k)+D_3(m-2,k-1)+D_3(m-3,k-1)$$ for $t=k+1$ and $n=m+1$.

By the inductive hypothesis, we have proved that
$N_3(n,t+1,t,1)\leqslant D_3(n-2,t-1)+2D_3(n-3,t-1)+D_3(n-3,t-2)+D_3(n-4,t-2)$ holds for any $\frac{2n}{3}> t\geqslant 1$. 
Furthermore, combining with Eq.~$(\ref{eq36})$, we conclude that $N_3(n,t+1,t,1)= D_3(n-2,t-1)+2D_3(n-3,t-1)+D_3(n-3,t-2)+D_3(n-4,t-2)$ holds for any $\frac{2n}{3}> t\geqslant 1$.
\end{proof}

\section{Proof of Lemma $\ref{lm9}$ }\label{APP-B}
The purpose of this appendix is to give the proof of Lemma $\ref{lm9}$.

\begin{proof}{ \rm (The proof of Lemma $\ref{lm9}$)}
If $|\mathbf{x}|=0$, then $(\mathbf{x},a,a,\mathbf{y})=(a,a,\mathbf{y})$. For convenience, let $\mathbf{y}=(y_1,y_2,...,y_n)$. By Lemma $\ref{lm1}$, 
\begin{align*}
|D_t(a,a,\mathbf{y})|&=|D_t(a,a,\mathbf{y})^a|+|D_t(a,a,\mathbf{y})^{a+1}|+|D_t(a,a,\mathbf{y})^{a+2}|\\
&=|D_t(a,\mathbf{y})|+|D_{t-2-\ell_1}(y_{2+\ell_1},...,y_n)|+|D_{t-2-\ell_2}(y_{2+\ell_2},...,y_n)|,
\end{align*}
where $\ell_1,\ell_2\geqslant 0$, $y_{1+\ell_1}=a+1$, and $y_{1+\ell_2}=a+2$. Since $c\neq 0$, we have $\{c,2c\}=\{1,2\}$. Thus,
\begin{align*}
|D_t(a,c+a,c+\mathbf{y})|&=|D_t(a,c+a,c+\mathbf{y})^{a}|+|D_t(a,c+a,c+\mathbf{y})^{a+c}|+|D_t(a,c+a,c+\mathbf{y})^{a+2c}|\\
&=|D_{t}(c+a,c+\mathbf{y})|+|D_{t-1}(c+\mathbf{y})|+|D_{t-2-\ell_3}(c+y_{2+\ell_3},...,c+y_n)|\\
&\overset{(a)}{=}|D_{t}(a,\mathbf{y})|+|D_{t-1}(\mathbf{y})|+|D_{t-2-\ell_3}(y_{2+\ell_3},...,y_n)|\\
&\overset{(b)}{\geqslant}|D_t(a,\mathbf{y})|+|D_{t-2-\ell_2}(y_{2+\ell_2},...,y_n)|+|D_{t-2-\ell_1}(y_{2+\ell_1},...,y_n)|=|D_t(a,a,\mathbf{y})|,
\end{align*}
where $\ell_3=\ell_1$ for $c=1$, or $\ell_3=\ell_2$ for $c=2$. Here, $(a)$ follows from Lemma $\ref{lm8}$, and $(b)$ holds because $(y_{2+\ell_1},...,y_n)$ or $(y_{2+\ell_2},...,y_n)$ is a subsequence of $\mathbf{y}$.

For convenience, let $\mathbf{x}=(x_1,...,x_m)$ and define the concatenated sequence $(\mathbf{x},a,a,\mathbf{y})=(x_1,x_2,...,x_m,x_{m+1},x_{m+2},\\x_{m+3},...,x_{m+n+2})$, where $x_{m+1}=x_{m+2}=a$ and $x_{m+2+i}=y_i$ for all $i\in [n]$. We define the parameter $c$ as follows:
\begin{itemize}
    \item If $x_m = a + 1$, then $c = 2$.
    \item If $x_m = a + 2$, then $c = 1$.
    \item If $x_m = a$, then we set $c = 1$ when $\mathbf{x}$ can be expressed as $\mathbf{x} = \mathbf{u} \circ (a+2) \circ a^t \text{ for some } t \geq 1 \text{ and } \mathbf{u} \in \mathbb{Z}_3^{m-t-1}$; otherwise, we set $c=2$.
\end{itemize}
Similarly, define the modified sequence $(\mathbf{x},a,c+a,c+\mathbf{y})=(x'_1,x'_2,...,x'_m,x'_{m+1},x'_{m+2},x'_{m+3},...,x'_{m+n+2})$, where $x'_{m+1}=a, x'_{m+2}=c+a$, $x'_{j}=x_j$ for all $j\in [m]$, $x'_{m+2+i}=c+y_i$ for all $i\in [n]$. By Lemma $\ref{lm1}$,
\begin{align*}
|D_t(\mathbf{x},a,a,\mathbf{y})|&=|D_t(\mathbf{x},a,a,\mathbf{y})^{x_1}|+|D_t(\mathbf{x},a,a,\mathbf{y})^{x_1+1}|+|D_t(\mathbf{x},a,a,\mathbf{y})^{x_1+2}|\\
&=|D_t(\mathbf{x}_{[2,m]},a,a,\mathbf{y})|+|D_{t-1-\ell_1}(x_{3+\ell_1},...,x_{m+2+n})|+|D_{t-1-\ell_2}(x_{3+\ell_2},...,x_{m+2+n})|,
\end{align*}
where $\ell_1, \ell_2 \geqslant 0$ are the smallest indices such that $x_{2+\ell_1}=x_1+1$ and $x_{2+\ell_2}=x_1+2$. Similarly, applying Lemma~$\ref{lm1}$ to the modified sequence yields
\begin{align*}
|D_t(\mathbf{x},a,c+a,c+\mathbf{y})|&=|D_t(\mathbf{x},a,c+a,c+\mathbf{y})^{x_1}|+|D_t(\mathbf{x},a,c+a,c+\mathbf{y})^{x_1+1}|+|D_t(\mathbf{x},a,c+a,c+\mathbf{y})^{x_1+2}|\\
&=|D_t(\mathbf{x}_{[2,m]},a,c+a,c+\mathbf{y})|+|D_{t-1-\ell_1^*}(x'_{3+\ell_1^*},...,x'_{m+2+n})|+|D_{t-1-\ell_2^*}(x'_{3+\ell_2^*},...,x'_{m+2+n})|,
\end{align*}
where $\ell_1^*, \ell_2^* \geqslant 0$ are the smallest indices such that
$x'_{2+\ell_1^*}=x_1+1$ and $x'_{2+\ell_2^*}=x_1+2$.

Next, we prove that $|D_t(\mathbf{x},a,a,\mathbf{y})|\leqslant |D_t(\mathbf{x},a,c+a,c+\mathbf{y})|$ by  induction on the length and the last element of $\mathbf{x}$.  

When $m=1$, if $x_m=a$, then $c=2$ such that  $\ell_1^*=\ell_2$, $\ell_2^*=1$, and $\ell_1\geqslant 2$. Moreover, we have
\begin{align*}
|D_t(\mathbf{x},a,c+a,c+\mathbf{y})|&=|D_t(\mathbf{x}_{[2,m]},a,c+a,c+\mathbf{y})|+|D_{t-1-\ell_1^*}(x'_{3+\ell_1^*},...,x'_{m+2+n})|+|D_{t-1-\ell_2^*}(x'_{3+\ell_2^*},...,
x'_{m+2+n})|\\
&=|D_t(a,c+a,c+\mathbf{y})|+|D_{t-1-\ell_2}(2+x_{3+\ell_2},...,2+x_{m+2+n})|+|D_{t-2}(2+\mathbf{y})|\\
&\geqslant |D_t(a,a,\mathbf{y})|+|D_{t-1-\ell_2}(x_{3+\ell_2},...,x_{m+2+n})|+|D_{t-1-\ell_1}(x_{3+\ell_1},...,x_{m+2+n})|=|D_t(\mathbf{x},a,a,\mathbf{y})|.
\end{align*}
Similarly, when $x_m=a+1$ or $x_m=a+2$, the same result holds. Thus, for $m=1$ and $x_1\in \{a,a+1,a+2\}$, the base case for the induction holds.

For the induction step, assume the claim holds for all cases where $|\mathbf{x}|<m$. We now consider the case $|\mathbf{x}|=m$  based on the last element of $\mathbf{x}$. When $x_m=a+1$ or $a+2$, $\mathbf{x}_{[k,m]}$ yields the same value $c$ for any $k\in [m]$. Moreover, it can be verified that $\{x'_i|i\in[m+2]\}=\mathbb{Z}_3$  for $x_m=a+1$ or $a+2$. When $x_m=a$, $\mathbf{x}_{[k,m]}$ also yields the same value $c$ by the structure of $\mathbf{x}$. Based on the value of $x_m$, we analyze the following three cases: \textbf{Case A:} $x_m=a+1$; \textbf{Case B:} $x_m=a+2$; \textbf{Case C:} $x_m=a$. 

\textbf{Case A:} Given $x_m=a+1$, it follows that $c=2$ and $\mathbf{x}_{[k,m]}$ ends with $a+1$ for $k\in [m]$. By the induction hypothesis on the length and last element of $\mathbf{x}_{[k,m]}$, we have
\begin{align*}
|D_t(\mathbf{x}_{[k,m]},a,2+a,2+\mathbf{y})|&\geqslant |D_t(\mathbf{x}_{[k,m]},a,a,\mathbf{y})|,
\end{align*}
for $k\in [2,m]$. We now consider the following four subcases: \textbf{Case A1}: $\ell_1\leqslant m-2$ and $\ell_2\leqslant m-2$; \textbf{Case A2}:  $\ell_1\geqslant m-1$ and $\ell_2\geqslant m-1$; \textbf{Case A3}:  $\ell_1\leqslant m-2$ and $\ell_2\geqslant m-1$; \textbf{Case A4}: $\ell_1\geqslant m-1$ and $\ell_2\leqslant m-2$.

\textbf{Case A1}: If $\ell_1\leqslant m-2$ and $\ell_2\leqslant m-2$, then $\mathbf{x}$ contains the elements $x_1+1$ and $x_1+2$. Thus $\ell_1^*=\ell_1$ and $\ell_2^*=\ell_2$. By the induction hypothesis on the length of $\mathbf{x}$, we have
\begin{align*}
|D_{t-1-\ell_1^*}(x_{3+\ell_1^*},...,x_{m+1},2+a,2+\mathbf{y})|&=|D_{t-1-\ell_1}(x_{3+\ell_1},...,x_{m+1},2+a,2+\mathbf{y})|\geqslant |D_{t-1-\ell_1}(x_{3+\ell_1},...,x_{m+2+n})|,\\
|D_{t-1-\ell_2^*}(x_{3+\ell_2^*},...,x_{m+1},2+a,2+\mathbf{y})|&=|D_{t-1-\ell_2}(x_{3+\ell_2},...,x_{m+1},2+a,2+\mathbf{y})|\geqslant |D_{t-1-\ell_2}(x_{3+\ell_2},...,x_{m+2+n})|.
\end{align*}
Therefore, $$|D_t(\mathbf{x},a,2+a,2+\mathbf{y})|\geqslant |D_t(\mathbf{x},a,a,\mathbf{y})|.$$
\textbf{Case A2}: If $\ell_1\geqslant m-1$ and $\ell_2\geqslant m-1$, then $\mathbf{x}$ does not contain the elements $x_1+1$ or $x_1+2$. Thus, $x_1=a+1$, $\ell_1^*=m$, $\ell_2^*=\ell_2=m-1$, $\ell_1\geqslant m+1$. \textbf{Case A3}: If $\ell_1\leqslant m-2$ and $\ell_2\geqslant m-1$, then $\ell_1^*=\ell_1$. For $x_1=a$, we have $\ell_2^*=m,\ell_2\geqslant m+1$; for $x_1=a+1$, we have  $\ell_2^*=\ell_2=m-1$. \textbf{Case A4}: If $\ell_1\geqslant m-1$ and $\ell_2\leqslant m-2$ then  $\ell_2^*=\ell_2$. For $x_1=a+1$, we have $\ell_1^*=m,\ell_1\geqslant m+1$; for $x_1=a+2$, we have $\ell_1^*=\ell_1=m-1$. Following the method used in \textbf{Case A1}, the same result holds for \textbf{Cases A2, A3, A4}.

\textbf{Case B:} For $x_m=a+2$, we have $c=1$. Similarly, it follows that $$|D_t(\mathbf{x},a,1+a,1+\mathbf{y})|\geqslant |D_t(\mathbf{x},a,a,\mathbf{y})|.$$

\textbf{Case C:} Consider the case where $x_m=a$. If $\mathbf{x}=\mathbf{u}\circ (a+2) \circ a^t$ for some $t\geqslant 1$ and $\mathbf{u}\in \mathbb{Z}_3^{m-t-1}$, then we set $c=1$; otherwise we set $c=2$. Based on the structure of $\mathbf{x}$, we consider the following three subcases: \textbf{Case C1:} $\mathbf{x}$ is of the form $\mathbf{u}\circ (a+2) \circ a^t$ for some $t\geqslant 1$ and $\mathbf{u}\in \mathbb{Z}_3^{m-t-1}$; 
\textbf{Case C2:} $\mathbf{x}$ is of the form $\mathbf{u}\circ (a+1) \circ a^t$ for some $t\geqslant 1$ and $\mathbf{u}\in \mathbb{Z}_3^{m-t-1}$; \textbf{Case C3:} $\mathbf{x}$ is of the form $a^m$.

\textbf{Case C1:}  When $\mathbf{x}$ is of the form $\mathbf{u}\circ (a+2) \circ a^t$ for some $t\geqslant 1$ and $\mathbf{u}\in \mathbb{Z}_3^{m-t-1}$, we have $\ell_2\leqslant m-2$ and $c=1$. We analyze the following two cases: \textbf{Case C1-1}: $\ell_1\leqslant m-2$ and $\ell_2\leqslant m-2$; \textbf{Case C1-2}: $\ell_1\geqslant m-1$ and $\ell_2\leqslant m-2$.
Here, the same result as in \textbf{Case C1-1} can be obtained by following the method used in \textbf{Case A1}.

\textbf{Case C1-2}: When $\ell_1\geqslant m-1$ and $\ell_2\leqslant m-2$, we have $\ell_2^*=\ell_2$, and  $\ell_1^*=m$, $\ell_1\geqslant m+1$. Here, $c=1$. Moreover, $\mathbf{x}_{[k,m]}$ is given by $\mathbf{u}_{[k,m-t-1]}\circ (a+2) \circ a^t$ for $k\in [m-t]$, or $a^{m-k+1}$ for $k\in [m-t+1,m]$. By the assumption on the length and structure of $\mathbf{x}_{[k,m]}$, we have
\begin{align*}
|D_t(\mathbf{x}_{[k,m]},a,1+a,1+\mathbf{y})|&\geqslant |D_t(\mathbf{x}_{[k,m]},a,a,\mathbf{y})|,
\end{align*}
for $k\in [2,m]$. Thus,
\begin{align*}
|D_{t-1-\ell_2^*}(x_{3+\ell_2^*},...,x_{m+1},1+a,1+\mathbf{y})|=|D_{t-1-\ell_2}(x_{3+\ell_2},...,x_{m+1},1+a,1+\mathbf{y})|\geqslant |D_{t-1-\ell_2}(x_{3+\ell_1},...,x_{m+2+n})|.
\end{align*}
Moreover, we have
\begin{align*}
|D_{t-1-\ell_1^*}(1+\mathbf{y})|&=|D_{t-1-m}(\mathbf{y})|\geqslant |D_{t-1-\ell_1}(x_{3+\ell_1},...,x_{m+2+n})|.
\end{align*}
Therefore, $$|D_t(\mathbf{x},a,1+a,1+\mathbf{y})|\geqslant |D_t(\mathbf{x},a,a,\mathbf{y})|.$$

\textbf{Case C2:} When $\mathbf{x}$ is of the form $\mathbf{u}\circ (a+1) \circ a^t$ for some $t\geqslant 1$ and $\mathbf{u}\in \mathbb{Z}_3^{m-t-1}$, we have $\ell_1\leqslant m-2$ and $c=2$. We analyze the following two cases: \textbf{Case C2-1}: $\ell_1\leqslant m-2$ and $\ell_2\leqslant m-2$; \textbf{Case C2-2}: $\ell_1\leqslant m-2$ and $\ell_2\geqslant m-1$.
Similarly, it follows that
$$|D_t(\mathbf{x},a,2+a,2+\mathbf{y})|\geqslant |D_t(\mathbf{x},a,a,\mathbf{y})|,$$
by using the method as for $\mathbf{x}=\mathbf{u}\circ (a+2)\circ a^t$.

\textbf{Case C3:} When $\mathbf{x}$ is of the form $a^m$, we have $\ell_1,\ell_2\geqslant m-2$ and $c=2$. Thus, $\ell_1^*=\ell_2$, $\ell_2^*=m$, and $\ell_1,\ell_2\geqslant m+1$. By the assumption of the length of $\mathbf{x}_{[k,m]}$ and its last element, we have
\begin{align*}
|D_t(\mathbf{x}_{[k,m]},a,2+a,2+\mathbf{y})|&\geqslant |D_t(\mathbf{x}_{[k,m]},a,a,\mathbf{y})|,
\end{align*}
for $k\in [2,m]$. Furthermore,
\begin{align*}
|D_{t-1-\ell_1^*}(x'_{3+\ell_1^*},...,x'_{m+2+n})|&=|D_{t-1-\ell_2}(2+x_{3+\ell_2},...,2+x_{m+2+n})|=|D_{t-1-\ell_2}(x_{3+\ell_2},...,x_{m+2+n})|\\
|D_{t-1-\ell_2^*}(2+\mathbf{y})|&=|D_{t-1-m}(\mathbf{y})|\geqslant |D_{t-1-\ell_1}(x_{3+\ell_1},...,x_{m+2+n})|.
\end{align*}
Therefore, $$|D_t(\mathbf{x},a,2+a,2+\mathbf{y})|\geqslant |D_t(\mathbf{x},a,a,\mathbf{y})|.$$

From the above discussion, when $|\mathbf{x}|=m$, there exists  $c\in\{1,2\}$ such that $$|D_t(\mathbf{x},a,c+a,c+\mathbf{y})|\geqslant |D_t(\mathbf{x},a,a,\mathbf{y})|.$$  By the inductive hypothesis, it follows that the proposition holds for any $\mathbf{x}\in \mathbb{Z}_3^m$ and $m\geqslant 0$.
\end{proof}

\section{Proofs of Lemmas $\ref{lm10}$ and $\ref{lm11}$}\label{APP-C}
This appendix provides the proofs of Lemmas $\ref{lm10}$ and $\ref{lm11}$.

\begin{proof}{ \rm (The proof of Lemma $\ref{lm10}$)}
Given that $(\mathbf{x},a,b)\in \mathbf{c}_3(s+2)$, it follows that 
\begin{equation*}
x_i=
\begin{cases}
c & \text{if } i\equiv s\mod 3, \\
b & \text{if } i\equiv s-1\mod 3, \\
a & \text{if } i\equiv s-2\mod 3.
\end{cases}
\end{equation*}
By interchanging $a$ and $c$ in $\mathbf{x}$, we have
\begin{equation*}
x'_i=
\begin{cases}
a & \text{if } i\equiv s\mod 3, \\
b & \text{if } i\equiv s-1\mod 3, \\
c & \text{if } i\equiv s-2\mod 3.
\end{cases}
\end{equation*}
Thus, $(\mathbf{x}',c,b,a)\in \mathbf{c}_3(s+3)$. We proceed by induction on the length of $\mathbf{x}$. If $|\mathbf{x}|=0$, then $(\mathbf{x},a,b,a,\mathbf{y})=(a,b,a,\mathbf{y})$. For convenience, let $\mathbf{y}=(y_1,y_2,...,y_p)$. By Lemma $\ref{lm1}$ we have
\begin{align*}
|D_t(a,b,a,\mathbf{y})|&=|D_t(a,b,a,\mathbf{y})^a|+|D_t(a,b,a,\mathbf{y})^{b}|+|D_t(a,b,a,\mathbf{y})^{c}|\\
&=|D_{t}(b,a,\mathbf{y})|+|D_{t-1}(a,\mathbf{y})|+|D_{t-3-\ell_1}(y_{2+\ell_1},...,y_p)|,
\end{align*}
where $\ell_1\geqslant 0$ and $y_{1+\ell_1}=c$. Moreover, we have
\begin{align*}
|D_t(c,b,a,\mathbf{y})|&=|D_t(c,b,a,\mathbf{y})^{c}|+|D_t(c,b,a,\mathbf{y})^{b}|+|D_t(c,b,a,\mathbf{y})^{a}|\\
&=|D_{t}(b,a,\mathbf{y})|+|D_{t-1}(a,\mathbf{y})|+|D_{t-2}(\mathbf{y})|\\
&\geqslant |D_{t}(b,a,\mathbf{y})|+|D_{t-1}(a,\mathbf{y})|+|D_{t-3-\ell_1}(y_{2+\ell_1},...,y_p)|=|D_t(a,b,a,\mathbf{y})|.
\end{align*}

For the induction step, assume the claim holds for all sequences with $|\mathbf{x}|<m$. We now consider the case where $|\mathbf{x}|=m$. Let $\mathbf{x}\in \mathbb{Z}_3^m$ be an arbitrary sequence. For convenience, we define the extended sequences as $(\mathbf{x},a,b,a,\mathbf{y})=(x_1,...,x_m,x_{m+1},...,x_{m+3+p})$ and $(\mathbf{x}',c,b,a,\mathbf{y})=(x'_1,...,x'_m,x'_{m+1},...,x_{m+3+p})$. Since $(x_1,...,x_m,x_{m+1},x_{m+2})\in \mathbf{c}_3(m+2)$, by Lemma $\ref{lm1}$ we have
\begin{align*}
|D_t(\mathbf{x},a,b,a,\mathbf{y})|&=|D_t(\mathbf{x},a,b,a,\mathbf{y})^{a}|+|D_t(\mathbf{x},a,b,a,\mathbf{y})^{b}|+|D_t(\mathbf{x},a,b,a,\mathbf{y})^{c}|\\
&=|D_{t}(x_{2},...,x_{m+2},a,\mathbf{y})|+|D_{t-1}(x_{3},...,x_{m+2},a,\mathbf{y})|+|D_{t-2}(x_{4},...,x_{m+2},a,\mathbf{y})|.
\end{align*}
Similarly, since $(x'_1,...,x'_m,x'_{m+1},x'_{m+2})\in \mathbf{c}_3(m+2)$, by Lemma $\ref{lm1}$ we have
\begin{align*}
|D_t(\mathbf{x}',c,b,a,\mathbf{y})|&=|D_t(\mathbf{x}',c,b,a,\mathbf{y})^{a}|+|D_t(\mathbf{x}',c,b,a,\mathbf{y})^{b}|+|D_t(\mathbf{x}',c,b,a,\mathbf{y})^{c}|\\
&=|D_{t}(x'_{2},...,x'_{m+2},a,\mathbf{y})|+|D_{t-1}(x'_{3},...,x'_{m+2},a,\mathbf{y})|+|D_{t-2}(x'_{4},...,x'_{m+2},a,\mathbf{y})|\\
&\overset{(a)}{\geqslant } |D_{t}(x_{2},...,x_{m+2},a,\mathbf{y})|+|D_{t-1}(x_{3},...,x_{m+2},a,\mathbf{y})|+|D_{t-2}(x_{4},...,x_{m+2},a,\mathbf{y})|=|D_t(\mathbf{x},a,b,a,\mathbf{y})|,
\end{align*}
where $(a)$ follows from the assumption of the length $|\mathbf{x}|<m$.

Therefore, by the induction hypothesis, we conclude that  $$|D_t(\mathbf{x}',c,b,a,\mathbf{y})|\geqslant |D_t(\mathbf{x},a,b,a,\mathbf{y})|,$$  for any $\mathbf{x}\in \mathbb{Z}_3^m$ and $m\geqslant 0$. 
\end{proof}

\begin{proof}{ \rm (The proof of Lemma $\ref{lm11}$)} 
By Lemma $\ref{lm9}$, it suffices to consider the case where $\mathbf{x}$ has exactly $n$ runs. Furthermore, Lemma $\ref{lm10}$ allows us to restrict our attention to sequences $\mathbf{x}$ with exactly $n$ runs that contain exactly one occurrence of the subsequence $(a,b,a)$ for some distinct $a,b\in \mathbb{Z}_3$. Assume without loss of generality that $a=1$ and $b=0$. 

In this case, $\mathbf{x}$ can be expressed as the concatenation of $\mathbf{c}_3(s,\sigma)$ and $\mathbf{c}_3(t,\gamma)$, where:
\begin{itemize}
\item $n-1\geqslant s\geqslant 2$ and $s+t=n$,
\item $\sigma$ is an ordering such that the last three elements of $\mathbf{c}_3(s,\sigma)$ are $2,1,0$ in sequence,
\item $\gamma$ is the ordering $1,2,0$.
\end{itemize}
For convenience, let $\gamma=(1,2,0)$. We define $\mathbf{c}_{n,i}$ as the concatenation of $\mathbf{c}_3(i+1,\sigma)$ and $\mathbf{c}_3(n-i-1,\gamma)$, where $n-2\geqslant i\geqslant 1$. For example, when $n=5$, we have $\mathbf{c}_{5,1}=(1,0,1,2,0), \mathbf{c}_{5,2}=(2,1,0,1,2), \mathbf{c}_{5,3}=(0,2,1,0,1)$.

We define the functions $f_i(n,t)$ for $i=1,2,3$ as follows:
\begin{align*}
f_1(n,t)&=D_3(n-1,t)+D_3(n-2,t-1)+D_3(n-4,t-3),\\
f_2(n,t)&=D_3(n-2,t)+D_3(n-2,t-1)+D_3(n-3,t-1)+D_3(n-3,t-2)+D_3(n-5,t-3),\\
f_3(n,t)&=D_3(n-3,t)+2D_3(n-3,t-1)+D_3(n-3,t-2)+D_3(n-4,t-1)+2D_3(n-4,t-2)\\
&~~~+D_3(n-6,t-3)+D_3(n-6,t-4).
\end{align*}
For any $4\leq i\leqslant n-2$, we recursively define
$$f_i(n,t)=f_{i-1}(n-1,t)+f_{i-2}(n-2,t-1)+f_{i-3}(n-3,t-2).$$
Next, we prove by induction that for all $1\leqslant i\leqslant n-2$,
$$|D_t(\mathbf{c}_{n,i})|=f_i(n,t).$$

When $i=1$, we have $\mathbf{c}_{n,1}=(1,0,1,2,0,\mathbf{c}_3(n-5,(1,2,0))).$ Thus,  $|D_t(\mathbf{c}_{n,1})|=|D_t(\mathbf{c}_{n,1})^0|+|D_t(\mathbf{c}_{n,1})^1|+|D_t(\mathbf{c}_{n,1})^2|=|D_{t-1}(\mathbf{c}_3(n-2,(1,2,0))|+|D_t(\mathbf{c}_3(n-1,(0,1,2))|+
|D_{t-3}(\mathbf{c}_3(n-4,(0,1,2))|=D_3(n-1,t)+D_3(n-2,t-1)+D_3(n-4,t-3)=f_1(n,t).$

When $i=2$, we have $\mathbf{c}_{n,2}=(2,1,0,1,2,0,\mathbf{c}_3(n-6,(1,2,0))).$ Hence, $|D_t(\mathbf{c}_{n,2})|=|D_t(\mathbf{c}_{n,2})^0|+|D_t(\mathbf{c}_{n,2})^1|+|D_t(\mathbf{c}_{n,2})^2|=|D_{t-2}(\mathbf{c}_3(n-3,(1,2,0))|+|D_{t-1}(\mathbf{c}_3(n-2,(0,1,2))|+
|D_{t}(\mathbf{c}_{n-1,1})|=D_3(n-3,t-2)+D_3(n-2,t-1)+f_1(n-1,t)=D_3(n-2,t)+D_3(n-2,t-1)+D_3(n-3,t-1)+D_3(n-3,t-2)+D_3(n-5,t-3)=f_2(n,t).$

When $i=3$, we have $\mathbf{c}_{n,3}=(0,2,1,0,1,2,0,\mathbf{c}_3(n-7,(1,2,0))).$ Thus, $|D_t(\mathbf{c}_{n,3})|=|D_t(\mathbf{c}_{n,3})^0|+|D_t(\mathbf{c}_{n,3})^1|+|D_t(\mathbf{c}_{n,3})^2|=|D_{t}(\mathbf{c}_{n-1,2})|+|D_{t-2}(\mathbf{c}_3(n-3,(0,1,2))|+
|D_{t-1}(\mathbf{c}_{n-2,1})|=f_2(n-1,t)+D_3(n-3,t-2)+f_1(n-2,t-1)=D_3(n-3,t)+2D_3(n-3,t-1)+D_3(n-3,t-2)+D_3(n-4,t-1)+2D_3(n-4,t-2)+D_3(n-6,t-3)+D_3(n-6,t-4)=f_3(n,t).$
From the above computations, we have  $|D_t(\mathbf{c}_{n,i})|=f_i(n,t)$ for $1\leqslant i\leqslant 3$, which indicate the base case of the induction holds. 

Assume that the claim holds for all $j<i$, and consider the sequence $\mathbf{c}_{n,i}$. For convenience, denote $D_t(\mathbf{c}_{n,i})=(c_1,c_2,...,c_n)$ where $\{c_1,c_2,c_3\}=\mathbb{Z}_3$. By decomposition, we have
\begin{align*}
|D_t(\mathbf{c}_{n,i})|&=|D_t(\mathbf{c}_{n,i})^{c_1}|+|D_t(\mathbf{c}_{n,i})^{c_2}|+|D_t(\mathbf{c}_{n,i})^{c_3}|\\
&=|D_t\big(\mathbf{c}_3(i,(c_2,c_3,c_1))\circ\mathbf{c}_3(n-i-1,(1,2,0))\big)|+|D_{t-1}\big(\mathbf{c}_3(i-1,(c_3,c_1,c_2))\circ\mathbf{c}_3(n-i-1,(1,2,0))\big)|\\
&~~~+|D_{t-2}\big(\mathbf{c}_3(i-2,(c_1,c_2,c_3))\circ\mathbf{c}_3(n-i-1,(1,2,0))\big)|\\
&=f_{i-1}(n-1,t)+f_{i-2}(n-2,t-1)+f_{i-3}(n-3,t-2)\\
&=f_i(n,t).
\end{align*}

Finally, we prove by induction that for all $1\leqslant i\leqslant n-2$, $f_i(n,t)\leqslant D_3(n-2,t)+D_3(n-2,t-1)+D_3(n-3,t-1)+D_3(n-3,t-2)+D_3(n-5,t-3)$. For the base case, equality holds for $f_2(n,t)$, and the inequality holds for $f_1(n,t)$ and $f_3(n,t)$, as shown below:
\begin{align*}
D_3&(n-2,t)+D_3(n-2,t-1)+D_3(n-3,t-1)+D_3(n-3,t-2)+D_3(n-5,t-3)-f_1(n,t)\\
&=D_3(n-2,t)+D_3(n-2,t-1)+D_3(n-3,t-1)+D_3(n-3,t-2)+D_3(n-5,t-3)-D_3(n-1,t)\\
&~~~-D_3(n-2,t-1)-D_3(n-4,t-3)\\
&=D_3(n-2,t)+D_3(n-3,t-1)+D_3(n-3,t-2)+D_3(n-5,t-3)-D_3(n-2,t)-D_3(n-3,t-1)\\
&~~~-D_3(n-4,t-2)-D_3(n-4,t-3)\\
&=D_3(n-3,t-2)+D_3(n-5,t-3)-D_3(n-4,t-2)-D_3(n-4,t-3)\\
&=2D_3(n-5,t-3)+D_3(n-6,t-4)-D_3(n-5,t-3)-D_3(n-6,t-4)-D_3(n-7,t-5)\\
&=D_3(n-5,t-3)-D_3(n-7,t-5)\geqslant 0,
\end{align*}
and
\begin{align*}
D_3&(n-2,t)+D_3(n-2,t-1)+D_3(n-3,t-1)+D_3(n-3,t-2)+D_3(n-5,t-3)-f_3(n,t)\\
&=D_3(n-2,t)+D_3(n-2,t-1)+D_3(n-3,t-1)+D_3(n-3,t-2)+D_3(n-5,t-3)-D_3(n-3,t)\\
&~~-2D_3(n-3,t-1)-D_3(n-3,t-2)-D_3(n-4,t-1)-2D_3(n-4,t-2)-D_3(n-6,t-3)-D_3(n-6,t-4)\\
&=D_3(n-8,t-5)-D_3(n-9,t-6)\geqslant 0.
\end{align*}
This follows from Eqs. $(\ref{eq3})$ and $(\ref{eq6})$ under the condition $n\geqslant t+3$. When $n=t+2$, the two inequalities also holds.

Assume that the claim holds for all $i\leq k$ and consider the case where $i=k+1$. Then
\begin{align*}
f_{k+1}(n,t)&=f_k(n-1,t)+f_{k-1}(n-2,t-1)+f_{k-2}(n-3,t-2)\\
&\leqslant D_3(n-3,t)+D_3(n-3,t-1)+D_3(n-4,t-1)+D_3(n-4,t-2)+D_3(n-6,t-3)\\
&~~~+D_3(n-4,t-1)+D_3(n-4,t-2)+D_3(n-5,t-2)+D_3(n-5,t-3)+D_3(n-7,t-4)\\
&~~~+D_3(n-5,t-2)+D_3(n-5,t-3)+D_3(n-6,t-3)+D_3(n-6,t-4)+D_3(n-8,t-5)\\
&\overset{(a)}{=}D_3(n-2,t)+D_3(n-2,t-1)+D_3(n-3,t-1)+D_3(n-3,t-2)+D_3(n-5,t-3),
\end{align*}
where $(a)$ follows from Eq. $(\ref{eq3})$ for $n\geqslant t+3$. When $n=t+2$, Equality $(a)$ becomes an inequality but still holds. Thus, the lemma follows by induction.
\end{proof}

\section{Proof of Lemma $\ref{lm12}$ }\label{APP-D}
The purpose of this appendix is to give the proof of Lemma $\ref{lm12}$.

\begin{proof}{ \rm (The proof of Lemma $\ref{lm12}$)} Since $n\geqslant \max\{9,\lfloor \frac{3t}{2}\rfloor +1\}$ and  $t\geqslant 2$, we have $n\geqslant t+5$  except for the cases where $(n,t)=(9,5), (10,6),$ or $(11,7)$. Let $\mathcal{X}=D_t(\mathbf{x})\cap D_t(\mathbf{y})$ denote the intersection of the two sets. 
When $(n,t)=(9,5), (10,6)$ or $(11,7)$, by a computerized search, it follows that $|\mathcal{X}|\leqslant M_0(n,t)$. Next, we restrict our attention to the case where $n\geqslant t+5$. Without loss of generality, let $a=0,b=1,c=2$ with $\mathbf{x}=(1,0,2,0,1,x_6,...,x_n)$ and $\mathbf{y}=(0,1,2,0,1,y_6,...,y_n)$.

Given that $d_L(\mathbf{x},\mathbf{y})\geqslant 2$, $\mathbf{x}_{[1,5]}=(1,0,2,0,1)$, and $\mathbf{y}_{[1,5]}=(0,1,2,0,1)$, it follows that $\mathbf{x}_{[6,n]}\neq \mathbf{y}_{[6,n]}$. By Lemmas $\ref{lm1}$ and $\ref{lm2}$, we have $|\mathcal{X}|=|\mathcal{X}^{(0,0)}|+|\mathcal{X}^{(0,1)}|+|\mathcal{X}^{(0,2)}|+|\mathcal{X}^{(1,0)}|+|\mathcal{X}^{(1,1)}|+|\mathcal{X}^{(1,2)}|+|\mathcal{X}^{(2,0)}|+|\mathcal{X}^{(2,1)}|+|\mathcal{X}^{(2,2)}|$, where
\begin{align}
|\mathcal{X}^{(0,0)}|&=|D_{t-2}(1,x_6,...,x_n)\cap D_{t-2}(1,y_6,...,y_n)|\leqslant N_3(n-4,1,t-2)=2D_3(n-6,t-3)+D_3(n-7,t-4),\nonumber\\
|\mathcal{X}^{(0,1)}|&=|D_{t-3}(x_6,...,x_n)\cap D_{t}(2,0,1,y_6,..,y_n)|\leqslant |D_{t-3}(x_6,...,x_n)|,\nonumber\\
|\mathcal{X}^{(0,2)}|&=|D_{t-1}(0,1,x_6,...,x_n)\cap D_{t-1}(0,1,y_6,..,y_n)|\leqslant N_3(n-3,1,t-1)=2D_3(n-5,t-2)+D_3(n-6,t-3),\nonumber\\
|\mathcal{X}^{(1,0)}|&=|D_t(2,0,1,x_6,...,x_n)\cap D_{t-2}(1,y_6,...,y_n)|,\nonumber\\
|\mathcal{X}^{(1,1)}|&=|D_{t-3}(x_6,...,x_n)\cap D_{t-3}(y_6,...,y_n)|\leqslant N_3(n-5,1,t-3)=2D_3(n-7,t-4)+D_3(n-8,t-5),\nonumber\\
|\mathcal{X}^{(1,2)}|&=|D_{t-1}(0,1,x_6,...,x_n)\cap D_{t-1}(0,1,y_6,...,y_n)|\leqslant N_3(n-3,1,t-1)=2D_3(n-5,t-2)+D_3(n-6,t-3),\nonumber\\
|\mathcal{X}^{(2,0)}|&=|D_{t-2}(1,x_6,...,x_n)\cap D_{t-2}(1,y_6,...,y_n)|\leqslant N_3(n-4,1,t-2)=2D_3(n-6,t-3)+D_3(n-7,t-4),\nonumber\\
|\mathcal{X}^{(2,1)}|&=|D_{t-3}(x_6,...,x_n)\cap D_{t-3}(y_6,...,y_n)|\leqslant N_3(n-5,1,t-3)=2D_3(n-7,t-4)+D_3(n-8,t-5),\nonumber\\
|\mathcal{X}^{(2,2)}|&=|D_{t-4-\ell}(x_{7+\ell},...,x_n)\cap D_{t-4-\ell^*}(y_{7+\ell^*},...,y_n)|,\label{eq38}
\end{align}
with $\ell,\ell^*\geqslant 0$. Next, we further decompose $\mathcal{X}^{(1,0)}$, yielding  
\begin{align}
|\mathcal{X}^{(1,0)}|&=|D_t(2,0,1,x_6,...,x_n)\cap D_{t-2}(1,y_6,...,y_n)|=|D_t(0,1,x_6,...,x_n)\cap D_{t-3-\ell_0}(y_{7+\ell_0},...,y_n)|\nonumber\\
&~~~+|D_{t-1}(1,x_6,...,x_n)\cap D_{t-3-\ell_1}(y_{7+\ell_1},...,y_n)|+|D_{t-2}(x_6,...,x_n)\cap D_{t-2}(y_{6},...,y_n)|\nonumber\\
&\leqslant |D_{t-3-\ell_0}(y_{7+\ell_0},...,y_n)|+|D_{t-3-\ell_1}(y_{7+\ell_1},...,y_n)|+N_3(n-5,1,t-2)\nonumber\\
&= |D_{t-3-\ell_0}(y_{7+\ell_0},...,y_n)|+|D_{t-3-\ell_1}(y_{7+\ell_1},...,y_n)|+2D_3(n-7,t-3)+D_3(n-8,t-4)\label{eq39}
\end{align}
under the condition $n\geqslant t+4$, where $\ell_0,\ell_1\geqslant 0$ and $\ell_0\neq \ell_1$. We analyze the
following three cases: \textbf{Case A}: $\ell\geqslant 1$ or $\ell^*\geqslant 1$; \textbf{Case B}:  $\ell=\ell^*=0$ and $\ell_1\geqslant 2$; \textbf{Case C}:   $\ell=\ell^*=0$ and $\ell_1<2$.

\textbf{Case A}: If $\ell\geqslant 1$ or $\ell^*\geqslant 1$, then $x_6\neq 2$ or $y_6\neq 2$. Under this condition, we obtain the following bound:
\begin{align}
|\mathcal{X}^{(2,2)}|\leqslant \min\{|D_{t-4-\ell}(x_{7+\ell},...,x_n)|,|D_{t-4-\ell^*}(y_{7+\ell^*},...,y_n)|\}\leqslant D_3(n-7,t-5).\label{eq40}    
\end{align}
Furthermore, 
\begin{align}
|\mathcal{X}^{(0,1)}|&\leqslant D_3(n-5,t-3), \label{eq41}\\
|\mathcal{X}^{(1,0)}|&\leqslant D_3(n-6,t-3)+D_3(n-7,t-4)+2D_3(n-7,t-3)+D_3(n-8,t-4),\label{eq42}
\end{align}
since $\ell_0\neq \ell_1$ and $\ell_0,\ell_1\geqslant 0$. Combining  Eqs. $(\ref{eq38})$-$(\ref{eq42})$, we have 
\begin{align*}
|\mathcal{X}|\leqslant M_0(n,t)+\big(D_3(n-10,t-5)-D_3(n-12,t-7)\big)-\big(D_3(n-8,t-5)-D_3(n-10,t-7)\big).    
\end{align*}
We analyze the term $D_3(n-8,t-5)-D_3(n-10,t-7)$ as follows:
\begin{align*}
D_3&(n-8,t-5)-D_3(n-10,t-7)\\
&\overset{(a)}{=}\big(D_3(n-9,t-5)-D_3(n-11,t-7)\big)+\big(D_3(n-10,t-6)-D_3(n-12,t-8)\big)+\big(D_3(n-11,t-7)-D_3(n-13,t-9)\big)\\
&\overset{(b)}{\geqslant} D_3(n-9,t-5)-D_3(n-11,t-7)\\
&\overset{(c)}{=}\big(D_3(n-10,t-5)-D_3(n-12,t-7)\big)+\big(D_3(n-11,t-6)-D_3(n-13,t-8)\big)+\big(D_3(n-12,t-7)-D_3(n-14,t-9)\big)\\
&\overset{(d)}{\geqslant} D_3(n-10,t-5)-D_3(n-12,t-7),
\end{align*}
where $(a),(c)$ follow from Eq. $(\ref{eq3})$ under the condition $n\geqslant t+5$, $(b),(d)$ follow from Eq. $(\ref{eq6})$. 

Therefore, when $\ell\geqslant 1$ or $\ell^*\geqslant 1$, we have $$|\mathcal{X}|\leqslant M_0(n,t).$$ 

\textbf{Case B}: If $\ell=\ell^*=0$, then $x_6=y_6=2$ and $\ell_0=0$. Thus,
\begin{align}
 |\mathcal{X}^{(2,2)}|&=|D_{t-4}(x_{7},...,x_n)\cap D_{t-4}(y_{7},...,y_n)|\leqslant N_3(n-6,1,t-4)=2D_3(n-8,t-5)+D_3(n-9,t-6),\label{eq43}   
\end{align}
for $\ell=\ell^*=0$.
Moreover, if $\ell_1\geqslant 2$, then 
\begin{align}
|\mathcal{X}^{(1,0)}|&\leqslant D_3(n-6,t-3)+D_3(n-8,t-5)+2D_3(n-7,t-3)+D_3(n-8,t-4),\label{eq44}
\end{align}
since $\ell_0=0$. Combining Eqs. $(\ref{eq38})$, $(\ref{eq41})$, $(\ref{eq43})$, and $(\ref{eq44})$, we have 
\begin{align*}
|\mathcal{X}|\leqslant M_0(n,t)+\big(D_3(n-10,t-5)-D_3(n-12,t-7)\big)-\big(D_3(n-7,t-4)-D_3(n-8,t-5)\big).    
\end{align*}
Since $\big(D_3(n-7,t-4)-D_3(n-8,t-5)\big)\geqslant \big(D_3(n-10,t-5)-D_3(n-12,t-7)\big)$, we have $$|\mathcal{X}|\leqslant M_0(n,t).$$ 

\textbf{Case C}:  If $\ell=\ell^*=0$, then $x_6=y_6=2$ and $\ell_0=0$. Since $\ell_1<2$ and $\ell_1\neq \ell_0$, we have $\ell_1=1$. Thus $y_7=0$. That is, $\mathbf{x}_{[1,6]}=(1,0,2,0,1,2)$ and $\mathbf{y}_{[1,7]}=(0,1,2,0,1,2,0)$. Since $d_L(\mathbf{x},\mathbf{y})\geqslant 2$, we have 
$\mathbf{x}_{[7,n]}\neq \mathbf{y}_{[7,n]}$. Here, we have
$$|\mathcal{X}^{(2,0)}|=|D_{t-2}(1,2,x_7,x_8,...,x_n)\cap D_{t-2}(1,2,0,y_8,...,y_n)|.$$
Furthermore,  $|\mathcal{X}^{(2,0)}|=|\mathcal{X}^{(2,0,0)}|+|\mathcal{X}^{(2,0,1)}|+|\mathcal{X}^{(2,0,2)}|$, where
\begin{align*}
|\mathcal{X}^{(2,0,0)}|&=|D_{t-4-\ell_2}(x_{8+\ell_2},...,x_n)\cap D_{t-4}(y_8,...,y_n)|,\\
|\mathcal{X}^{(2,0,1)}|&=|D_{t-2}(2,x_7,x_8,...,x_n)\cap D_{t-2}(2,0,y_8,...,y_n)|\leqslant N_3(n-5,1,t-2)=2D_3(n-7,t-3)+D_3(n-8,t-4),\\
|\mathcal{X}^{(2,0,2)}|&=|D_{t-3}(x_7,x_8,...,x_n)\cap D_{t-3}(0,y_8,...,y_n)|\leqslant N_3(n-6,1,t-3)=2D_3(n-8,t-4)+D_3(n-9,t-5),
\end{align*}
with $\ell_2\geqslant 0$.
In \textbf{Case C}, we consider the following two subcases: \textbf{Case C1}: with $\ell_2\geqslant 1$;  \textbf{Case C2}: with $\ell_2=0$.

\textbf{Case C1}: If $\ell_2\geqslant 1$, then $$|\mathcal{X}^{(2,0,0)}|\leqslant |D_{t-4-\ell_2}(x_{8+\ell_2},...,x_n)|\leqslant D_3(n-8,t-5).$$
Combining with Eqs. $(\ref{eq38})$ and $(\ref{eq41})$-$(\ref{eq43})$, we have 
\begin{align*}
|\mathcal{X}|\leqslant M_0(n,t)+\big(D_3(n-10,t-5)-D_3(n-12,t-7)\big)-\big(D_3(n-9,t-5)-D_3(n-11,t-7)\big).    
\end{align*}
Thus, if $\ell_2\geqslant 1$, then $$|\mathcal{X}|\leqslant M_0(n,t).$$ 

\textbf{Case C2}: If $\ell_2=0$, then $x_7=0$. That is, $\mathbf{x}_{[1,7]}=(1,0,2,0,1,2,0)$ and $\mathbf{y}_{[1,7]}=(0,1,2,0,1,2,0)$. Since $n\geqslant 9$ and  $\frac{2n}{3}> t\geqslant 2$, we have  $n\geqslant t+5$ except for $(n,t)=(9,5), (10,6),$ or $(11,7)$. Based on the relation between  $\mathbf{x}_{[6,n]}$ and $\mathbf{c}_3(n-5)$, we discuss the value of $|\mathcal{X}|$ in the following two cases: \textbf{Case C2-1}: $\mathbf{x}_{[6,n]}\notin \mathbf{c}_3(n-5)$; \textbf{Case C2-2}: $\mathbf{x}_{[6,n]}\in \mathbf{c}_3(n-5)$.

\textbf{Case C2-1}: In this case, when $(n,t)=(9,5)$, $(10,6)$, or  $(11,7)$, by using a computerized search, we obtain  
$$|D_{t-3}(x_6,...,x_n)| \leqslant D_3(n-7,t-3)+D_3(n-7,t-4)+D_3(n-8,t-4)+D_3(n-8,t-5)+D_3(n-10,t-6).$$
Since $(n-5)\geqslant (t-3)+2$ and $\mathbf{x}_{[6,n]}\notin \mathbf{c}_3(n-5)$, by Lemma $\ref{lm11}$  
\begin{align*}
|\mathcal{X}^{(0,1)}|&\leqslant |D_{t-3}(x_6,...,x_n)|\\ 
&\leqslant D_3(n-7,t-3)+D_3(n-7,t-4)+D_3(n-8,t-4)+D_3(n-8,t-5)+D_3(n-10,t-6).
\end{align*}
%Therefore, if $\mathbf{x}_{[6,n]}\notin \mathbf{c}_3(n-5)$, then 
%\begin{align*}
%|\mathcal{X}^{(0,1)}|\leqslant D_3(n-7,t-3)+D_3(n-7,t-4)+D_3(n-8,t-4)+D_3(n-8,t-5)+D_3(n-10,t-6).
%\end{align*}
Combining Eqs. $(\ref{eq38})$, $(\ref{eq42})$, $(\ref{eq43})$, we have 
\begin{align*}
|\mathcal{X}|\leqslant M_0(n,t)+\big(D_3(n-10,t-5)-D_3(n-12,t-7)\big)-\big(D_3(n-9,t-5)-D_3(n-10,t-6)\big)\leqslant M_0(n,t).    
\end{align*}

\textbf{Case C2-2}: If $\mathbf{x}_{[6,n]}\in \mathbf{c}_3(n-5)$, $x_6=2$, and $x_7=0$, then we have $\mathbf{x}_{[6,n]}=(2,\mathbf{a}_{n-6})$. Thus, $x_8=1$ and 
\begin{align}
|\mathcal{X}^{(0,1)}|=D_3(n-5,t-3).\label{eq45}
\end{align}
Decomposing $\mathcal{X}^{(0,2)}$:   $|\mathcal{X}^{(0,2)}|=|\mathcal{X}^{(0,2,0)}|+|\mathcal{X}^{(0,2,1,2)}|+|\mathcal{X}^{(0,2,1,0)}|+|\mathcal{X}^{(0,2,1,1)}|+|\mathcal{X}^{(0,2,2)}|$, where
\begin{align}
|\mathcal{X}^{(0,2,0)}|&=|D_{t-1}(1,x_6,...,x_n)\cap D_{t-1}(1,y_6,...,y_n)|\leqslant N_3(n-4,1,t-1)=2D_3(n-6,t-2)+D_3(n-7,t-3),\nonumber\\
|\mathcal{X}^{(0,2,1,2)}|&=|D_{t-2}(0,1,x_9,...,x_n)\cap D_{t-2}(0,y_8,...,y_n)|\leqslant N_3(n-6,1,t-2)=2D_3(n-8,t-3)+D_3(n-9,t-4),\nonumber\\
|\mathcal{X}^{(0,2,1,0)}|&=|D_{t-3}(1,x_9,...,x_n)\cap D_{t-3}(y_8,...,y_n)|\leqslant N_3(n-7,1,t-3)=2D_3(n-9,t-4)+D_3(n-10,t-5),\nonumber\\
|\mathcal{X}^{(0,2,1,1)}|&=|D_{t-4}(x_9,...,x_n)\cap D_{t-4-\ell_3}(y_{9+\ell_3},...,y_n)|,\nonumber\\
|\mathcal{X}^{(0,2,2)}|&=|D_{t-3}(0,1,x_9,...,x_n)\cap D_{t-3}(0,y_8,...,y_n)|\leqslant N_3(n-6,1,t-3)=2D_3(n-8,t-4)+D_3(n-9,t-5),\label{eq46}
\end{align}
with $\ell_3\geqslant 0$.
Here, we discuss the value of $|\mathcal{X}|$ in the following two cases: \textbf{Case C2-2-1}: $\ell_3\neq 0$; \textbf{Case C2-2-2}: $\ell_3=0$.

\textbf{Case C2-2-1}: If $\ell_3\neq 0$, then 
\begin{align}
|\mathcal{X}^{(0,2,1,1)}|&\leqslant |D_{t-4-\ell_3}(y_{9+\ell_3},...,y_n)|\leqslant D_3(n-9,t-5).\label{eq47}   
\end{align}
By Eqs. $(\ref{eq46})$ and $(\ref{eq47})$, we obtain 
\begin{align}
|\mathcal{X}^{(0,2)}|\leqslant 2D_3(n-5,t-2)+D_3(n-6,t-3)-\big(D_3(n-10,t-5)-D_3(n-12,t-7)\big).\label{eq48}      
\end{align}
Combining Eqs. $(\ref{eq38})$, $(\ref{eq42})$, $(\ref{eq43})$, $(\ref{eq45})$, $(\ref{eq48})$, we have 
\begin{align*}
|\mathcal{X}|\leqslant M_0(n,t).    
\end{align*}

\textbf{Case C2-2-2}: If $\ell_3=0$ then $y_8=1$. That is, $\mathbf{x}=(1,0,2,0,1,2,\mathbf{a}_{n-6})$ and $\mathbf{y}_{[1,8]}=(0,1,2,0,1,2,0,1)$. When $\mathbf{y}_{[8,n]}\notin \mathbf{c}_3(n-7)$, then $\mathbf{y}_{[7,n]}\notin \mathbf{c}_3(n-6)$. Since $n\geqslant t+5$, it follows that $(n-6)\geqslant (t-3)+2$ and $(n-7)\geqslant (t-4)+2$. By Eq. $(\ref{eq39})$ and Lemma $\ref{lm11}$, we have 
\begin{align}
|\mathcal{X}^{(1,0)}|&\leqslant  |D_{t-3}(y_{7},...,y_n)|+|D_{t-4}(y_{8},...,y_n)|+2D_3(n-7,t-3)+D_3(n-8,t-4) \nonumber\\ 
&\leqslant D_3(n-8,t-3)+D_3(n-8,t-4)+D_3(n-9,t-4)+D_3(n-9,t-5)+D_3(n-11,t-6)\nonumber\\
&~~~+D_3(n-9,t-4)+D_3(n-9,t-5)+D_3(n-10,t-5)+D_3(n-10,t-6)+D_3(n-12,t-7)\nonumber\\
&~~~+2D_3(n-7,t-3)+D_3(n-8,t-4)\nonumber\\
&= D_3(n-6,t-3)+D_3(n-7,t-4)+2D_3(n-7,t-3)+D_3(n-8,t-4)-\big(D_3(n-10,t-5)-D_3(n-12,t-7)\big).\label{eq49}
\end{align}
Combining Eqs. $(\ref{eq38})$, $(\ref{eq43})$, $(\ref{eq45})$, $(\ref{eq49})$, we obtain 
\begin{align*}
|\mathcal{X}|\leqslant M_0(n,t).    
\end{align*}

When $\mathbf{y}_{[8,n]}\in \mathbf{c}_3(n-7)$, we have $\mathbf{y}_{[8,n]}=\mathbf{c}(n-7,(1,0,2))$ or $\mathbf{y}_{[8,n]}=\mathbf{c}(n-7,(1,2,0))$ under the condition $y_8=1$. Since $d_L(\mathbf{x},\mathbf{y})\geqslant 2$, we have $\mathbf{y}_{[8,n]}=\mathbf{c}(n-7,(1,0,2))$. Thus,  
\begin{align*}
|D_{t-3}(y_{7},...,y_n)|&=D_3(n-7,t-3)+D_3(n-8,t-4)+D_3(n-10,t-6),\\
|D_{t-4}(y_{8},...,y_n)|&=D_3(n-7,t-4).   
\end{align*}
If $\mathbf{y}_{[8,n]}\in \mathbf{c}_3(n-7)$, then 
\begin{align}
|\mathcal{X}^{(1,0)}|&\leqslant  |D_{t-3}(y_{7},...,y_n)|+|D_{t-4}(y_{8},...,y_n)|+2D_3(n-7,t-3)+D_3(n-8,t-4) \nonumber\\ 
&\leqslant D_3(n-7,t-3)+D_3(n-8,t-4)+D_3(n-10,t-6)+D_3(n-7,t-4)+2D_3(n-7,t-3)+D_3(n-8,t-4)\nonumber\\
&= D_3(n-6,t-3)+D_3(n-7,t-4)+2D_3(n-7,t-3)+D_3(n-8,t-4)-\big(D_3(n-9,t-5)-D_3(n-10,t-6)\big).\label{eq50}
\end{align}
Combining Eqs. $(\ref{eq38})$, $(\ref{eq43})$, and $(\ref{eq45})$, $(\ref{eq50})$, we obtain 
\begin{align*}
|\mathcal{X}|&\leqslant M_0(n,t)+\big(D_3(n-10,t-5)-D_3(n-12,t-7)\big)-\big(D_3(n-9,t-5)-D_3(n-10,t-6)\big)\\   
&\leqslant M_0(n,t).
\end{align*}
By the above discussion, the lemma follows.
\end{proof}

\section{Proof of Lemma $\ref{lm13}$ }\label{APP-E}
The purpose of this appendix is to give the proof of Lemma $\ref{lm13}$.

\begin{proof}{ \rm (The proof of Lemma $\ref{lm13}$)}
Since $n\geqslant \max\{9,\lfloor \frac{3t}{2}\rfloor +1\}$ and  $t\geqslant 2$, we have $n\geqslant t+4$. Without loss of generality, let $a=0,b=1,c=2$ with $\mathbf{x}=(1,0,2,1,0,x_6,...,x_n)$ and $\mathbf{y}=(0,1,2,0,1,y_6,...,y_n)$. Let $\mathcal{X}=D_t(\mathbf{x})\cap D_t(\mathbf{y})$.

By Lemmas $\ref{lm1}$ and $\ref{lm2}$, we have $|\mathcal{X}|=|\mathcal{X}^{(0,0)}|+|\mathcal{X}^{(0,1)}|+|\mathcal{X}^{(0,2)}|+|\mathcal{X}^{(1,0)}|+|\mathcal{X}^{(1,1)}|+|\mathcal{X}^{(1,2)}|+|\mathcal{X}^2|$, where
\begin{align}
|\mathcal{X}^{(0,0)}|=&|D_{t-3}(\mathbf{x}_{[6,n]})\cap D_{t-2}(1,\mathbf{y}_{[6,n]})|\leqslant |D_{t-3}(\mathbf{x}_{[6,n]})|\leqslant D_3(n-5,t-3)\nonumber,\\
|\mathcal{X}^{(0,1)}|=&|D_{t-2}(0,\mathbf{x}_{[6,n]})\cap D_{t}(2,0,1,\mathbf{y}_{[6,n]})|\leqslant |D_{t-2}(0,\mathbf{x}_{[6,n]})|\leqslant D_3(n-4,t-2)\nonumber,\\
|\mathcal{X}^{(0,2)}|=&|D_{t-1}(1,0,\mathbf{x}_{[6,n]})\cap D_{t-1}(0,1,\mathbf{y}_{[6,n]})|\leqslant N_3(n-3,1,t-1)=2D_3(n-5,t-2)+D_3(n-6,t-3)\nonumber,\\
|\mathcal{X}^{(1,0)}|=&|D_{t}(2,1,0,\mathbf{x}_{[6,n]})\cap D_{t-2}(1,\mathbf{y}_{[6,n]})|\leqslant |D_{t-2}(1,\mathbf{y}_{[6,n]})|\leqslant D_3(n-4,t-2)\nonumber,\\
|\mathcal{X}^{(1,1)}|=&|D_{t-2}(0,\mathbf{x}_{[6,n]})\cap D_{t-3}(\mathbf{y}_{[6,n]})|\leqslant |D_{t-3}(\mathbf{y}_{[6,n]})|\leqslant D_3(n-5,t-3)\nonumber,\\
|\mathcal{X}^{(1,2)}|=&|D_{t-1}(1,0,\mathbf{x}_{[6,n]})\cap D_{t-1}(0,1,\mathbf{y}_{[6,n]})|\leqslant N_3(n-3,1,t-1)=2D_3(n-5,t-2)+D_3(n-6,t-3)\nonumber,\\
|\mathcal{X}^2|=&|D_{t-2}(1,0,\mathbf{x}_{[6,n]})\cap D_{t-2}(0,1,\mathbf{y}_{[6,n]})|\leqslant N_3(n-3,1,t-2)=2D_3(n-5,t-3)+D_3(n-6,t-4).\label{eq51}
\end{align}
We analyze the
following three cases: \textbf{Case A}: $x_6\neq 2$ or $y_6\neq 2$; \textbf{Case B}:  $x_6=y_6=2$, and $\mathbf{x}_{[7,n]}\notin \mathbf{c}_3(n-6)$ or $\mathbf{y}_{[7,n]}\notin \mathbf{c}_3(n-6)$; \textbf{Case C}: $x_6=y_6=2$, and $\mathbf{x}_{[7,n]},\mathbf{y}_{[7,n]}\in \mathbf{c}_3(n-6)$.

\textbf{Case A}: First, consider $x_6\neq 2$. For convenience, let $A(n,t)=D_{t-2}(1,0,\mathbf{x}_{[6,n]})\cap D_{t-2}(0,1,\mathbf{y}_{[6,n]})$. Then we decompose $A(n,t)$ by $|A(n,t)|=|A(n,t)^0|+|A(n,t)^1|+|A(n,t)^2|$ such that
\begin{align*}
|A(n,t)^0|=&|D_{t-3}(\mathbf{x}_{[6,n]})\cap D_{t-2}(1,\mathbf{y}_{[6,n]})|\leqslant D_3(n-5,t-3),\\
|A(n,t)^1|=&|D_{t-2}(0,\mathbf{x}_{[6,n]})\cap D_{t-3}(\mathbf{y}_{[6,n]})|\leqslant D_3(n-5,t-3),\\
|A(n,t)^2|=&|D_{t-4-\ell_1}(\mathbf{x}_{[7+\ell_1,n]})\cap D_{t-4-\ell_1^*}(\mathbf{y}_{[7+\ell_1^*,n]})|\leqslant |D_{t-4-\ell_1}(\mathbf{x}_{[7+\ell_1,n]})|\leqslant D_{3}(n-7,t-5),
\end{align*}
where $\ell_1\geqslant 1$ and $\ell_1^*\geqslant 0$. Thus, $|A(n,t)|\leqslant 2D_3(n-5,t-3)+D_{3}(n-7,t-5)$. Similarly,
\begin{align*}
|\mathcal{X}^{(0,2)}|=|\mathcal{X}^{(1,2)}|\leqslant 2D_3(n-5,t-2)+D_{3}(n-7,t-4).\nonumber
\end{align*}
Therefore, when $x_6\neq 2$, we have
\begin{align*}
|\mathcal{X}|&\leqslant D_3(n-5,t-3)+D_3(n-4,t-2)+2D_3(n-5,t-2)+D_{3}(n-7,t-4)+D_3(n-5,t-3)\\
&~~~+D_3(n-4,t-2)+2D_3(n-5,t-2)+D_{3}(n-7,t-4)+2D_3(n-5,t-3)+D_{3}(n-7,t-5)\\
&=2D_3(n-4,t-2)+4D_3(n-5,t-2)+4D_3(n-5,t-3)+2D_{3}(n-7,t-4)+D_{3}(n-7,t-5)\\
&=M_1(n,t)-\big(D_3(n-6,t-3)-D_3(n-7,t-4)\big)-\big(D_3(n-7,t-3)-D_3(n-10,t-6)\big)\\
&~~~-\big(D_3(n-8,t-5)-D_3(n-10,t-7)\big)\leqslant M_1(n,t),
\end{align*}
from Eqs. $(\ref{eq3})$ and $(\ref{eq6})$ under the condition $n\geqslant t+4$. 

Similarly, in the case where $y_6\neq 2$, we also have $|\mathcal{X}|\leqslant M_1(n,t)$ by using the above method. 

\textbf{Case B}: We now consider the case where $x_6=y_6=2$ and $\mathbf{x}_{[7,n]}\notin \mathbf{c}_3(n-6)$. If $\mathbf{x}_{[7,n]}\notin \mathbf{c}_3(n-6)$, then $\mathbf{x}_{[6,n]}\notin \mathbf{c}_3(n-5)$ and $\mathbf{x}_{[5,n]}\notin \mathbf{c}_3(n-4)$. Thus, by Lemma $\ref{lm11}$
\begin{align*}
|\mathcal{X}^{(0,0)}|&\leqslant |D_{t-3}(\mathbf{x}_{[6,n]})|\leqslant D_3(n-7,t-3)+D_3(n-7,t-4)+D_3(n-8,t-4)+D_3(n-8,t-5)+D_3(n-10,t-6)\nonumber,\\
|\mathcal{X}^{(0,1)}|&\leqslant |D_{t-2}(\mathbf{x}_{[5,n]})|\leqslant D_3(n-6,t-2)+D_3(n-6,t-3)+D_3(n-7,t-3)+D_3(n-7,t-4)+D_3(n-9,t-5).
\end{align*}
Furthermore, we have $|\mathcal{X}^2|=|\mathcal{X}^{(2,0)}|+|\mathcal{X}^{(2,1)}|+|\mathcal{X}^{(2,2)}|$, where
\begin{align*}
|\mathcal{X}^{(2,0)}|=&|D_{t-3}(2,\mathbf{x}_{[7,n]})\cap D_{t-2}(1,2,\mathbf{y}_{[7,n]})|\leqslant |D_{t-3}(2,\mathbf{x}_{[7,n]})|\leqslant D_3(n-5,t-3)\nonumber,\\
|\mathcal{X}^{(2,1)}|=&|D_{t-2}(0,2,\mathbf{x}_{[7,n]})\cap D_{t-3}(2,\mathbf{y}_{[7,n]})|\leqslant |D_{t-3}(2,\mathbf{y}_{[7,n]})|\leqslant D_3(n-5,t-3)\nonumber,\\
|\mathcal{X}^{(2,2)}|=&|D_{t-4}(\mathbf{x}_{[7,n]})\cap D_{t-4}(\mathbf{y}_{[7,n]})|\leqslant |D_{t-4}(\mathbf{x}_{[7,n]})|\\
\leqslant& D_3(n-8,t-4)+D_3(n-8,t-5)+D_3(n-9,t-5)+D_3(n-9,t-6)+D_3(n-11,t-7).
\end{align*}
Thus, when $x_6=y_6=2$ and $\mathbf{x}_{[7,n]}\notin \mathbf{c}_3(n-6)$, we obtain
\begin{align*}
|\mathcal{X}&|\leqslant 2D_3(n-5,t-2)+D_{3}(n-6,t-3)+D_3(n-5,t-3)+D_3(n-4,t-2)+2D_3(n-5,t-2)+D_{3}(n-6,t-3)\\
&~~~+2D_3(n-5,t-3)+D_3(n-6,t-2)+D_3(n-6,t-3)+D_3(n-7,t-3)+D_3(n-7,t-4)+D_3(n-9,t-5)\\
&~~~+D_3(n-7,t-3)+D_3(n-7,t-4)+D_3(n-8,t-4)+D_3(n-8,t-5)+D_3(n-10,t-6)\\
&~~~+D_3(n-8,t-4)+D_3(n-8,t-5)+D_3(n-9,t-5)+D_3(n-9,t-6)+D_3(n-11,t-7)\\
&=M_1(n,t),
\end{align*}
from Eq. $(\ref{eq3})$ under the condition $n\geqslant t+4$. 

Similarly, when $x_6=y_6=2$ and $\mathbf{y}_{[7,n]}\notin \mathbf{c}_3(n-6)$, we also obtain $|\mathcal{X}|\leqslant M_1(n,t)$.

\textbf{Case C}: We now consider the case where $x_6=y_6=2$, and $\mathbf{x}_{[7,n]},\mathbf{y}_{[7,n]}\in \mathbf{c}_3(n-6)$.  
If $x_7=0$ or $2$, then
\begin{align}
|\mathcal{X}^{(0,1)}|&\leqslant |D_{t-2}(0,2,x_7,\mathbf{x}_{[8,n]})|\nonumber\\
&\leqslant |D_{t-2}(0,2,x_7,\mathbf{x}_{[8,n]})^0|+|D_{t-2}(0,2,x_7,\mathbf{x}_{[8,n]})^1|+|D_{t-2}(0,2,x_7,\mathbf{x}_{[8,n]})^2|\nonumber\\
&=|D_{t-2}(2,x_7,\mathbf{x}_{[8,n]})|+|D_{t-5-\ell_1}(\mathbf{x}_{[9+\ell_1,n]})|+|D_{t-3}(x_7,\mathbf{x}_{[8,n]})|\nonumber\\
&\leqslant D_3(n-5,t-2)+D_3(n-6,t-3)+D_3(n-8,t-5),\label{eq52}
\end{align}
where $x_{8+\ell_1}=1$ and $\ell_1\geqslant 0$. Thus, if $x_7=0$ or $2$, then by Eqs. $(\ref{eq51})$ and $(\ref{eq52})$ we have
$$|\mathcal{X}|\leqslant M_1(n,t).$$
Similarly, if $y_7=1$ or $2$, then we also obtain $|\mathcal{X}|\leqslant M_1(n,t).$ 

Thus, we only consider the case where $x_7=1$ and $y_7=0$. That is, $\mathbf{x}_{[1,7]}=(1,0,2,1,0,2,1)$ and $\mathbf{y}_{[1,7]}=(0,1,2,0,1,2,0)$. Furthermore, we have
$\mathbf{x}_{[7,n]}\neq \mathbf{y}_{[7,n]}$. Moreover, $|\mathcal{X}^2|=|\mathcal{X}^{(2,0)}|+|\mathcal{X}^{(2,1)}|+|\mathcal{X}^{(2,2)}|$, where
\begin{align*}
|\mathcal{X}^{(2,0)}|=&|D_{t-3}(2,\mathbf{x}_{[7,n]})\cap D_{t-2}(1,2,\mathbf{y}_{[7,n]})|\leqslant |D_{t-3}(2,\mathbf{x}_{[7,n]})|\leqslant D_3(n-5,t-3)\nonumber,\\
|\mathcal{X}^{(2,1)}|=&|D_{t-2}(0,2,\mathbf{x}_{[7,n]})\cap D_{t-3}(2,\mathbf{y}_{[7,n]})|\leqslant |D_{t-3}(2,\mathbf{y}_{[7,n]})|\leqslant D_3(n-5,t-3)\nonumber,\\
|\mathcal{X}^{(2,2)}|=&|D_{t-4}(\mathbf{x}_{[7,n]})\cap D_{t-4}(\mathbf{y}_{[7,n]})|\leqslant N_3(n-6,1,t-4)=2D_3(n-8,t-5)+D_3(n-9,t-6).
\end{align*}
Combining with Eq. $(\ref{eq51})$, we obtain
\begin{small}
\begin{align*}
|\mathcal{X}&|\leqslant D_3(n-5,t-3)+D_3(n-4,t-2)+2D_3(n-5,t-2)+D_{3}(n-6,t-3)+D_3(n-5,t-3)+D_3(n-4,t-2)\\
&~~~~+2D_3(n-5,t-2)+D_{3}(n-6,t-3)+2D_3(n-5,t-3)+2D_3(n-8,t-5)+D_3(n-9,t-6)\\
&=M_1(n,t),
\end{align*}
\end{small}
from Eq. $(\ref{eq3})$ under the condition $n\geqslant t+4$. By the above discussion, the lemma follows.
\end{proof}

\section{Proof of Lemma $\ref{lm15}$ }\label{APP-F}
The purpose of this appendix is to give the proof of Lemma $\ref{lm15}$.

\begin{proof}{ \rm (The proof of Lemma $\ref{lm15}$)}
For $k\geqslant 3$, we have
\begin{align*}
f(3k+1,2k)&=D_3(3k-4,2k-2)+D_3(3k-5,2k-2)+D_3(3k-7,2k-5)+D_3(3k-9,2k-5)-D_3(3k-6,2k-3)\\
&~~~~-2D_3(3k-6,2k-4)-D_3(3k-11,2k-7)\\
&\overset{(a)}{=}3^{k-2}+3^{k-3}+D_3(3k-7,2k-5)-3^{k-3}-2\cdot 3^{k-2}+D_3(3k-9,2k-5)-D_3(3k-11,2k-7)\\
&=D_3(3k-7,2k-5)-3^{k-2}\\
&\overset{(b)}{=}\sum\limits_{i=0}^{2k-5}\binom{k-2}{i}\sum\limits_{j=0}^{2k-5-i}\binom{i}{j}-3^{k-2},
\end{align*}
where $(a)$ follows from $D_3(m,s)=3^{m-s}$ under the condition $2m\leqslant 3s$, $(b)$ follows from $D_3(m,s)=\sum\limits_{i=0}^{s}\binom{m-s}{i}\sum\limits_{j=0}^{s-i}\binom{i}{j}$. Since $(2k-5)\geqslant k-2$, we have 
\begin{align*}
\sum\limits_{i=0}^{2k-5}\binom{k-2}{i}\sum\limits_{j=0}^{2k-5-i}\binom{i}{j}&=\sum\limits_{i=0}^{k-2}\binom{k-2}{i}\sum\limits_{j=0}^{2k-5-i}\binom{i}{j}=
\sum\limits_{i=0}^{k-3}\binom{k-2}{i}\sum\limits_{j=0}^{k-2}\binom{i}{j}+\sum\limits_{j=0}^{k-3}\binom{k-2}{j}\\
&=\sum\limits_{i=0}^{k-3}\binom{k-2}{i}\sum\limits_{j=0}^{k-2}\binom{i}{j}+\sum\limits_{j=0}^{k-2}\binom{k-2}{j}-1
=\sum\limits_{i=0}^{k-2}\binom{k-2}{i}\sum\limits_{j=0}^{k-2}\binom{i}{j}-1=3^{k-2}-1.
\end{align*}
Thus, for $k\geqslant 3$, we have $f(3k+1,2k)=-1$.

For $k\geqslant 3$, it follows that
\begin{align*}
f(3k+2,2k+1)&=D_3(3k-3,2k-1)+D_3(3k-4,2k-1)+D_3(3k-6,2k-4)+D_3(3k-8,2k-4)-D_3(3k-5,2k-2)\\
&~~~~-2D_3(3k-5,2k-3)-D_3(3k-10,2k-6)\\
&\overset{(a)}{=}3^{k-2}+3^{k-3}+3^{k-2}+3^{k-4}-3^{k-3}-2\cdot 3^{k-2}-3^{k-4}=0,
\end{align*}
where $(a)$ follows from $D_3(m,s)=3^{m-s}$ when $2m\leqslant 3s$. Thus, for $k\geqslant 3$, we conclude that $f(3k+2,2k+1)=0$.

For $t\geqslant 6$ and $n\geqslant 3t$, we have
\begin{align*}
f(n,t)&=D_3(n-5,t-2)+D_3(n-6,t-2)+D_3(n-8,t-5)+D_3(n-10,t-5)-D_3(n-7,t-3)\\
&~~~~-2D_3(n-7,t-4)-D_3(n-12,t-7)\\
&=\big(D_3(n-5,t-2)-D_3(n-7,t-4)\big)+\big(D_3(n-6,t-2)-D_3(n-7,t-3)\big)-\big(D_3(n-7,t-4)-D_3(n-8,t-5)\big)\\
&~~~~+\big(D_3(n-10,t-5)-D_3(n-12,t-7)\big)\\
&\overset{(a)}{\geqslant} \big(D_3(n-5,t-2)-D_3(n-6,t-3)\big)-\big(D_3(n-7,t-4)-D_3(n-8,t-5)\big)\\
&\overset{(b)}{=}\sum\limits_{i=0}^{t-2}\binom{n-t-3}{i}\sum\limits_{j=0}^{t-2-i}\binom{i}{j}-\sum\limits_{i=0}^{t-3}\binom{n-t-3}{i}\sum\limits_{j=0}^{t-3-i}\binom{i}{j}-
\big(\sum\limits_{i=0}^{t-4}\binom{n-t-3}{i}\sum\limits_{j=0}^{t-4-i}\binom{i}{j}-\sum\limits_{i=0}^{t-5}\binom{n-t-3}{i}\sum\limits_{j=0}^{t-5-i}\binom{i}{j}\big)\\
&=\sum\limits_{i=0}^{t-2}\binom{n-t-3}{i}\binom{i}{t-2-i}-\sum\limits_{i=0}^{t-4}\binom{n-t-3}{i}\binom{i}{t-4-i}\\
&=\sum\limits_{j=0}^{t-2}\binom{n-t-3}{t-2-j}\binom{t-2-j}{j}-\sum\limits_{j=0}^{t-4}\binom{n-t-3}{t-4-j}\binom{t-4-j}{j}\\
&\overset{(c)}{\geqslant} \sum\limits_{j=0}^{t-4}\binom{n-t-3}{t-2-j}\binom{t-2-j}{j}-\sum\limits_{j=0}^{t-4}\binom{n-t-3}{t-4-j}\binom{t-2-j}{j}\\
&=\sum\limits_{j=0}^{t-4}\big(\binom{n-t-3}{t-2-j}-\binom{n-t-3}{t-4-j}\big)\binom{t-2-j}{j}\\
&>0,
\end{align*}
where $(a)$ follows from Eq. $(\ref{eq6})$, $(b)$ follows from $D_3(m,s)=\sum\limits_{i=0}^{s}\binom{m-s}{i}\sum\limits_{j=0}^{s-i}\binom{i}{j}$, $(c)$ follows from $\binom{t-2-j}{j}\geqslant \binom{t-4-j}{j}$ for $j\in \{0,1,...,t-4\}$, and $\binom{n-t-3}{t-2-j}>\binom{n-t-3}{t-4-j}$ for $j\in \{0,1,...,t-4\}$.
\end{proof}

\section{Proofs of Lemmas $\ref{lm16}$-$\ref{lm18}$ }\label{APP-G}
The purpose of this appendix is to give the proofs of Lemmas $\ref{lm16}$-$\ref{lm18}$.

\begin{proof}{ \rm (The proof of Lemma $\ref{lm16}$)}
We now consider the case where $i=1$. Given $n\geqslant \lfloor \frac{3t}{2} \rfloor +1$, it follows that $2n\geqslant 3t+1$. If $2(n-1)\leqslant 3t$, then $2n=3t+1$ or $2n=3t+2$. 

In the first subcase where $2n=3t+1$, we parameterize this relationship by setting $t=2k+1$ and $n=3k+2$. The condition $n\geqslant 9$ implies $k\geqslant 3$. %For the initial values $k=3$ and $4$, a direct computation verifies that $M_{0}(10,7)=3^3=3^{n-t-1}$ and $M_{0}(13,9)=3^4=3^{n-t-1}$. For $k\geqslant 5$, by Eq. $(\ref{eq4})$ we have
%\begin{align*}
%M_0(n-1,t)&=D_3(3k-3,2k-1)+3D_3(3k-4,2k-1)+4D_3(3k-4,2k-2)+3D_3(3k-5,2k-2)+D_3(3k-5,2k-3)+\\
%&~~+2D_3(3k-6,2k-2)+2D_3(3k-6,2k-3)+D_3(3k-7,2k-3)+D_3(3k-11,2k-6)-D_3(3k-9,2k-4)\\
%&=3^{k-2}+3\cdot 3^{k-3}+4\cdot 3^{k-2}+3\cdot 3^{k-3}+3^{k-2}+2\cdot 3^{k-4}+2\cdot 3^{k-3}+3^{k-4}+3^{k-5}-3^{k-5}=3^k=3^{n-t-1}.
%\end{align*}
For $k\geqslant 3$, we have 
\begin{align*}
M_1(n-1,t)&=D_3(3k-3,2k-1)+5D_3(3k-4,2k-1)+4D_3(3k-4,2k-2)+3D_3(3k-5,2k-2)+D_3(3k-5,2k-3)\\
&~~~~+D_3(3k-7,2k-4)\\
&=3^{k-2}+5\cdot 3^{k-3}+4\cdot 3^{k-2}+3\cdot 3^{k-3}+3^{k-2}+3^{k-3}=3^k=3^{n-t-1}.
\end{align*}
Thus, when $2n=3t+1$ and $n\geqslant \max\{9,\lfloor \frac{3t}{2} \rfloor +1\}$, it follows that $M_1(n-1,t)=3^{n-t-1}$. 

The second subcase, $2n=3t+2$, can be handled similarly. Here, we set $t=2k$ and $n=3k+1$ for $k\geqslant 3$. An analogous calculation shows that
 $M_1(n-1,t)=3^{n-t-1}$. 

We now consider the case where $i=2$. From the condition $n\geqslant \lfloor \frac{3t}{2} \rfloor +1$, it follows that $2n\geqslant 3t+1$. If $2(n-2)\leqslant 3(t-1)$, then $2n=3t+1$. For the case $2n=3t+1$, we set $t=2k+1$ and $n=3k+2$ for some integer $k\geqslant 3$. %For the initial value $k=4$, direct computation verifies that  $M_{0}(12,8)=3^4=3^{n-t-1}$. For $k\geqslant 5$, by Eq. $(\ref{eq4})$ we have
%\begin{align*}
%M_0&(n-2,t-1)=D_3(3k-4,2k-2)+3D_3(3k-5,2k-2)+4D_3(3k-5,2k-3)+3D_3(3k-6,2k-3)+D_3(3k-6,2k-4)+\\
%&~~~+2D_3(3k-7,2k-3)+2D_3(3k-7,2k-4)+D_3(3k-8,2k-4)+D_3(3k-12,2k-7)-D_3(3k-10,2k-5)\\
%&=3^{k-2}+3\cdot 3^{k-3}+4\cdot 3^{k-2}+3\cdot 3^{k-3}+3^{k-2}+2\cdot 3^{k-4}+2\cdot 3^{k-3}+3^{k-4}+3^{k-5}-3^{k-5}=3^k=3^{n-t-1}.
%\end{align*}
For $k\geqslant 3$, we have
\begin{align*}
M_1(n-2,t-1)&=D_3(3k-4,2k-2)+5D_3(3k-5,2k-2)+4D_3(3k-5,2k-3)+3D_3(3k-6,2k-3)\\
&~~~~+D_3(3k-6,2k-4)+D_3(3k-8,2k-5)\\
&=3^{k-2}+5\cdot 3^{k-3}+4\cdot 3^{k-2}+3\cdot 3^{k-3}+3^{k-2}+3^{k-3}=3^k=3^{n-t-1}.
\end{align*}
Thus, when $2n=3t+1$ and $n\geqslant \max\{9,\lfloor \frac{3t}{2} \rfloor +1\}$, we have $M_1(n-2,t-1)=3^{n-t-1}$. 

%Finally, consider the case where $i=3$. If $2(n-3)>3(t-2)$ then by Lemma $\ref{lm12}$ we have $|D_{t-2}(\mathbf{x})\cap D_{t-2}(\mathbf{y})|\leqslant \max\{M_0(n-3,t-2),M_1(n-3,t-2)\}$. When $2(n-3)\leqslant 3(t-2)$, we have $2n\leqslant 3t$ which causes a contradiction because of $2n\geqslant 3t+1$.
%Hence, we have $|D_{t-2}(\mathbf{x})\cap D_{t-2}(\mathbf{y})|\leqslant \max\{M_0(n-3,t-2),M_1(n-3,t-2)\}$, where $|\mathbf{x}|=|\mathbf{y}|=n-3$.
\end{proof}

\begin{proof}{ \rm (The proof of Lemma $\ref{lm17}$)}
For $n\geqslant t+3$, by Eq. $(\ref{eq3})$ we have
\begin{align*}
D_3(n-2,t-2)=D_3(n-4,t-2)+2D_3(n-5,t-3)+3D_3(n-6,t-4)+2D_3(n-7,t-5)+D_3(n-8,t-6).
\end{align*}
In the following, we compare the two expressions by examining their difference, which is rearranged into pairs of the form $D_3(a,b)-D_3(c,d)$ for some integers $a,b,c,d$. By Eq. $(\ref{eq6})$, each of these pairs is non-negative, and thus all subsequent inequalities labeled $(a)$ or $(a_i)$ hold for $i\in [3]$. 

Thus, for $n\geqslant t+5$,
\begin{align*}
M_0&(n,t)-D_3(n-2,t-2)=\big(D_3(n-5,t-3)-D_3(n-6,t-4)\big)+3\big(D_3(n-5,t-2)-D_3(n-10,t-7)\big)\\
&~~~~+4\big(D_3(n-6,t-3)-D_3(n-9,t-6)\big)+2\big(D_3(n-7,t-4)-D_3(n-8,t-5)\big)+\big(2D_3(n-7,t-3)\\
&~~~~+D_3(n-8,t-4)-D_3(n-11,t-8)\big)+\big(D_3(n-12,t-7)-D_3(n-10,t-5)\big)\\
&\overset{(a_1)}{\geqslant}2D_3(n-7,t-3)+D_3(n-8,t-4)+D_3(n-12,t-7)-D_3(n-11,t-8)-D_3(n-10,t-5)\\
&\overset{(b_1)}{=}2D_3(n-7,t-3)+D_3(n-8,t-4)+D_3(n-12,t-7)-D_3(n-12,t-8)-D_3(n-13,t-9)\\
&~~~~-D_3(n-14,t-10)-D_3(n-10,t-5)\\
&=\big(D_3(n-7,t-3)-D_3(n-12,t-8)\big)+\big(D_3(n-7,t-3)-D_3(n-13,t-9)\big)\\
&~~~~+\big(D_3(n-8,t-4)-D_3(n-14,t-10)\big)+D_3(n-12,t-7)-D_3(n-10,t-5)\\
&\overset{(a_2)}{\geqslant}\big(D_3(n-7,t-3)-D_3(n-12,t-8)\big)+D_3(n-12,t-7)-D_3(n-10,t-5)\\
&\overset{(b_2)}{=}\big(D_3(n-8,t-3)+D_3(n-9,t-4)+D_3(n-10,t-5)-D_3(n-13,t-8)-D_3(n-14,t-9)\\
&~~~~-D_3(n-15,t-10)\big)+D_3(n-12,t-7)-D_3(n-10,t-5)\\
&=\big(D_3(n-8,t-3)-D_3(n-13,t-8)\big)+\big(D_3(n-9,t-4)-D_3(n-14,t-9)\big)+\big(D_3(n-12,t-7)\\
&~~~~-D_3(n-15,t-10)\big)\\
&\overset{(a_3)}{\geqslant}0,
\end{align*}
where $(a_i)$ follows from Eq. $(\ref{eq6})$ for $i\in [3]$, and $(b_i)$ follows from Eq. $(\ref{eq3})$ for $i\in [2]$. When $n=t+4$, the result also holds.

Similarly, for $n\geqslant t+3$, it follows that
\begin{align*}
M_1&(n,t)-D_3(n-2,t-2)=2\big(D_3(n-5,t-3)-D_3(n-6,t-4)\big)+2\big(D_3(n-5,t-2)-D_3(n-8,t-5)\big)\\
&~~~+3\big(D_3(n-5,t-2)-D_3(n-9,t-6)\big)+3\big(D_3(n-6,t-3)-D_3(n-10,t-7)\big)\\
&~~~+\big(D_3(n-8,t-5)-D_3(n-11,t-8)\big)\\
&\overset{(a)}{\geqslant}0,
\end{align*}
where $(a)$ follows from Eq. $(\ref{eq6})$.

Given that $n\geqslant \max\{9,\lfloor \frac{3t}{2} \rfloor +1\}$ and $t\geqslant 2$, it follows that $n\geqslant t+4$. 
For $n\geqslant t+6$, we have 
\begin{align*}
M_1&(n,t)-\big(M_0(n-1,t)+2N_3(n-3,t-1,t-2,1)\big)\\
&=M_1(n,t)-\bigg(D_3(n-5,t-2)+3D_3(n-6,t-2)+4D_3(n-6,t-3)+3D_3(n-7,t-3)+D_3(n-7,t-4)\\
&~~~+2D_3(n-8,t-3)+2D_3(n-8,t-4)+D_3(n-9,t-4)+D_3(n-13,t-7)-D_3(n-11,t-5)\\
&~~~+2\big(D_3(n-5,t-3)+2D_3(n-6,t-3)+D_3(n-6,t-4)+D_3(n-7,t-4)\big)\bigg)\\
&=\big(2D_3(n-7,t-2)+D_3(n-9,t-4)+D_3(n-11,t-7)-D_3(n-12,t-8)-D_3(n-10,t-5)\\
&~~~~-D_3(n-11,t-6)-D_3(n-12,t-7)\big)+(D_3(n-11,t-5)-D_3(n-13,t-7))\\
&=\big(D_3(n-7,t-2)-D_3(n-10,t-5)\big)+\big(D_3(n-7,t-2)-D_3(n-11,t-6)\big)+\big(D_3(n-9,t-4)-D_3(n-12,t-7)\big)\\
&~~~~+\big(D_3(n-11,t-7)-D_3(n-12,t-8)\big)+\big(D_3(n-11,t-5)-D_3(n-13,t-7)\big)\\
%&\overset{(b)}{\geqslant} \big(D_3(n-7,t-2)-D_3(n-10,t-5)\big)+\big(D_3(n-7,t-2)-D_3(n-11,t-6)\big)+\big(D_3(n-11,t-7)-D_3(n-12,t-8)\big)\\
&\overset{(a)}{\geqslant} 0,
\end{align*}
where $(a)$ follows from Eq. $(\ref{eq6})$.
%, and $(b)$ follows from $\big(D_3(n-9,t-4)-D_3(n-12,t-7)\big)\geqslant \big(D_3(n-11,t-5)-D_3(n-13,t-7)\big)$.
For the remaining boundary cases, $n=t+4$ and $n=t+5$, the inequality $M_0(n-1,t)+2N_3(n-3,t-1,t-2,1)\leqslant M_1(n,t)$ can be verified by direct computation.

Similarly, for $n\geqslant t+4$, we have 
\begin{align*}
M_1&(n,t)-(M_1(n-1,t)+2N_3(n-3,t-1,t-2,1))\\
&=M_1(n,t)-\bigg(D_3(n-5,t-2)+5D_3(n-6,t-2)+4D_3(n-6,t-3)+3D_3(n-7,t-3)+D_3(n-7,t-4)\\
&~~~~+D_3(n-9,t-5)+2\big(D_3(n-5,t-3)+2D_3(n-6,t-3)+D_3(n-6,t-4)+D_3(n-7,t-4)\big)\bigg)\\
&=\big(2D_3(n-8,t-4)+D_3(n-11,t-7)-2D_3(n-9,t-5)-D_3(n-12,t-8)\big)\\
&=2\big(D_3(n-8,t-4)-D_3(n-9,t-5)\big)+\big(D_3(n-11,t-7)-D_3(n-12,t-8)\big)\\
&\overset{(a)}{\geqslant} 0,
\end{align*}
where $(a)$ follows from Eq. $(\ref{eq6})$.
\end{proof}

\begin{proof}{ \rm (The proof of Lemma $\ref{lm18}$)} Given that $n\geqslant \max\{9,\lfloor \frac{3t}{2} \rfloor +1\}$ and $t\geqslant 2$, it follows that $n\geqslant t+4$. By the definition of $N_3(n-3,t-1,t-2,1)$ in Lemma $\ref{lm7}$, we have $N_3(n-3,t-1,t-2,1)=D_3(n-5,t-3)+2D_3(n-6,t-3)+D_3(n-6,t-4)+D_3(n-7,t-4)$. In the following, we compare the two expressions by examining their difference, which is rearranged into pairs of the form $D_3(a,b)-D_3(c,d)$ for some integers $a,b,c,d$. By Eq. $(\ref{eq6})$, each of these pairs is non-negative, and thus all subsequent inequalities labeled $(a)$ hold.

For $n\geqslant t+4$,
\begin{align*}
N_3(n-3&,t-1,t-2,1)-D_3(n-4,t-3)\\
&=D_3(n-5,t-3)+2D_3(n-6,t-3)+D_3(n-6,t-4)+D_3(n-7,t-4)-D_3(n-4,t-3)\\
&=2D_3(n-6,t-3)+D_3(n-7,t-4)-D_3(n-8,t-5)-D_3(n-9,t-6)-D_3(n-10,t-7)\\
&=\big(D_3(n-6,t-3)-D_3(n-8,t-5)\big)+\big(D_3(n-6,t-3)-D_3(n-9,t-6)\big)\\
&~~~~+\big(D_3(n-7,t-4)-D_3(n-10,t-7)\big)\\
&\overset{(a)}{\geqslant}0,
\end{align*}
where $(a)$ follows from Eq. $(\ref{eq6})$.

By the definitions of $N_3(n-4,t-2,t-3,1)$ and $M_1(n-3,t-2)$, for $n\geqslant t+4$, we have
\begin{align*}
M_1(n-&3,t-2)-N_3(n-4,t-2,t-3,1)=D_3(n-7,t-4)+5D_3(n-8,t-4)+4D_3(n-8,t-5)+3D_3(n-9,t-5)\\
&+D_3(n-9,t-6)+D_3(n-11,t-7)-\big(D_3(n-6,t-4)+2D_3(n-7,t-4)+D_3(n-7,t-5)+D_3(n-8,t-5)\big)\\
&=2\big(D_3(n-8,t-4)-D_3(n-10,t-6)\big)+\big(D_3(n-8,t-4)-D_3(n-12,t-8)\big)\\
&~~~~+\big(D_3(n-8,t-5)-D_3(n-9,t-6)\big)+\big(D_3(n-9,t-5)-D_3(n-13,t-9)\big)\\
&\overset{(a)}{\geqslant} 0,
\end{align*}
where $(a)$ follows from Eq. $(\ref{eq6})$.

For $n\geqslant t+6$, we have 
\begin{align*}
M_1&(n-3,t-2)+M_1(n-2,t-1)-M_0(n-2,t-1)-N_3(n-4,t-2,t-3,1)\\
&=3D_3(n-8,t-4)+D_3(n-8,t-5)+D_3(n-9,t-5)-\big(2D_3(n-10,t-6)+D_3(n-9,t-6)\\
&~~~~+D_3(n-12,t-8)+D_3(n-13,t-9)\big)+\big(M_1(n-2,t-1)-M_0(n-2,t-1)\big)\\
&=\big(D_3(n-8,t-4)-D_3(n-10,t-6)\big)+\big(D_3(n-8,t-4)-D_3(n-12,t-8)\big)+\big(D_3(n-8,t-4)-D_3(n-13,t-9)\big)\\
&~~~~+\big(D_3(n-8,t-5)-D_3(n-9,t-6)\big)+\big(D_3(n-7,t-3)-D_3(n-9,t-5)\big)+\big(D_3(n-8,t-3)-D_3(n-9,t-4)\big)\\
&~~~~+\big(D_3(n-12,t-6)-D_3(n-14,t-8)\big)\\
&\overset{(a)}{\geqslant} 0,
\end{align*}
where $(a)$ follows from Eq. $(\ref{eq6})$. 
For the boundary cases $n=t+4$ and $t+5$, direct substitution and computation confirm that the expression is also non-negative.

For $n\geqslant t+7$, we have
\begin{align*}
M_1&(n-3,t-2)+M_1(n-2,t-1)-M_0(n-3,t-1)-M_0(n-3,t-2)-2N_3(n-5,t-2,t-3,1)\\
&=\big(M_1(n-3,t-2)-M_0(n-3,t-2)\big)+\big(M_1(n-2,t-1)-M_0(n-3,t-1)-2N_3(n-5,t-2,t-3,1)\big)\\
&=\big(D_3(n-8,t-4)+D_3(n-9,t-4)+D_3(n-11,t-7)-D_3(n-10,t-5)-2D_3(n-10,t-6)\big)+\big(D_3(n-6,t-3)\\
&~~~~+4D_3(n-7,t-3)+2D_3(n-7,t-4)+D_3(n-10,t-6)-3D_3(n-8,t-3)-5D_3(n-8,t-4)\\
&~~~~-D_3(n-8,t-5)-3D_3(n-9,t-4)-3D_3(n-9,t-5)-2D_3(n-10,t-4)-2D_3(n-10,t-5)\\
&~~~~-D_3(n-11,t-5)\big)+\big(D_3(n-13,t-6)-D_3(n-15,t-8)\big)+\big(D_3(n-13,t-7)-D_3(n-15,t-9)\big)\\
&=\big(D_3(n-8,t-3)-D_3(n-11,t-6)\big)+\big(D_3(n-9,t-4)-D_3(n-12,t-7)\big)+\big(D_3(n-9,t-3)\\
&~~~~-D_3(n-10,t-4)\big)+\big(D_3(n-10,t-4)-D_3(n-13,t-7)\big)+\big(D_3(n-11,t-5)-D_3(n-14,t-8)\big)\\
&~~~~+\big(D_3(n-13,t-6)-D_3(n-15,t-8)\big)+\big(D_3(n-13,t-7)-D_3(n-15,t-9)\big)\\
&\overset{(a)}{\geqslant}0,
\end{align*}
where $(a)$ follows from Eq. $(\ref{eq6})$.
For the cases where $n=t+4, t+5$ or $t+6$, the result can be verified by direct computation.

For $n\geqslant t+6$, we have
\begin{align*}
M_1&(n-3,t-2)+M_1(n-2,t-1)-M_1(n-3,t-1)-M_0(n-3,t-2)-2N_3(n-5,t-2,t-3,1)\\
&=\big(M_1(n-3,t-2)-M_0(n-3,t-2)\big)+\big(M_1(n-2,t-1)-M_1(n-3,t-1)-2N_3(n-5,t-2,t-3,1)\big)\\
&=\big(D_3(n-8,t-4)+D_3(n-9,t-4)+D_3(n-11,t-7)-D_3(n-10,t-5)-2D_3(n-10,t-6)+D_3(n-13,t-7)\\
&~~~~-D_3(n-15,t-9)\big)+\big(D_3(n-6,t-3)+4D_3(n-7,t-3)+2D_3(n-7,t-4)+D_3(n-10,t-6)\\
&~~~~-5D_3(n-8,t-3)-5D_3(n-8,t-4)-D_3(n-8,t-5)-3D_3(n-9,t-4)-3D_3(n-9,t-5)-D_3(n-11,t-6)\big)\\
&=2\big(D_3(n-9,t-4)-D_3(n-11,t-6)\big)+\big(D_3(n-10,t-5)-D_3(n-11,t-6)\big)+\big(D_3(n-10,t-5)\\
&~~~~-D_3(n-12,t-7)\big)+\big(D_3(n-13,t-7)-D_3(n-15,t-9)\big)\\
&\overset{(a)}{\geqslant}0,
\end{align*}
where $(a)$ follows from Eq. $(\ref{eq6})$.
As with the previous cases, the result for $n=t+4$ and $t+5$
is established by direct verification. 
\end{proof}

\section*{Acknowledgments}
%\textcolor{red}{The authors would like to express their sincere gratefulness to the editor and the three anonymous reviewers for their valuable suggestions and comments which have greatly improved this paper.}
The work of X.~Wang is supported by the National Natural Science Foundation of China (Grant No. 12001134). The work of F.-W.~Fu is supported by the National Key Research and Development Program of China (Grant No. 2022YFA1005000), the National Natural Science Foundation of China (Grant Nos. 12141108, 62371259, 12226336),  the Fundamental Research Funds for the Central Universities of China (Nankai University), and the Nankai Zhide Foundation.

\end{document}